\documentclass[conference]{IEEEtran}
\ifCLASSINFOpdf
\else
\fi
\usepackage[utf8]{inputenc}
\usepackage[utf8]{inputenc} 
\usepackage[T1]{fontenc}
\usepackage{url}
\usepackage{ifthen}
\usepackage{cite}
\usepackage[cmex10]{amsmath} 

\usepackage{cite}
\usepackage{amsmath,amssymb,amsfonts}
\usepackage{mathtools}
\usepackage{amsthm}
\usepackage{algorithmic}
\usepackage{graphicx}
\usepackage{textcomp}
\usepackage{tikz}
\usepackage{caption}
\usepackage{cuted}
\usepackage{romannum}
\usepackage[utf8]{inputenc}
\usepackage{pgfplots} 
\usepackage{pgfgantt}
\usepackage{pdflscape}
\usepackage{amssymb}
\usepackage{comment}
\usepackage{pst-plot}
\usetikzlibrary{spy}
\usetikzlibrary{positioning,calc}
\usetikzlibrary{decorations.pathmorphing,calc,shapes,shapes.geometric,patterns}
\usetikzlibrary{shapes.multipart}
\usepackage{xfrac}
\usepackage{colortbl}
\usepackage{cancel}
\usetikzlibrary{arrows,positioning,calc,intersections}
\usetikzlibrary{datavisualization.formats.functions}
\def\BibTeX{{\rm B\kern-.05em{\sc i\kern-.025em b}\kern-.08em
    T\kern-.1667em\lower.7ex\hbox{E}\kern-.125emX}}
    
\usepackage{pgfplots}
\usepgfplotslibrary{fillbetween}
\usetikzlibrary{arrows, decorations.markings}
\newtheorem{theorem}{Theorem}
\newtheorem{claim}{Claim}
\newtheorem{claimproof}{Proof of Claim}
\newtheorem{lemma}{Lemma}

\newtheorem{definition}{Definition}

\newcommand{\seta}{\ensuremath{\mathcal{A}}}

\newcommand{\setd}{\ensuremath{\mathcal{D}}}

\newcommand{\bs}[1]{\boldsymbol{#1}}

\newcommand{\mc}{\mathcal}
\newcommand{\infsub}{\mathrm{inf}}
\definecolor{calpolypomonagreen}{rgb}{0.12, 0.3, 0.17}
\newcounter{remarkcount}
\newenvironment{remark}{\refstepcounter{remarkcount}\begin{trivlist}\item \textbf{Remark \theremarkcount.}}{\end{trivlist}}
\newcommand{\circlearrow}{}
\DeclareRobustCommand{\circlearrow}{%
  \mathrel{\vphantom{\rightarrow}\mathpalette\circle@arrow\relax}%
}
\newcommand{\circle@arrow}[2]{%
  \m@th
  \ooalign{%
    \hidewidth$#1\circ\mkern1mu$\hidewidth\cr
    $#1-$\cr}%
}
\makeatother
\let\emptyset\varnothing
\usepackage{amsmath, amssymb, amsfonts, amsthm}
\usepackage{bm}
\usepackage{xcolor}
\usepackage{framed}

\usepackage{pgf}
\usepackage{tikz}
\usetikzlibrary{arrows,automata}
\usetikzlibrary{positioning}
\usetikzlibrary{decorations.pathmorphing}
\usetikzlibrary{shapes.geometric}
\usetikzlibrary{fit}
\usetikzlibrary{backgrounds}

\usepackage[caption=false,font=footnotesize]{subfig}

\newcommand{\mbf}{\mathbf}
\newcommand{\mbb}{\mathbb}

\theoremstyle{definition}

\theoremstyle{remark}

\usepackage{hyperref}
\def\BibTeX{{\rm B\kern-.05em{\sc i\kern-.025em b}\kern-.08em
    T\kern-.1667em\lower.7ex\hbox{E}\kern-.125emX}}

\begin{document}
\onecolumn
%
\title{Message Transmission and Common Randomness Generation over MIMO Slow Fading Channels with Arbitrary Channel State Distribution}

\author{
\IEEEauthorblockN{Rami Ezzine\IEEEauthorrefmark{1}, Moritz Wiese\IEEEauthorrefmark{1}\IEEEauthorrefmark{4}, Christian Deppe\IEEEauthorrefmark{2}\IEEEauthorrefmark{4} and Holger Boche\IEEEauthorrefmark{1}\IEEEauthorrefmark{3}\IEEEauthorrefmark{4}}
\IEEEauthorblockA{\IEEEauthorrefmark{1}Technical University of Munich, Chair of Theoretical Information Technology, Munich, Germany\\
\IEEEauthorrefmark{2}Technical University of Munich, Institute for Communications Engineering,  Munich, Germany\\
\IEEEauthorrefmark{3}CASA -- Cyber Security in the Age of Large-Scale Adversaries–
Exzellenzcluster, Ruhr-Universit\"at Bochum, Germany\\
\IEEEauthorrefmark{4}BMBF Research Hub 6G-life, Munich, Germany\\
Email: \{rami.ezzine, wiese, christian.deppe, boche\}@tum.de}
}


%


\maketitle
\thispagestyle{plain}
\pagenumbering{arabic}
\pagestyle{plain}
\begin{abstract}
We investigate the problem of message transmission and the problem of common randomness (CR) generation over single-user multiple-input multiple-output (MIMO) slow fading channels with average input power constraint, additive white Gaussian noise (AWGN), arbitrary state distribution and with complete channel state information available at the receiver side (CSIR). First, we derive a  lower and an upper bound on  the outage transmission capacity of MIMO slow fading channels for arbitrary state distribution and show that the bounds coincide except possibly at the points of  discontinuity of the outage transmission capacity, of which there are, at most, countably many. To prove the lower bound on the outage transmission capacity, we also establish the capacity of a specific compound MIMO Gaussian channel.
Second, we define the outage CR capacity for a two-source model with unidirectional communication over a  MIMO slow fading channel with arbitrary state distribution and establish a lower  and an upper bound on it using our bounds on the outage transmission capacity of the MIMO slow fading channel. 
\end{abstract}


\begin{IEEEkeywords}
Common randomness, outage transmission capacity, MIMO slow fading channels, compound MIMO Gaussian channels
\end{IEEEkeywords}

%
\IEEEpeerreviewmaketitle

\section{Introduction}
Motivated by its striking applications in the theory of identification, Ahlswede and Csiszár introduced the concept of generation of non-secret common randomness (CR) in \cite{part2}.
The identification scheme is an approach in communications
developed by Ahlswede and Dueck \cite{identification} in 1989. In the identification framework, the decoder is not interested in knowing what the received message is. He rather wants to know if a specific message of special interest to him has been sent or not. Naturally, the sender has no knowledge of that specific message, otherwise, the problem would be trivial. It turns out that CR may allow a significant increase in the identification capacity of channels\cite{Generaltheory,part2,CRincrease}. 
 While the number of identification messages (also called identities) increases exponentially with the block-length in the deterministic identification scheme for discrete memoryless channels (DMCs), the size of the identification code increases doubly exponentially with the block-length when CR is used as a resource.
The identification scheme is more suitable than the classical transmission scheme proposed by Shannon \cite{shannon} in many practical applications which require robust and ultra-reliable low latency information exchange including several machine-to-machine
and human-to-machine systems \cite{applications}, industry 4.0 \cite{industrie40} and 6G communication systems \cite{6Gcomm}. It is therefore expected that CR will be an important resource for future communication systems \cite{6Gcomm}\cite{6Gpostshannon} and, in particular, that resilience requirements \cite{6Gcomm} and security requirements \cite{semanticsecurity} can also be met on the basis of CR. These requirements are again of particular importance for achieving trustworthiness, which is a key challenge for future communication systems due to modern applications \cite{6Gandtrustworthiness}. For this reason,  CR generation for future communication networks is a central research question in large 6G research projects \cite{researchgroup1}\cite{researchgroup2}.

The applications of CR generation are not restricted to the identification scheme. The availability of CR as a resource plays in general a key role in distributed  settings\cite{survey}. It allows to design correlated random protocols that often perform faster and more efficiently than the deterministic ones or the ones using independent
randomization. 
Further examples of the applications of CR include correlated random coding over arbitrarily varying channels (AVCs) \cite{capacityAVC} and oblivious transfer and bit commitment schemes \cite{commitmentcapacity}\cite{unconditionallysecure}. 
CR is also of high relevance in the key generation problem. Indeed, under additional secrecy constraints, the generated CR can be used as secret keys, as shown in the fundamental two papers \cite{part1}\cite{Maurer}. The generated secret keys can be used to perform cryptographic tasks including secure message transmission and message authentication. In our work, however, we will not impose any secrecy requirements.

We study the problem of CR generation in the basic two-party communication setting in which
Alice and Bob aim to agree on a common random variable with high probability 
by observing independent and identically distributed (i.i.d.) samples of correlated discrete sources and while communicating as little as possible. Ahlswede and Csizár initially introduced the problem of CR generation from discrete correlated sources where the communication was over discrete noiseless channels with limited capacity\cite{part2}. A single-letter characterization of the CR capacity for this model was established in \cite{part2}. CR capacity refers to the maximum rate of CR that Alice and Bob can generate using the resources available in the model.  The results on CR capacity were later extended to single-input single-output (SISO) and multiple-antenna Gaussian channels in \cite{CRgaussian}   for  their  practical  relevance  in many communication situations such as  wired  and  wireless communications,  satellite  and  deep  space  communication  links,  etc.
 The results on CR capacity over
Gaussian channels have been used to establish a lower-bound
on the corresponding correlation-assisted secure identification
capacity in the log-log scale in \cite{CRgaussian}. This lower bound can
already exceed the secure identification capacity over Gaussian
channels with randomized encoding established in \cite{wafapaper}.

In our work, we consider the CR generation problem over single-user multiple-input multiple-output (MIMO) slow fading channels with complete channel state information available at the receiver side (CSIR). The focus is on the MIMO setting since multiple-antenna systems present considerable practical benefits including increased capacity, reliability and spectrum
efficiency. This is due to a combination of both diversity and spatial multiplexing gains \cite{Tse}. In particular, a practically relevant model in wireless communications is the slow fading model with  additive white Gaussian noise (AWGN)\cite{Tse,goldsmith,inftheoretic,infaspects}.
 In the multiple-antenna slow fading scenario, the channel state, represented by the channel matrix, is random but remains constant during the codeword transmission. Therefore, channel fades cannot be averaged out and ensuring reliable communication is consequently challenging.
 
A commonly used concept to assess the performance in slow fading environments is the $\eta$-outage transmission capacity defined to be the supremum of all rates for which the outage probability is lower than or equal to $\eta$\cite{Tse}\cite{goldsmith}. From the channel transmission perspective and for a given coding scheme, outage occurs when the  instantaneous channel state is so poor  that that coding scheme 
is not able to establish reliable communication over the channel. The capacity versus outage approach was initially proposed in \cite{inftheoretic} for fading channels. Later, this approach was applied to multi-antenna channels in \cite{telatar}, where the analysis was restricted to MIMO Rayleigh fading channels. However, to the best of our knowledge, no rigorous proof of the outage transmission capacity of SISO and MIMO slow fading channels with arbitrary state distribution is provided in the literature. For instance, the capacity formula provided in the literature for the SISO case is not valid
when the distribution function of the absolute value of the state is discontinuous. 

The first contribution of this paper lies in deriving a lower and an upper bound on the $\eta$-outage transmission capacity of MIMO slow fading channels with average input power constraint, AWGN and arbitrary state distribution. We will show that the bounds coincide except possibly at the points of  discontinuity of the outage transmission capacity, of which there are, at most, countably many and we will show that when the state has a density, which is positive except on a set with Lebesgue measure equal to zero, then the bounds on the $\eta$-outage capacity coincide for all possible values of $\eta,$ regardless of whether the outage capacity is continuous at $\eta$ or not. To prove the lower bound on the outage transmission capacity, we will also establish the capacity of a compound MIMO Gaussian channel corresponding to the set of MIMO Gaussian channels that are not in outage for some fixed input covariance matrix and target rate, and for which the operator norm of the state is upper-bounded by some positive constant.
We will additionally establish the $\eta$-outage transmission capacity of single-input multiple-output (SIMO) slow fading channels and provide an alternative proof of the outage transmission capacity for the SISO case based on the degradedness of SISO Gaussian channels as well as the strong converse for this type of channels. It is here worth-mentioning that the $\eta$-outage capacity formula that we prove for the SISO and the SIMO case hold regardless whether the $\eta$-outage capacity is continuous at $\eta$ or not.
The outage transmission capacity formula that we prove for the SISO case is an extension of the formula presented in the literature to arbitrary state distribution.

 The second contribution of this paper lies in introducing the concept of outage in the CR generation framework as well as deriving a lower and an upper bound on the $\eta$-outage CR  capacity for a two-source model with one-way communication over MIMO slow fading channels with AWGN and arbitrary state distribution. In the CR generation framework, outage occurs when the channel state is so poor that Alice and Bob cannot agree on a common random variable with high probability. The $\eta$-outage CR capacity is defined to be the maximum of all achievable CR rates for which the outage probability from the CR generation perspective does not exceed $\eta.$ In the proof of the bounds on the $\eta$-outage CR capacity, we will use our bounds on the $\eta$-outage transmission capacity of MIMO slow fading channels.

\quad \textit{ Paper Outline:} Section \ref{sec2} describes the system model and provides the key definitions as well as the main and auxiliary results. In Section \ref{proofoutagecapacity}, we derive a lower and an upper bound on the $\eta$-outage transmission capacity of MIMO slow fading channels with average input power constraint, AWGN and with arbitrary state distribution. In Section \ref{SIMOSISO}, we establish the $\eta$-outage transmission capacity for the SIMO case and provide an alternative proof of it for the SISO case.  Section \ref{proofoutagecrcapacity} is devoted to the derivation of a lower  and an upper bound on the $\eta$-outage CR capacity for a two-source model with unidirectional communication  over MIMO slow fading channels. In Section \ref{proofcompoundcapacity}, we establish the capacity of a specific compound MIMO complex Gaussian channel.  Section \ref{conclusion} contains concluding remarks and proposes
potential future research in this field.

\quad \textit{Notation:}  
$\mathbb{C}$ denotes the set of complex numbers and $\mbb R$ denotes the set of real numbers; $H(\cdot)$ and $h(\cdot)$  correspond to the entropy and the differential entropy function, respectively; $I(\cdot;\cdot)$ denotes the mutual information between two random variables. All information
quantities are taken to base 2. Throughout the paper,  $\log$ is taken to  base 2.  The natural exponential and the natural logarithm are denoted by $\exp$ and $\ln$, respectively.  For any random variables $X$, $Y$ and $Z$, we use the notation $\color{black}X \circlearrow{Y} \circlearrow{Z}\color{black}$ to indicate a Markov chain.
 $\mathcal{T}_{U}^{n}$ denotes the set of $U$-typical sequences of block-length $n$ and of type $P_{U}$. For any matrix $\mbf A,$ $\text{tr}(\mbf A)$ refers to the trace of $\mbf A,$ $\lVert \mbf A\rVert$ stands for the operator norm of $\mbf A$ with respect to the Euclidean norm,$\mbf A^{H}$ stands for the standard Hermitian transpose of $\mbf A$ and $\mbf A^{-1}$ refers to the matrix inverse of $\mbf A.$ For any random matrix $\mbf A \in \mathbb{C}^{m \times n}$ with entries $\mbf A_{i,j}$ $i=1,\hdots,m, j=1,\hdots, n,$ we define
\[
\mbb E \left[\mbf A\right] = \begin{bmatrix} 
    \mbb E\left[\mbf A_{11}\right] &  \mbb E\left[\mbf A_{12}\right] & \dots \\
    \vdots & \ddots & \\
    \mbb E\left[\mbf A_{m1}\right] &        &  \mbb E\left[\mbf A_{mn}\right]
    \end{bmatrix}.
\]
 For any random vector $\bs{X},$ $\text{cov}(\bs{X})$ refers to its covariance matrix. For any set $\mathcal{E}$, $\mathcal{E}^c$ refers to its  complement and $|\mathcal{E}|$ refers to its cardinality.

\section{System Model, Definitions and  Results}
\label{sec2}
\subsection{System Model}
\label{systemmodel}
Let a MIMO slow fading channel $W_{\mbf G}$ be given.
First, we define the MIMO slow fading channel $W_{\mbf G}.$ Suppose that one terminal called Terminal $A$
 wants to transmit a message to another terminal called  Terminal $B$ by sending, for arbitrary $n>0,$ an input sequence $\bs{t}^n=(\bs{t}_1,\hdots,\bs{t}_n)\in\mbb C^{N_{T}\times n}$ of block-length $n$ over the MIMO slow fading channel. Terminal $B$ observes the output sequence 
$\bs{z}^n=(\bs{z}_1,\hdots,\bs{z}_n)\in \mbb C^{N_{R}\times n}$ of block-length $n$ such that
\begin{align}
\bs{z}_{i}=\mbf G\bs{t}_{i}+\bs{\xi}_{i} \quad i=1, \hdots,n.
\label{MIMOchannelmodel} \nonumber
\end{align} 
Here, $N_T$ and $N_R$ refer to the number of transmit and receive antennas, respectively.
$\mbf G \in \mbb C^{N_{R}\times N_{T}}$ models the random complex gain, where we assume that both terminals $A$ and $B$ know  the distribution of the gain $\mbf G$ and that the actual realization of the gain is known  by Terminal $B$ only. \color{black} \
 $\bs{\xi}^n=(\bs{\xi}_1,\hdots,\bs{\xi}_n)\in \mathbb{C}^{N_{R}\times n}$ models the noise sequence.
We assume that the $\bs{\xi}_i$s are i.i.d. such that $\bs{\xi}_{i} \sim \mathcal{N}_{\mathbb{C}}\left(\bs{0}_{N_R},\sigma^2 \mbf I_{N_{R}}\right), i=1, \hdots,n.$ We further assume that $\mbf G$ and $\bs{\xi}^n$ are mutually independent.

We are interested in the problem of \textit{common randomness (CR) generation} over $W_{\mbf G}.$ 
Let a discrete memoryless multiple source (DMMS) $P_{XY}$ with two components, with  generic variables $X$ and $Y$ on alphabets $\mathcal{X}$ and $\mathcal{Y}$, respectively, be given. The DMMS emits i.i.d. samples of $(X,Y).$
Suppose that the outputs of $X$ are observed only by Terminal $A$ and those of $Y$ only by Terminal $B.$ We further assume that the joint distribution of $(X,Y)$ is known to both terminals.
Terminal $A$
can communicate with Terminal $B$ over the MIMO slow fading channel $W_{\mbf G}.$
We also assume that $(X^n,Y^n)$ is independent of $(\mbf G,\bs{\xi}^n).$  There are no other resources available to any of the terminals. 
\begin{definition}
A CR-generation protocol of block-length $n$ consists of:
\begin{enumerate}
    \item A function $\Phi$ that maps $X^n$ into a random variable $K$ with alphabet $\mathcal{K}$ generated by Terminal $A.$
    \item A function $\Lambda$ that maps $X^n$ into the input sequence $\bs{T}^n \in \mbb C^{N_T\times n}$ satisfying the power constraint
    \begin{equation}
\frac{1}{n}\sum_{i=1}^{n}\bs{T}_{i}^H\bs{T}_{i}\leq P, \quad \text{almost surely.}   \ \
\label{energyconstraintMIMOCorrelated}
\end{equation}
    \item A function $\Psi$ that maps $Y^n$ and the  output sequence $\bs{Z}^n=(\bs{Z}_1,\hdots,\bs{Z}_n) \in \mbb C^{N_R\times n}$ into a random variable $L$ with alphabet $\mathcal{K}$ generated by Terminal $B.$
\end{enumerate}
Such a protocol induces a pair of random variables $(K,L)$ whose joint distribution is determined by $P_{XY}$ and by the channel $W_\mbf G$. Such a pair of random variables $(K,L)$ is called permissible.
This is illustrated in Fig. \ref{correlatedMIMO}.
\end{definition}
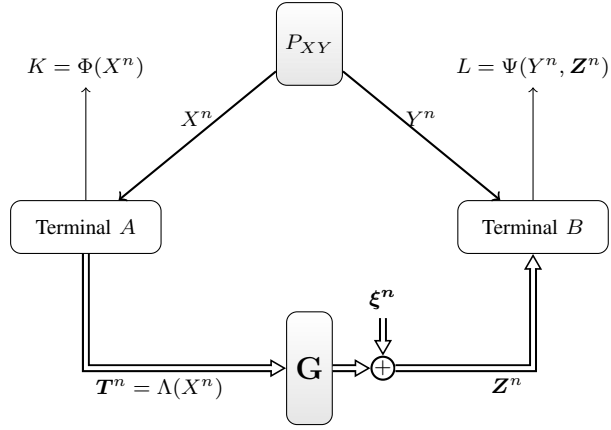
\begin{figure}
\centering
\tikzstyle{block} = [draw, rectangle, rounded corners,
minimum height=2em, minimum width=2cm]
\tikzstyle{blocksource} = [draw, top color=white, bottom color=white!80!gray, rectangle, rounded corners,
minimum height=1.1cm, minimum width=.31cm]
\tikzstyle{blockchannel} = [draw, top color=white, bottom color=white!80!gray, rectangle, rounded corners,
minimum height=1.5cm, minimum width=.35cm]
\tikzstyle{input} = [coordinate]
\tikzstyle{vectorarrow} = [thick, decoration={markings,mark=at position
   1 with {\arrow[semithick]{open triangle 60}}},
   double distance=1.4pt, shorten >= 5.5pt,
   preaction = {decorate},
   postaction = {draw,line width=1.4pt, white,shorten >= 4.5pt}]
   \tikzstyle{sum} = [draw, circle,inner sep=0pt, minimum size=2mm,  thick]
\usetikzlibrary{arrows}
\begin{tikzpicture}[scale= 1,font=\footnotesize]
\node[blocksource] (source) {$P_{XY}$};
\node[blockchannel, below=3cm of source](channel) {\large$\mathbf{G}$};
\node[sum, right=.5cm of channel] (sum) {$+$};

\node[block, below left=2.15cm of source] (x) {Terminal $A$};
\node[block, below right=2.15cm of source] (y) {Terminal $B$};
\node[above=.5cm of sum] (noise) {$\bs{\xi^n}$};
\node[above=1.5cm of x] (k) {$K=\Phi(X^n)$};
\node[above=1.5cm of y] (l) {$L=\Psi(Y^n,\bs{Z}^n)$};

\draw[->,thick] (source) -- node[above] {$X^n$} (x);
\draw[->, thick] (source) -- node[above] {$Y^n$} (y);
\draw [vectorarrow] (x) |- node[below right] {$\bs{T}^n=\Lambda(X^n)$} (channel);
\draw [vectorarrow] (channel) -- (sum);
\draw[vectorarrow] (noise) -- (sum);
\draw[vectorarrow] (sum) -| node[below left] {$\bs{Z}^n$} (y);
\draw[->] (x) -- (k);
\draw[->] (y) -- (l);

\end{tikzpicture}
\caption{Standard two-source model \cite{part2} with unidirectional communication over a  MIMO slow fading channel}
\label{correlatedMIMO}
\end{figure}
We define first an achievable $\eta$-outage CR rate and the $\eta$-outage CR capacity for the model presented above. This is an extension of the definition of an achievable CR rate and of the CR capacity over rate-limited discrete noiseless channels introduced in \cite{part2}.
\begin{definition} 
\label{etaoutagecrrate}
Fix a non-negative constant $\eta<1.$ A number $H$ is called an achievable $\eta$-outage CR rate  if there exists a non-negative constant $c$ such that for every $\alpha>0$ and $\delta>0$ and for sufficiently large $n$ there exists a permissible pair of random variables $(K,L)$ such that
\begin{equation}
    \mbb P\left[\mbb P\left[K\neq L|\mbf G\right]\leq \alpha \right]\geq 1-\eta, 
    \label{errorMIMOcorrelated}
\end{equation}
\begin{equation}
    |\mathcal{K}|\leq 2^{cn},
    \label{cardinalityMIMOcorrelated}
\end{equation}
\begin{equation}
    \frac{1}{n}H(K)> H-\delta,
     \label{rateMIMOcorrelated}
\end{equation}
where the constant $0 \leq \eta <1$ and the constant $\alpha>0$ in \eqref{errorMIMOcorrelated} correspond to an upper-bound on the outage probability and to an upper-bound on the error probability, from the common randomness generation perspective, respectively, and where the outer probability in \eqref{errorMIMOcorrelated} is with respect to $\mbf G.$
\end{definition}

\begin{remark}
\label{fastgleichentropyrate}
Together with \eqref{errorMIMOcorrelated}, the technical condition  \eqref{cardinalityMIMOcorrelated} ensures for every $\epsilon>0$ and sufficiently large block-length $n$ that $\mbb P\left[\mbf G \in \seta^{(n,\epsilon)}\right]\geq 1-\eta,$
where
$$\seta^{(n,\epsilon)} = \bigg\{ \mbf g \in \mbb C^{N_{R}\times N_{T}}: \bigg| \frac{H(K|\mbf G=\mbf g)}{n}-\frac{H(L|\mbf G=\mbf g)}{n} \bigg| \leq \epsilon \bigg\}.$$
 This follows from the analogous statement in\cite{part2}.
\end{remark}
\begin{remark}
\label{remarkUCR}
The most convenient form of CR is \textit{uniform} CR, i.e., $K$ and $L$ are uniform (or nearly uniform) random variables \cite{part2}. In Section \ref{prooflowerboundcrcapacity}, we will provide a scheme for generation of nearly uniform random variables that coincide with high probability when the system is not in outage from the CR generation perspective.
\end{remark}
\begin{definition} 
The $\eta$-outage CR capacity $C_{\eta,CR}^{X,Y}(P,W_\mbf G)$ is the maximum achievable $\eta$-outage CR rate defined according to Definition \ref{etaoutagecrrate}.
\end{definition}
Next, we define an achievable $\eta$-outage transmission rate for the MIMO slow fading channel $W_{\mbf G}$ and the corresponding $\eta$-outage transmission capacity.
For this purpose, we begin by providing the definition of a  transmission-code for  $W_{\mbf G}.$
\begin{definition}
\label{defcode}
A transmission-code $\Gamma$ of block-length $n$ and size \footnote{\text{This is the same notation used in} \cite{codingtheorems}.} $\lVert \Gamma \rVert $ and with average power constraint $P$ for the MIMO channel $W_{\mbf G}$ is a family of pairs of codewords and decoding regions $\left\{(\mbf{t}_\ell,\setd_\ell^{(\mbf g)}): \mbf g \in \mbb C^{N_R\times N_T}, \ \ell=1,\ldots,\lVert \Gamma \rVert \right\}$ such that for all $\ell,j \in \{1,\ldots,\lVert \Gamma \rVert\}$ and all $\mbf g \in \mbb C^{N_R \times N_T}:$ 
\begin{align}
& \mbf{t}_\ell \in \mbb C^{N_{T}\times n},\quad \setd_\ell^{(\mbf g)} \subset \mbb C^{N_{R}\times n}, \nonumber \\
&\frac{1}{n}\sum_{i=1}^{n}\bs{t}_{\ell,i}^H\bs{t}_{\ell,i}\leq P 
 \ \ \mbf{t}_\ell=(\bs{t}_{\ell,1},\hdots,\bs{t}_{\ell,n}), \nonumber \\
&\setd_\ell^{(\mbf g)}  \cap \setd_j^{(\mbf g)} = \emptyset,\quad \ell \neq j. \nonumber
\end{align}The maximum error probability for gain $\mbf g$ is expressed as 
\begin{align}
    e(\Gamma,\mbf g)=\underset{\ell \in \{1,\ldots,\lVert \Gamma \rVert\}}{\max}W_{\mbf g}({\setd_\ell^{(\mbf g)_c}}|\mbf{t}_\ell). \nonumber
\end{align}
\end{definition}
\begin{remark}
Since we do not assume any channel state information at the transmitter side, the codewords $\mbf{t}_\ell, \ \ell=1,\hdots,\lVert \Gamma \rVert,$ do not depend on the gain.
\end{remark}
\begin{remark}
Throughout the paper, we consider the maximum error probability criterion. However, due to the converse, the rate and capacity expressions hold also for the average error probability criterion.
\end{remark}
\begin{definition}
\label{deftransmissionrate}
    Let $0 \leq \eta<1$. A real number $R$ is called an \textit{achievable} $\eta$-\textit{outage transmission rate} of the channel $W_{\mbf G}$ if for every $\theta,\delta>0$ there exists a code sequence $(\Gamma_n)_{n=1}^\infty$, where each code $\Gamma_n$ of block-length $n$ is defined according to Definition \ref{defcode},  such that
    \[
        \frac{\log\lVert \Gamma_n\rVert}{n}\geq R-\delta
    \]
    and
    \begin{align}
        \mbb P[e(\Gamma_n,\mbf G)\leq\theta]\geq 1-\eta 
        \label{errorouterprob}
    \end{align}
    for sufficiently large $n,$ where the probability in \eqref{errorouterprob} is with respect to $\mbf G.$
\end{definition}
\begin{definition}
The $\eta$-\textit{outage transmission capacity} of the channel $W_{\mbf G}$ is the supremum of all achievable $\eta$-outage transmission rates defined according to Definition \ref{deftransmissionrate} and it is denoted by $C_\eta(P,W_\mbf G)$.
\end{definition}
\subsection{Main Results}
\begin{theorem}
   \label{cetathmMIMO}
Let $\mc Q_{P}$  be the set of complex positive semi-definite Hermitian $N_T\times N_T$ matrices  whose trace is smaller than or equal to $P.$
 For any $\mbf g \in \mbb C^{N_R\times N_T}$and any $\mbf Q \in \mc Q_{P},$ we define
 \begin{align}
     f(\mbf g,\mbf Q)=\log\det(\mathbf{I}_{N_{R}}+\frac{1}{\sigma^2}\mathbf{g}\mathbf{Q}\mathbf{g}^{H}).
     \label{fgQ}
 \end{align}
   Let $\mbf G \in \mbb C^{N_{R}\times N_{T}}$ be any random matrix. Then, the $\eta$-outage transmission capacity of the channel $W_{\mbf G}$ satisfies
 \begin{align}
C_\eta(P,W_\mbf G)\geq l(\eta)
\label{lowerboundoutagecapacity}
\end{align}
and
\begin{align}
    C_\eta(P,W_\mbf G)\leq u(\eta),
\label{upperboundoutagecapacity}
\end{align}
where
\begin{align}
   l(\eta)=\sup \ \Big\{R: \underset{\mathbf{Q}\in\mc Q_{P}}{\inf }\mbb P\left[f(\mbf G,\mbf Q)<R \right] < \eta\Big\}
\label{Retasupell}
\end{align}
and
\begin{align}
    u(\eta)= \sup \ \Big\{R: \underset{\mathbf{Q}\in\mc Q_{P}}{\inf }\mbb P\left[f(\mbf G,\mbf Q)<R \right] \leq \eta\Big\}. 
    \label{Retasupu}
\end{align}
The lower and upper bound in \eqref{lowerboundoutagecapacity} and \eqref{upperboundoutagecapacity} hold with equality except possibly at the points of discontinuity of $C_\eta(P,W_\mbf G),$ of which there are, at most, countably many. Furthermore, if  $\mbf G$ has a density, which is positive except on a set with Lebesgue measure equal to zero, then the bounds in \eqref{lowerboundoutagecapacity} and \eqref{upperboundoutagecapacity} coincide for all possible values of $\eta,$ regardless of whether $C_{\eta}(P,W_\mbf G)$ is continuous at $\eta$ or not and it holds that
\begin{align}
    C_{\eta}(P,W_\mbf G)=\max \ \Big\{R: \underset{\mathbf{Q}\in\mc Q_{P}}{\min }\mbb P\left[f(\mbf G,\mbf Q)<R \right] \leq \eta\Big\}. \label{outagecapacitycontinuitystrongmontonewachsend} 
\end{align}
\end{theorem}
The proof of the Theorem \ref{cetathmMIMO} is provided in Section \ref{proofoutagecapacity}.
\begin{remark}
The outage capacity formula in \eqref{outagecapacitycontinuitystrongmontonewachsend} is the one provided in \cite{Tse} for MIMO slow Rayleigh fading channels.
\end{remark}
    \begin{theorem}
    \label{outagecapacitySIMOSISO}
If $N_T=1,$ then the $\eta$-outage transmission capacity of the SIMO slow fading channel $W_{\mbf G}$ is equal to
\begin{align}
    C_\eta(P,W_\mbf G) =\sup \ \Big\{R: \mbb P\left[\log\det(\mathbf{I}_{N_{R}}+\frac{P}{\sigma^2}\mathbf{G}\mathbf{G}^{H})<R \right] \leq \eta\Big\}, \nonumber
\end{align}
regardless of whether it is continuous at $\eta$ or not. 

Furthermore, if $N_T=N_R=1,$ then the $\eta$-outage transmission capacity of the SISO slow fading channel $W_{\mbf G}$ is equal to
 
 \begin{align}   
        C_\eta(P,W_\mbf G)=\log\left(1+\frac{P\gamma_0}{\sigma^2}\right),\label{outageTSE}
    \end{align}
    where 
    \begin{align}
     \gamma_0=\inf\{\gamma:\mbb P[\lvert \mbf G\rvert^2\geq \gamma]\geq 1-\eta\}
     \label{generalizedinverse}
    \end{align}
    is the generalized inverse of the complementary cdf of $|\mbf G|^2.$
    \end{theorem}
The proof of Theorem \ref{outagecapacitySIMOSISO} is provided in Section \ref{SIMOSISO}.
  \begin{remark}
  If the cdf of $\lvert \mbf G \rvert^2$ is continuous and strictly monotone increasing, then the generalized inverse in \eqref{generalizedinverse} coincides with the normal inverse and the outage transmission capacity formula in \eqref{outageTSE} coincides with the one provided in \cite{Tse}.
  \end{remark}
\begin{theorem}
For the model described in Section \ref{systemmodel}, the $\eta$-outage CR capacity satisfies
\begin{align}
C_{\eta,CR}^{X,Y}(P,W_\mbf G) \geq  
  \underset{ \substack{U \\{\substack{U \circlearrow{X} \circlearrow{Y}\\ I(U;X)-I(U;Y) \leq l(\eta)}}}}{\max} I(U;X)  \label{lowerboundcr}
\end{align}
and
\begin{align}
C_{\eta,CR}^{X,Y}(P,W_\mbf G) \leq  
  \underset{ \substack{U \\{\substack{U \circlearrow{X} \circlearrow{Y}\\ I(U;X)-I(U;Y) \leq u(\eta)}}}}{\max} I(U;X),  \label{upperboundcr}
\end{align}
where
$l(\eta)$ and $ u(\eta)$ are defined in \eqref{Retasupell} and \eqref{Retasupu}, respectively. 
The lower and upper bound in \eqref{lowerboundcr} and \eqref{upperboundcr} hold with equality except at the points where $l(\eta)$ and $u(\eta)$ do not coincide, of which there are, at most, countably many.
\label{ccretathmMIMO}
\end{theorem}
The proof of Theorem \ref{ccretathmMIMO} is provided in Section \ref{proofoutagecrcapacity}.
\subsection{Auxiliary Result}
For the proof the lower bound in Theorem \ref{cetathmMIMO}, we require the following result on the  capacity of a compound MIMO complex Gaussian channel with fixed noise covariance matrix equal to $\sigma^2 \mbf I_{N_{R}}$ and with $N_R\times N_T$
channel matrix  whose operator norm is bounded from above. This is illustrated in what follows. Let $a>0$ be fixed arbitrarily.
Define the set 
 \begin{equation}
    \mc B_a=\{ \mbf g \in \mbb C^{N_{R}\times N_{T}}: \lVert \mbf g \rVert \leq a\}.
    \label{setba} \nonumber
\end{equation}
Let $\mc G_{a}$ be any closed subset of $\mc B_a.$ We define the compound channel $$\mathcal{C}=\{ W_{\mbf g}: \mbf g \in \mc G_{a}\}.$$ 
We define first an achievable transmission rate and the transmission capacity for the compound channel $\mc C.$
\begin{definition}
\label{defcompundrate}
 A real number $R$ is called an \textit{achievable} rate for the compound channel $\mc C=\{ W_{\mbf g}: \mbf g \in \mc G_{a}\} $ if for every $\theta,\delta>0$ and all $\mbf g \in \mc G_{a}$ there exists a code sequence $(\Gamma_n)_{n=1}^\infty,$ where each code $\Gamma_n$ of block-length $n$ is defined according to Definition \ref{defcode},  such that
    \[
        \frac{\log\lVert \Gamma_n\rVert}{n}\geq R-\delta
    \]
    and 
    \[
     \ e(\Gamma_n,\mbf g)\leq\theta,
    \]
    for sufficiently large $n,$ where $e(\Gamma_n,\mbf g)$ is defined in Definition \ref{defcode}.
\end{definition}
\begin{definition}
The compound capacity of  $\mc C$ is the supremum of all achievable rates for $\mc C$ defined according to Definition \ref{defcompundrate}.
\end{definition}
\begin{theorem}
\label{capacitycompoundchannels}
 The compound capacity of $\mathcal{C}$ is equal to $$\underset{\mbf Q \in \mc Q_{P}}{\max}\underset{\mbf g \in \mc G_{a}}{\min} \log\det(\mbf I_{N_{R}}+\frac{1}{\sigma^2}\mbf g \mbf Q \mbf g^H).$$
\end{theorem}
The proof of Theorem \ref{capacitycompoundchannels} is provided in Section \ref{proofcompoundcapacity}.
\section{Proof of Theorem \ref{cetathmMIMO}}
\label{proofoutagecapacity}
\subsection{Proof of the Lower Bound on the Outage Transmission Capacity}
\label{prooflowerboundMIMO}
Under the assumption of the validity of Theorem \ref{capacitycompoundchannels}, which will be proved in Section \ref{proofcompoundcapacity}, we will show that
    \begin{equation}
       C_\eta(P,W_\mbf G)\geq l(\eta)-\mu \epsilon, \nonumber
    \end{equation}
   for some $1\leq \mu \leq 2,$ where $\epsilon$ is an arbitrarily small positive constant and where 
    \begin{align}
   l(\eta)=\sup \ \Big\{R: \underset{\mathbf{Q}\in\mc Q_{P}}{\inf }\mbb P\left[f(\mbf G,\mbf Q)<R \right] < \eta\Big\}.
\nonumber
\end{align}
   Clearly, from the definition of $l(\eta),$  it holds that
    \begin{align}
         P_{\infsub}^{(\epsilon)}=\underset{\mbf Q \in \mc Q_{P}}{\inf }\mbb P\left[f(\mbf G,\mbf Q)<l(\eta)-\epsilon \right] <\eta.
        \nonumber
    \end{align}
We fix $\alpha_1>0$ to be sufficiently small such that $P_{\infsub}^{(\epsilon)}+\alpha_1\leq \eta.$
    We will show in the following lemma that for sufficiently large $n$, we can  choose a non-singular
    $\hat{\mbf Q}\in \mc Q_{P}$ such that for some $1\leq \mu \leq 2,$
    \begin{align}
        \mbb P \left[f(\mbf G,\hat{\mbf Q})<l(\eta)-\mu\epsilon \right] \nonumber
        &\leq \eta. \nonumber 
        \nonumber
    \end{align}
    \begin{lemma}
    \label{existencenonsingularcovariancematrix}
    For sufficiently large $n,$ there exists a non-singular
    $\hat{\mbf Q}\in \mc Q_{P}$ satisfying for some $1\leq \mu \leq 2$
        \begin{align}
        \mbb P \left[f(\mbf G,\hat{\mbf Q})<l(\eta)-\mu \epsilon \right] \nonumber
        &\leq  \eta. \nonumber
        \end{align}
    \end{lemma}
    \begin{proof}
 
      Notice first that from the definition of $P_{\inf}^{(\epsilon)},$ there exists a $\tilde{\mbf Q}\in \mc Q_{P}$ such that
    \begin{align}
      \mbb P\left[f(\mbf G,\tilde{\mbf Q})<l(\eta)-\epsilon \right] \leq P_{\inf}^{(\epsilon)}+\frac{\alpha_1}{2}. \nonumber\end{align}
      Now, one can find a $1\leq \mu \leq 2$ such that
       \begin{align}
      \mbb P\left[f(\mbf G,\tilde{\mbf Q})=l(\eta)-\mu\epsilon \right]=0. \label{existencemu}\end{align}
      It holds then that
      \begin{align}
        \mbb P\left[f(\mbf G,\tilde{\mbf Q})<l(\eta)-\mu\epsilon \right]
        &\leq  \mbb P\left[f(\mbf G,\tilde{\mbf Q})<l(\eta)-\epsilon \right] \nonumber \\
        &\leq P_{\inf}^{(\epsilon)}+\frac{\alpha_1}{2}. \nonumber
        \end{align}
    It is known that for all $\mbf g \in \mbb C^{N_R\times N_T}$ there exists a sequence of non-singular $(\mbf Q_{n})_{n=1}^\infty,$ with each $\mbf Q_n \in \mc Q_P,$ converging to $\tilde{\mbf Q},$ regardless of whether $\tilde{\mbf Q}$ is singular or not. It follows from \eqref{existencemu} that
   $\left( \mbf I_{\{ \mbf g \in \mbb C^{N_R\times N_T}: f(\mbf g, \mbf Q_{n})<l(\eta)-\mu\epsilon \}}(\mbf g) \right)_{n=1}^{\infty}$ converges to $\mbf I_{\{ \mbf g \in \mbb C^{N_R\times N_T}: f(\mbf g, \tilde{\mbf Q})<l(\eta)-\mu\epsilon \}}(\mbf g) $ almost surely, where $\mbf I_{\{\cdot\}}$ is the indicator function.
   Therefore, it follows using the Lebesgue's dominated convergence theorem that for sufficiently large $n:$
   \begin{align}
       \mbb P \left[f(\mbf G, \mbf Q_n)<l(\eta)-\mu\epsilon \right]&= \int \mbb I_{\{ f(\mbf G, \mbf Q_{n})<l(\eta)-\mu \epsilon \}}   d\mbb P \nonumber \\
       &\leq \int \mbb I_{\{ f(\mbf G, \tilde{\mbf Q})<l(\eta)-\mu\epsilon \}}   d\mbb P+\frac{\alpha_1}{2} \nonumber \\&= \mbb P \left[f(\mbf G,\tilde{\mbf Q})<l(\eta)-\mu\epsilon \right]+\frac{\alpha_1}{2} \nonumber \\
       &\leq P_{\infsub}^{(\epsilon)}+\alpha_1 \nonumber \\
       & \leq \eta. \nonumber
   \end{align}
   Therefore, one can find for sufficiently large $n$ a non-singular a  $\hat{\mbf Q} \in \mc Q_P$ such that
   \begin{align}
          \mbb P \left[f(\mbf G, \hat{\mbf Q})<l(\eta)-\mu\epsilon\right]\leq \eta. \nonumber 
   \end{align}
    \end{proof}
	 Since $\hat{\mbf Q}\in \mc Q_{P}$ is non-singular, it holds that
    \begin{align}
       \underset{a\rightarrow \infty}{\lim}  \underset{ \substack{\mbf g \\{\lVert\mbf g \rVert=a}}}{\min}f(\mbf g,\hat{ \mbf Q})=\infty. \label{equalaintildega} 
    \end{align}
		Now, consider the set
   $$\hat{\mc G}_{a}=\{ \mbf g \in \mbb C^{N_{R}\times N_{T}}:l(\eta)-\mu\epsilon\leq f(\mbf g,\hat{\mbf Q}) \  \text{and} \ \lVert \mbf g \rVert \leq a \}   $$
   for some $a>0$ chosen sufficiently large such that
 \begin{align}
 \Big\{\mbf g \in \mbb C^{N_R\times N_T}: \lVert \mbf g
 \rVert=a\Big\}\subseteq \hat{\mc G}_{a}.
  \label{choicea}
 \end{align}
   From \eqref{equalaintildega}, we know the existence of an $a>0$ satisfying \eqref{choicea}.

   Now, since the set $\Big\{ \mbf g \in \mbb C^{N_R\times N_T}:l(\eta)-\mu\epsilon \leq f(\mbf g,\hat{\mbf Q})\Big\}$ is closed, it follows that $\hat{\mc G}_{a}$ is a closed subset of $\mc B_a=\{ \mbf g \in \mbb C^{N_{R}\times N_{T}}: \lVert \mbf g \rVert \leq a\}.$ By applying Theorem \ref{capacitycompoundchannels}, it follows that the compound capacity of $\hat{\mc C}=\{W_{\mbf g}: \mbf g\in \hat{\mc G}_{a}\}$ is equal to 
\begin{align}
    \underset{\mbf Q \in \mc Q_{P}}{\max}\underset{\mbf g\in\hat{\mc G}_{a}}{\min} f(\mbf g,\mbf Q). \nonumber
\end{align} \color{black}
Since $\hat{\mbf Q}\in \mc Q_{P}$, it follows
that $$\underset{g\in\hat{\mc G}_{a}}{\min} f(\mbf g,\hat{\mbf Q})$$ is  an achievable rate for $\hat{\mc C}$.
   
  Let $\theta,\delta>0$.
Since $\color{black}l(\eta)-\mu\epsilon \leq \underset{\mbf g\in\hat{\mc G}_{a}}{\min} f(\mbf g,\hat{\mbf Q})\ \color{black},$ there exists a code sequence $(\Gamma_{\hat{\mc G}_{a},n})_{n=1}^\infty$ and a block-length $n_0$ such that
    \[
        \frac{\log\lVert \Gamma_{\hat{\mc G}_{a},n}\rVert}{n}\geq l(\eta)-\mu\epsilon-\delta
    \]
    and such that 
 \begin{align}
    &\mbf g \in \hat{\mc G}_{a} \implies e(\Gamma_{\hat{\mc G}_{a},n},\mbf g) \leq \theta     \nonumber 
    \end{align}
    for $n\geq n_0.$
   
    Next, we will prove the following lemma:
    \begin{lemma}
     For $n\geq n_0$
    \begin{align}
     \mbf g \in \mc B_a^{c}=\{\mbf g \in \mbb C^{N_R\times N_T}: \lVert \mbf g \rVert >a    \} \Longrightarrow e(\Gamma_{\mbf g,n},\mbf g) \leq \theta, \nonumber
    \end{align}
    where $\Gamma_{\mbf g,n}$ is some code  with block-length $n,$ and with the same size and the same encoder as $\Gamma_{\hat{\mc G}_{a},n}.$
    \end{lemma}
 
\begin{proof}
Suppose first that for a Gaussian channel $W_{\mbf g_{1}}$, a code $\Gamma^{(1)}$ satisfies $e(\Gamma^{(1)},\mbf g_{1}) \leq \theta.$ Then, it can be shown that  there exists a code $\Gamma^{(2)}$ for $W_{\mbf g_{2}}$, the channel from which $W_{\mbf g_{1}}$  was degraded  such that  $e(\Gamma^{(2)},\mbf g_{2})\leq \theta.$ The code $\Gamma^{(2)}$ has the same encoder as $\Gamma^{(1)}$ but has possibly a different decoder. The analogous statement for DMCs  is a special case of the statement provided in \cite{NoteShannon}.

 Now, let $\mbf g$ with $\lVert \mbf g \rVert >a$ be fixed arbitrarily. We recall that $a$ satisfies $\eqref{choicea}.$ Then, the channel $W_{\mbf g'}$ with $\mbf g'=\frac{a}{\lVert \mbf g \rVert}\mbf g\in \hat{\mc G}_{a}$ is a degraded version of the channel $W_{\mbf g}.$ 
 It follows that there exists a code sequence $(\Gamma_{\mbf g,n})_{n=1}^{\infty}$ for $W_\mbf g$ such that  each code $\Gamma_{\mbf g,n}$ of block-length $n$ has the same encoder and the same size as the code $\Gamma_{\hat{\mc G}_{a},n}$ of block-length $n$ but a different decoder adjusted to $\mbf g$ and such that for $n\geq n_0,$ $e(\Gamma_{\mbf g,n},\mbf g)\leq \theta$. Here, we require channel state information at the receiver side (CSIR) so that the decoder can adjust its decoding strategy according to the channel state.
 \end{proof}

So far, we have proved the existence of a block-length $n_0$ and of a
code sequence $(\Gamma_{n})_{n=1}^\infty,$ where each code $\Gamma_n$ of block-length $n$ has the same size and the same encoder as  the code $\Gamma_{\hat{\mc G}_{a},n}$ of block-length $n$ and a decoder adjusted to the actual gain $\mbf g,$ such that
    \[
        \frac{\log\lVert\Gamma_{n}\rVert}{n}\geq l(\eta)-\mu\epsilon-\delta
    \]
    and such that 
 \begin{align}
    &\mbf g \in \hat{\mc G}_{a}\cup \mc B_a^{c} \implies e(\Gamma_{n},\mbf g) \leq \theta     \nonumber 
    \end{align}
    for $n\geq n_0.$

    Now, we have for $n\geq n_0$
    \begin{align}
       \mbb P \left[e(\Gamma_{n},\mbf G) \leq \theta\right]& \geq \mbb P \left[ \mbf G \in \hat{\mc G}_{a}\cup \mc B_a^{c}  \right] \nonumber\\
       &=\mbb P\left[ \mbf G \in \hat{\mc G}_{a}\right]+\mbb P \left[ \mbf G \in \mc B_a^{c}   \right] \nonumber \\
       &\overset{(a)}{\geq} \mbb P \left[ f(\mbf G,\hat{\mbf Q})\geq l(\eta)-\mu\epsilon  \right] \nonumber \\
       &\geq 1-\eta, \nonumber
    \end{align}
    where $(a)$ follows from the choice of the constant $a.$
    This completes the proof of the lower-bound on the $\eta$-outage transmission capacity.
\subsection{Proof of the Upper Bound on the Outage Transmission Capacity}
We will show that
\begin{align}
   C_\eta(P,W_\mbf G)\leq u(\eta),
    \label{conversequation}
\end{align}
where
\begin{align}
    u(\eta)= \sup \ \Big\{R: \underset{\mathbf{Q}\in\mc Q_{P}}{\inf }\mbb P\left[f(\mbf G,\mbf Q)<R \right] \leq \eta\Big\}. 
    \nonumber
\end{align}
The weak converse for compound channels does not guarantee that the error probability cannot be made arbitrarily small for all possible states when the target rate exceeds the compound capacity. Therefore, we cannot use the weak converse theorem of compound channels to prove the upper bound in \eqref{conversequation}. We will proceed differently. 
   Suppose \eqref{conversequation} were not true. Then there exists an $\epsilon>0$ such that $u(\eta)+\epsilon$ is an achievable $\eta$-outage transmission rate for $W_{\mbf G}$. The goal is to find a contradiction. Choose $\theta>0$ so small that 
	\[
		\frac{(1-\theta)\epsilon}{2}-\theta u(\eta)>\frac{\epsilon}{4}.
	\]
	Due to the achievability of $u(\eta)+\epsilon,$ there exists a code sequence $(\Gamma_n)_{n=1}^{\infty}$ such that 
	\begin{align}
		\frac{\log\lVert\Gamma_n\rVert}{n}&\geq u(\eta)+\frac{\epsilon}{2} 
		\label{ratecondition}
		\end{align}
		and
		\begin{align}
		\mbb P[e(\Gamma_n,\mbf G)>\theta]&\leq\eta
		\label{erro2}
	\end{align}
	for sufficiently large $n.$ Choose an $n$ for which the above holds and which satisfies
	\begin{align}
		\frac{1}{n}\leq\frac{\epsilon}{8}.
		\label{fixedblocklength}
	\end{align}
	
 The uniformly-distributed message $W$ is mapped to the random input sequence $\bs{T}^n=(\bs{T}_1,\hdots,\bs{T}_n)$ of $W_\mbf G.$
We fix the covariance matrices $\mbf Q_1,\dots, \mbf Q_n$ of the random  inputs $\bs{T}_1,\hdots,\bs{T}_n,$ respectively and let $\mbf  Q^{\star}=\frac{1}{n}\sum_{i=1}^{n}\mbf Q_{i}.$ 
	Furthermore, we let 
	\[
		\epsilon'=\frac{\epsilon}{8}.
	\]
We consider the following set:

\begin{align}
    \mc G_{\theta}
    &=\{\mbf g \in \mbb C^{N_R\times N_T}: f(\mbf g,\mbf Q^\star)< u(\eta)+\epsilon' \ \text{and} \ e(\Gamma_n,\mbf g)\leq \theta     \}.
\nonumber
\end{align}
To complete the proof of the upper-bound in \eqref{conversequation} by contradiction, the next step is to show that the set $\mc G_\theta$ is non-empty. For this purpose, we will prove that $\mbb P \left[f(\mbf G,\mbf Q^{\star})<u(\eta)+\epsilon'\right] > \eta$ in what follows:
\begin{lemma}
\label{probbiggereta}
\begin{align}
    \mbb P\left[f(\mbf G,\mbf Q^\star)<u(\eta)+\epsilon'\right]>\eta. \nonumber
    \end{align}
\end{lemma}
\begin{proof}
By Lemma \ref{traceQ} below, we know that $\text{tr}(\mbf Q^\star)\leq P$ and therefore $\mbf Q^\star \in \mc Q_{P}.$ By Lemma \ref{supQ} below, it follows that
\begin{align}
    R(\mbf Q^\star)&= \sup \Big\{R: \mbb P\left[f(\mbf G,\mbf Q^\star)<R\right]\leq \eta \Big\} \nonumber \\
    &\leq u(\eta). \nonumber
\end{align}
This yields
\begin{align}
    \mbb P \left[ f(\mbf G,\mbf Q^\star)< u(\eta)+\epsilon' \right] 
    &\geq \mbb P \left[ f(\mbf G,\mbf Q^\star)< R(\mbf Q^\star)+\epsilon' \right] \nonumber \\
    &>\eta. \nonumber
\end{align}
\end{proof}
\begin{lemma}
\begin{align}
    \mathrm{tr}(\mbf Q^\star)\leq P. \nonumber
\end{align}
\label{traceQ}
\end{lemma}
\begin{proof}
From \eqref{energyconstraintMIMOCorrelated}, it holds that
\begin{align}
    \frac{1}{n}\sum_{i=1}^{n} \bs{T}_{i}^H\bs{T}_{i}\leq P, \quad \text{almost surely.} \nonumber
\end{align}
This implies that
\begin{align}
    \mbb E\left[\frac{1}{n}\sum_{i=1}^{n} \bs{T}_{i}^H\bs{T}_{i}\right] &=\frac{1}{n}\sum_{i=1}^{n}\mbb E\left[ \bs{T}_{i}^H\bs{T}_{i}\right] \nonumber \\
    &\leq P.\nonumber
\end{align}
This yields
\begin{align}
    \text{tr}\left[\mbf Q^\star\right] &= \text{tr} \left[\frac{1}{n}\sum_{i=1}^{n} \mbf Q_{i}\right] \nonumber \\
    &=\frac{1}{n}\sum_{i=1}^{n}\text{tr}\left[ \mbf Q_i\right] \nonumber\\
    &\leq \frac{1}{n}\sum_{i=1}^{n} \text{tr}\left(\mbb E\left[ \bs{T}_{i}\bs{T}_{i}^H\right]\right) \nonumber \\
    &=\frac{1}{n}\sum_{i=1}^{n} \mbb E\left[ \text{tr}\left(\bs{T}_{i}\bs{T}_{i}^H\right)\right] \nonumber \\
    &=\frac{1}{n}\sum_{i=1}^{n} \mbb E\left[ \text{tr}\left(\bs{T}_{i}^H\bs{T}_{i}\right)\right] \nonumber \\
    &=\frac{1}{n}\sum_{i=1}^{n} \mbb E\left[ \bs{T}_{i}^H\bs{T}_{i}\right] \nonumber \\
    &\leq P, \nonumber
\end{align}
where we used $r=\text{tr}(r)$ for scalar $r$, $\text{tr}\left(\mbf A \mbf B\right)= \text{tr}\left(\mbf B \mbf A\right)$ and the linearity of the expectation and of the trace operators.
\end{proof}
\begin{lemma}
\label{supQ}
For any $\mbf Q \in \mc Q_{P}$, it holds that
\begin{align}
    \sup \Big\{R: \mbb P\left[f(\mbf G,\mbf Q)<R\right]\leq \eta \Big\} 
   &\leq \sup \Big\{R: \underset{\mbf Q \in \mc Q_{P}}{\inf}\mbb P\left[f(\mbf G,\mbf Q)<R\right] \leq \eta\Big\}.\nonumber
\end{align}
\end{lemma}
\begin{proof}
For any $\mbf Q \in \mc Q_{P}$, we have
\begin{align}
    \Big\{R: \ \mbb P\left[f(\mbf G,\mbf Q)<R\right]\leq \eta \Big\} 
   &\subseteq \Big\{R:\ \underset{\mbf Q \in \mc Q_{P}}{\inf}\mbb P\left[f(\mbf G,\mbf Q)<R\right] \leq \eta\Big\}. \nonumber
\end{align}
As a result:
\begin{align}
    \sup \Big\{R: \mbb P\left[f(\mbf G,\mbf Q)<R\right]\leq \eta \Big\}\leq \sup \Big\{R: \underset{\mbf Q \in \mc Q_{P}}{\inf}\mbb P\left[f(\mbf G,\mbf Q)<R \right] \leq \eta\Big\}.\nonumber
\end{align}
\end{proof}
Now, we can prove that the set $\mc G_\theta$ is non-empty in what follows:
\begin{lemma}
 $\mc G_{\theta}$ is a non-empty set.
\end{lemma}
\begin{proof}
By Lemma \ref{probbiggereta}, we have 
\begin{align}
    \eta &<\mbb P \left[ f(\mbf G,\mbf Q^\star)<u(\eta)+\epsilon'\right] \nonumber \\
    &=\mbb P \left[ f(\mbf G,\mbf Q^\star)<u(\eta)+\epsilon'| e(\Gamma_n,\mbf G)\leq \theta \right] \mbb P \left[e(\Gamma_n,\mbf G)\leq \theta\right] +\mbb P \left[ f(\mbf G,\mbf Q^\star)<u(\eta)+\epsilon'| e(\Gamma_n,\mbf G)>\theta \right] \mbb P \left[e(\Gamma_n,\mbf G)> \theta\right]  \nonumber \\
    &\leq \mbb P \left[ f(\mbf G,\mbf Q^\star)<u(\eta)+\epsilon'| e(\Gamma_n,\mbf G)\leq \theta \right]+ \mbb P \left[e(\Gamma_n,\mbf G)> \theta\right] \nonumber \\
    &\leq \mbb P \left[ f(\mbf G,\mbf Q^\star)<u(\eta)+\epsilon'| e(\Gamma_n,\mbf G)\leq \theta \right]+\eta, \nonumber
\end{align}
where we used  \eqref{erro2} in the last step.
This implies that
\begin{align}
    \mbb P \left[ f(\mbf G,\mbf Q^\star)<u(\eta)+\epsilon'| e(\Gamma_n,\mbf G)\leq \theta \nonumber \right]>0. \nonumber
\end{align}
Furthermore, since $\eta <1$, it follows  that
\begin{align}
    \mbb P\left[ e(\Gamma_n,\mbf G)\leq \theta\right]\geq 1-\eta >0. \nonumber
\end{align}
As a result, we have
\begin{align}
       \mbb P \left[ f(\mbf G,\mbf Q^\star)<u(\eta)+\epsilon', e(\Gamma_n,\mbf G)\leq \theta \nonumber \right]>0, \nonumber
\end{align}
which means that
\begin{align}
\mbb P \left[\mbf G \in \mc G_{\theta}\right]>0 \nonumber
\end{align}
and therefore $\mc G_{\theta}$ is a non-empty set.
\end{proof}
Pick a $\mbf g \in \mc G_{\theta}$ and consider the channel
\begin{align}
\bs{z}_{i}=\mbf g\bs{t}_{i}+\bs{\xi}_{i} \quad i=1, \hdots,n.  \label{channgelgtheta}
\end{align}
 The uniformly-distributed message $W$ is mapped to the random input sequence $\bs{T}^n=(\bs{T}_1,\hdots,\bs{T}_n)$ of the channel in \eqref{channgelgtheta}.
 We model the random output sequence of the channel in \eqref{channgelgtheta} by $\bs{Z}^n=(\bs{Z}_1,\hdots,\bs{Z}_n).$ We model the random decoded message by $\hat{W}.$ The set of messages is denoted by $\mathcal{W}.$ 
 We use $\Gamma_n$ as a transmission-code for the channel in \eqref{channgelgtheta} with the fixed block-length $n$ satisfying \eqref{fixedblocklength}.
Since $\mbf g \in \mc G_\theta,$ it follows that
\begin{align}
    \mbb P\left[ W\neq \hat{W}\right]\leq e(\Gamma_n,\mbf g)\leq  \theta. \nonumber
\end{align}
We have
\begin{align}
    H(W)&= \log \lvert \mathcal{W} \rvert \nonumber\\
    &=\log \lVert \Gamma_n \rVert \nonumber \\
    &\geq n\left(u(\eta)+\frac{\epsilon}{2}\right), 
    \label{entropy1}
\end{align}
where we used \eqref{ratecondition} in the last step.
By applying Fano's inequality, we obtain
\begin{align}
    H(W|\hat{W})&\leq 1+ \mbb P\left[W\neq \hat{W}\right] \log \lvert \mathcal{W} \rvert \nonumber \\
    &\leq 1+\theta \log \lvert \mathcal{W} \rvert \nonumber \\
    &=1+\theta H(W). \nonumber
\end{align}

Now, on the one hand, it holds that
\begin{align}
    I(W;\hat{W})&=H(W)-H(W|\hat{W}) \nonumber \\
    &\geq (1-\theta)H(W)-1, \nonumber 
\end{align}
which yields
\begin{align}
   H(W)\leq \frac{1+I(W;\hat{W})}{1-\theta}. \label{applyfano}
\end{align}
On the other hand, we have
\begin{align}
    \frac{1}{n} I(W;\hat{W})&\overset{(a)}{\leq} \frac{1}{n} I(\bs{T}^n;\bs{Z}^n) \nonumber \\
    &\overset{(b)}{=}\frac{1}{n}\sum_{i=1}^{n} I(\bs{Z}_i;\bs{T}^n|\bs{Z}^{i-1}) \nonumber \\
    &=\frac{1}{n} \sum_{i=1}^{n} h(\bs{Z}_i|\bs{Z}^{i-1})-h(\bs{Z}_i|\bs{T}^{n},\bs{Z}^{i-1}) \nonumber \\
    &\overset{(c)}{=} \frac{1}{n} \sum_{i=1}^{n} h(\bs{Z}_i|\bs{Z}^{i-1})-h(\bs{Z}_i|\bs{T}_i) \nonumber \\
    &\overset{(d)}{\leq} \frac{1}{n} \sum_{i=1}^{n} h(\bs{Z}_i)-h(\bs{Z}_i|\bs{T}_i) \nonumber \\
    &= \frac{1}{n} \sum_{i=1}^{n} I(\bs{T}_i,\bs{Z}_i)
    \nonumber \\
    &\leq  \sum_{i=1}^{n} \frac{1}{n} \log\det\left(\mbf I_{N_{R}}+\frac{1}{\sigma^2}\mbf g \mbf Q_{i} \mbf g^{H}\right) \nonumber \\
    &\overset{(e)}{\leq}\log\det\left(\frac{1}{n} \sum_{i=1}^{n}\left[\mbf I_{N_{R}}+\frac{1}{\sigma^2}\mbf g \mbf Q_{i} \mbf g^{H} \right]     \right) \nonumber \\
    &=\log\det\left(\mbf I_{N_{R}}+\frac{1}{\sigma^2}\mbf g \left(\frac{1}{n}\sum_{i=1}^{n}\mbf Q_{i}\right) \mbf g^{H}  \right) \nonumber \\
    &=\log\det\left(\mbf I_{N_{R}}+\frac{1}{\sigma^2}\mbf g \mbf Q^{\star} \mbf g^{H}  \right),
    \label{mutinfconverse}
\end{align} 
where $(a)$ follows from the Data Processing Inequality because $W\circlearrow{\bs{T}^n}\circlearrow{\bs{Z}^n}\circlearrow{\hat{W}}$ forms a Markov chain, $(b)$ follows from the chain rule of mutual information, (c) follows because
$\bs{T}_{1},\dots, \bs{T}_{i-1},\bs{T}_{i+1},\dots, \bs{T}_{n},\bs{Z}^{i-1} \circlearrow{\bs{T}_{i}}\circlearrow{\bs{Z}_{i}}$ forms a Markov chain, $(d)$ follows because conditioning does not increase entropy and $(e)$ follows because $\log\circ\det$ is concave on the set of Hermitian positive semi-definite matrices.

This yields
\begin{align}
H(W)\leq \frac{1+n\log\det(\mbf I_{N_{R}}+\frac{1}{\sigma^2}\mbf g \mbf Q^{\star} \mbf g^{H})}{1-\theta}. 
\label{entropy2}
\end{align}
The inequalities \eqref{entropy1} and \eqref{entropy2} imply that
\begin{align}
    n\left(u(\eta)+\frac{\epsilon}{2}\right)&\leq \frac{1+n\log\det(\mbf I_{N_{R}}+\frac{1}{\sigma^2}\mbf g \mbf Q^{\star} \mbf g^{H})}{1-\theta} \nonumber \\
    &< \frac{1+n(u(\eta)+\epsilon')}{1-\theta},
    \label{equivalent}
\end{align}
where we used that $\mbf g \in \mc G_{\theta}.$
The inequality \eqref{equivalent} is equivalent to
\begin{align}
    -\theta u(\eta)+(1-\theta)\frac{\epsilon}{2}-\frac{1}{n}< \epsilon'. \nonumber
\end{align}
However, by the choice of $\theta$ and $n$, the left-hand side of this inequality is strictly larger than $\frac{\epsilon}{8},$ whereas $\epsilon'=\frac{\epsilon}{8}.$ This is a contradiction. 
 Thus \eqref{conversequation}  must be true. This completes the proof of the upper-bound on the $\eta$-outage transmission capacity.
\subsection{\text{Equality of the Bounds at the Points of Continuity of} \texorpdfstring{$C_{\eta}(P,W_{\mbf G})$}{TEXT}}
We will show that the bounds in \eqref{lowerboundoutagecapacity} and in \eqref{upperboundoutagecapacity} are tight except at the points of discontinuity of $C_{\eta}(P,W_\mbf G).$
Notice first that $u: \eta \rightarrow \sup \ \Big\{R: \underset{\mathbf{Q}\in\mc Q_{P}}{\inf }\mbb P\left[f(\mbf G,\mbf Q)<R \right] \leq \eta\Big\}$ is monotone non-decreasing. Therefore, the set $\mc D\subset [0,1)$ of $\eta,$ at which it is discontinuous, is at most countable. We will prove next the following lemma.
\begin{lemma}
The function
\begin{align}
    g_{\inf}: R\rightarrow\underset{\mathbf{Q}\in\mc Q_{P}}{\inf }\mbb P\left[f(\mbf G,\mbf Q)<R \right] \label{ginfR}
\end{align}
is  non-decreasing.
\label{nondecreasingginf}
\end{lemma}
\begin{proof}
Let $0\leq R_1\leq R_2.$ For any $\mbf Q \in \mc Q_{P},$ it holds that
\begin{align}
    \underset{\mathbf{Q}\in\mc Q_{P}}{\inf }\mbb P\left[f(\mbf G,\mbf Q)<R_1 \right] \leq \mbb P\left[f(\mbf G,\mbf Q)<R_1 \right]. \label{eq12}
\end{align}
 Clearly, the function $s_{\mbf Q}: R \rightarrow \mbb P\left[f(\mbf G,\mbf Q)<R \right]$ is non-decreasing for any $\mbf Q \in \mc Q_{P}.$  Therefore, it follows that for any $\mbf Q \in \mc Q_{P}$
\begin{align}
     \mbb P\left[f(\mbf G,\mbf Q)<R_1 \right]\leq \mbb P\left[f(\mbf G,\mbf Q)<R_2 \right]. \label{eq22}
\end{align}
It follows from \eqref{eq12} and \eqref{eq22}
that for all $\mbf Q \in \mc Q_{P}$
\begin{align}
  \underset{\mathbf{Q}\in\mc Q_{P}}{\inf }\mbb P\left[f(\mbf G,\mbf Q)<R_1 \right] \leq \mbb P\left[f(\mbf G,\mbf Q)<R_2 \right]. \nonumber
\end{align}
This yields
\begin{align}
     \underset{\mathbf{Q}\in\mc Q_{P}}{\inf }\mbb P\left[f(\mbf G,\mbf Q)<R_1 \right] \leq\underset{\mathbf{Q}\in\mc Q_{P}}{\inf } \mbb P\left[f(\mbf G,\mbf Q)<R_2 \right]. \nonumber
\end{align}
This proves that
\begin{align}
    g_{\inf}(R_1)\leq g_{\inf}(R_2). \nonumber
\end{align}
We deduce that the function in \eqref{ginfR} is non-decreasing.
\end{proof}
 Select now any $\eta^{\star} \in [0,1)\setminus \mc D$ and a strictly increasing sequence $(\eta^{(n)})_{n=1}^{\infty}$ in $[0,1)$ converging to $\eta^\star.$ One can show analogously to the proof of Lemma \ref{supremummaximum} below and using Lemma \ref{nondecreasingginf} that
 \begin{align}
 \sup \ \Big\{R: \underset{\mathbf{Q}\in\mc Q_{P}}{\inf }\mbb P\left[f(\mbf G,\mbf Q)<R \right] <\eta^\star\Big\} = \underset{n\rightarrow\infty}{\lim} \sup \ \Big\{R: \underset{\mathbf{Q}\in\mc Q_{P}}{\inf }\mbb P\left[f(\mbf G,\mbf Q)<R \right] \leq \eta^{(n)}\Big\}. \nonumber
 \end{align}
 It follows that
\begin{align}
   l(\eta^\star)&= \sup \ \Big\{R: \underset{\mathbf{Q}\in\mc Q_{P}}{\inf }\mbb P\left[f(\mbf G,\mbf Q)<R \right] <\eta^\star\Big\} \nonumber \\
    &= \underset{n\rightarrow\infty}{\lim} \sup \ \Big\{R: \underset{\mathbf{Q}\in\mc Q_{P}}{\inf }\mbb P\left[f(\mbf G,\mbf Q)<R \right] \leq \eta^{(n)}\Big\} \nonumber \\
    &= \underset{n\rightarrow\infty}{\lim} u(\eta^{(n)}) \nonumber\\
    &\overset{(a)}{=} u(\eta^\star),\nonumber
\end{align}
where $(a)$ follows because $u(\eta)$ is continuous non-decreasing at $\eta^\star.$
Sofar, we know that 
    $u$ has at most countably many points of discontinuity and that $l(\eta)$ and $u(\eta)$ coincide in points of continuity of $u(\eta)$, and in particular, they are equal to $C_{\eta}(P,W_\mbf G)$ in these points.
    
    Now assume that $l(\eta_0)\neq u(\eta_0)$ in some point $\eta_0$. We are going to show that $C_{\eta}(P,W_\mbf G)$ is not continuous at $\eta_0$.

By assumption, $l(\eta_0)<u(\eta_0)$. Let $(\eta_n^+)$ be a sequence of points of continuity of $u(\eta)$ converging to $\eta_0$ from above, and let $(\eta_n^-)$ be a sequence of points of continuity of $u(\eta)$ converging to $\eta_0$ from below. Then by Lemma 1,
\[
    l(\eta_n^-)=C_{\eta_n^-}(P,W_{\mbf G}),\quad u(\eta_n^+)=C_{\eta_n^+}(P,W_{\mbf G})
\]
for all $n$. In particular,
\[
    \limsup_{n\to\infty}C_{\eta_n^-}(P,W_{\mbf G})
    =\limsup_{n\to\infty}l(\eta_n^-)
    \leq l(\eta_0)
    <u(\eta_0)
    \leq\liminf_{n\to\infty}u(\eta_n^+)
    =\liminf_{n\to\infty}C_{\eta_n^+}(P,W_{\mbf G}).
\]
Hence $C_{\eta}(P,W_{\mbf G})$ is not continuous at $\eta_0$.

\subsection{\text{Equality of the Bounds in} \eqref{Retasupell} and \eqref{Retasupu} \text{when} $\mbf G$ \text{has a positive density except on a set with Lebesgue measure equal to zero}}
Let us first introduce and prove the following lemma:
\begin{lemma}
\label{strongmonotonicitiy}
When $\mbf G$ has a positive density except on a set with Lebesgue measure equal to zero, the function
\begin{align}
g_{\inf}:R\rightarrow \underset{\mathbf{Q}\in\mc Q_{P}}{\inf }\mbb P\left[f(\mbf G,\mbf Q)<R \right]. \label{ginf2}
\end{align}
is strictly monotone increasing.
\end{lemma}
\begin{proof}
We introduce and prove first the following claims:
\begin{claim} 
\textit{The infimum in} \eqref{ginf2} is a minimum. \label{claim1}
\end{claim}
\begin{claimproof}
Let $(\mbf Q_n)_{n=1}^{\infty}$ with each $\mbf Q_n \in \mc Q_P,$ such that
$\mbb P\left[f(\mbf G,\mbf Q_n)\leq R \right]$ converges to $\underset{\mbf Q \in \mc Q_{P}}{\inf} \ \mbb P\left[f(\mbf G,\mbf Q)\leq R \right].$
Since $\mc Q_P$ is a compact set, there exists a $\mbf Q_0 \in \mc Q_P$ such that 
$\underset{n\rightarrow\infty}{\lim} \mbf Q_n=\mbf Q_0.$ 
Notice that $\mbf I_{\{f(\mbf G,\mbf Q_n)\leq R \}}$ converges to  $\mbf I_{\{f(\mbf G,\mbf Q_0)\leq R \}}$ except on a set with Lebesgue measure equal to zero, where $\mbf I_{\{\cdot\}}$ refers to the indicator function.
Thus, it follows using the dominated  convergence theorem that
\begin{align}
    \mbb P \left[ f(\mbf G,\mbf Q_0)\leq R \right]&= \int \mbf I_{\{f(\mbf G,\mbf Q_0)\leq R\}} d\mbb P \nonumber \\
    &\overset{(a)}{=}\int \underset{n\rightarrow\infty}{\lim} \mbf I_{\{f(\mbf G,\mbf Q_n)\leq R\}} d\mbb P \nonumber \\
    &= \underset{n\rightarrow\infty}{\lim} \mbb P\left[ f(\mbf G,\mbf Q_n)\leq R \right] \nonumber \\
    &=  \underset{\mathbf{Q}\in\mc Q_{P}}{\inf }\mbb P\left[f(\mbf G,\mbf Q)<R \right], \nonumber
\end{align}
where $(a)$ follows from the absolute continuity of $r_{\mbf g}: \mbf Q \rightarrow f(\mbf g ,\mbf Q)$ in $\mc Q_P.$ 
Therefore, the infimum in \eqref{ginf2} is actually a minimum. This completes the proof of Claim \ref{claim1}.
\end{claimproof}
\begin{claim}
When $\mbf G$ has a positive density except on a set with Lebesgue measure equal to zero, the function $s_{\mbf Q}: R \rightarrow \mbb P\left[f(\mbf G,\mbf Q)<R \right]$ is strictly monotone increasing for any $\mbf Q \in \mc Q_{P}.$ \label{claim2}
\end{claim}
\begin{claimproof}
Notice first that for any $\mbf Q \in \mc Q_P,$ the function $f_{\mbf Q}: \mbf g \rightarrow f(\mbf g ,\mbf Q)$ is continuous. Let $0\leq R_1<R_2.$ Consider the open interval $(R_1,R_2).$ Denote the inverse image of $(R_1,R_2)$ under $f_{\mbf Q}$ by $f_{\mbf Q}^{-1}(R_1,R_2).$
From the continuity of $f_\mbf Q,$ it follows that $f_{\mbf Q}^{-1}(R_1,R_2)$
is an open and non-empty set.
This yields
\begin{align*}
    \mbb P\left[ R_1\leq f(\mbf G,\mbf Q)< R_2 \right] &\geq \mbb P \left[ \mbf G \in f_{\mbf Q}^{-1}(R_1,R_2)  \right] \nonumber \\
    &>0,
\end{align*}
where we used the fact that $\mbf G$ has a positive density except on a set with Lebesgue measure equal to zero. Therefore, the function $s_{\mbf Q}: R \rightarrow \mbb P\left[f(\mbf G,\mbf Q)<R \right]$ is strictly monotone increasing for any $\mbf Q \in \mc Q_{P}.$ This completes the proof of Claim \ref{claim2}.
\end{claimproof}
Now that we proved the two claims, we let $0\leq R_1< R_2.$ For any $\mbf Q \in \mc Q_{P},$ it holds that
\begin{align}
    \underset{\mathbf{Q}\in\mc Q_{P}}{\min }\mbb P\left[f(\mbf G,\mbf Q)<R_1 \right] \leq \mbb P\left[f(\mbf G,\mbf Q)<R_1 \right]. \label{eq1}
\end{align}

From Claim \ref{claim2}, we know that the function $s_{\mbf Q}: R \rightarrow \mbb P\left[f(\mbf G,\mbf Q)<R \right]$ is strictly monotone increasing in $R$ for any $\mbf Q \in \mc Q_{P}.$ This implies that for any $\mbf Q \in \mc Q_{P}$
\begin{align}
     \mbb P\left[f(\mbf G,\mbf Q)<R_1 \right]< \mbb P\left[f(\mbf G,\mbf Q)<R_2 \right]. \label{eq2}
\end{align}
It follows from \eqref{eq1} and \eqref{eq2}
that for all $\mbf Q \in \mc Q_{P}$
\begin{align}
  \underset{\mathbf{Q}\in\mc Q_{P}}{\min }\mbb P\left[f(\mbf G,\mbf Q)<R_1 \right] < \mbb P\left[f(\mbf G,\mbf Q)<R_2 \right]. \nonumber
\end{align}
This yields
\begin{align}
     \underset{\mathbf{Q}\in\mc Q_{P}}{\min }\mbb P\left[f(\mbf G,\mbf Q)<R_1 \right] <\underset{\mathbf{Q}\in\mc Q_{P}}{\min} \mbb P\left[f(\mbf G,\mbf Q)<R_2 \right]. \nonumber
\end{align}
It follows using Claim \ref{claim1} that
\begin{align}
    g_{\inf}(R_1)< g_{\inf}(R_2). \nonumber
\end{align}
We deduce that the function in \eqref{ginf2} is strictly monotone increasing. This completes the proof of Lemma \ref{strongmonotonicitiy}.
\end{proof}
Now that we proved Lemma \ref{strongmonotonicitiy}, suppose that $l(\eta)\neq u(\eta).$ Then, for any $l(\eta)<R<u(\eta),$ it follows from the strict monotonicity of $g_{\inf}$ that
\begin{align}
    g_{\inf}(l(\eta))<g_{\inf}(R)<g_{\inf}(u(\eta)), \nonumber
\end{align}
where $g_{\inf}(u(\eta))\leq \eta$ and since $R>l(\eta),$ it follows that $g_{\inf}(R)\geq \eta.$ Therefore, we have $g_{\inf}(R)<g_{\inf}(u(\eta))\leq g_{\inf}(R),$ which is a contradiction. Therefore, $l(\eta)$ and $u(\eta)$ must be equal.

\section{ Proof of Theorem \ref{outagecapacitySIMOSISO}}
\label{SIMOSISO}
\subsection{Proof of the outage transmission capacity for \texorpdfstring{$N_T=1$}{TEXT}}
\subsubsection{Direct Proof}
Under the assumption of the validity of Theorem \ref{capacitycompoundchannels}, which will be proved in Section \ref{proofcompoundcapacity}, we will show that for $N_T=1$
    \begin{equation}
       C_\eta(P,W_\mbf G)\geq R_{\eta,\sup}, \nonumber
    \end{equation}
   where 
    \begin{align}
   R_{\eta,\sup}=\sup \ \Big\{R: \mbb P\left[\log\det(\mathbf{I}_{N_{R}}+\frac{P}{\sigma^2}\mathbf{G}\mathbf{G}^{H})<R \right] \leq \eta\Big\}.   \label{supremumprob}
\end{align}
We first show that the supremum in \eqref{supremumprob} is actually a maximum.
 \begin{lemma}
 \label{supremummaximum}
 \begin{align}
   \mbb P\left[\log\det(\mathbf{I}_{N_{R}}+\frac{P}{\sigma^2}\mathbf{G}\mathbf{G}^{H})<R_{\eta,\sup}\right] \leq \eta \nonumber
 \end{align}
 so the supremum in \eqref{supremumprob} is actually a maximum.
 \end{lemma}
	\begin{proof}
    Let $R_n\nearrow R_{\eta,\sup}$ be a sequence converging to $R_{\eta,\sup}$ from the left. Then 
    \[
        \{R\in \mbb R:R <R_{\eta,\sup}\}=\bigcup_{n=1}^\infty\{R\in \mbb R:R <R_n\}.
    \]
    From the sigma-continuity of probability measures, it follows that 
    \begin{align*}
        \mbb P[ \log\det(\mathbf{I}_{N_{R}}+\frac{P}{\sigma^2}\mathbf{G}\mathbf{G}^{H})<R_{\eta,\sup}]&=\underset{n\rightarrow \infty}{\lim}\mbb P[\log\det(\mathbf{I}_{N_{R}}+\frac{P}{\sigma^2}\mathbf{G}\mathbf{G}^{H})<R_n] \\
        &\leq\eta.
        \end{align*}
\end{proof}
	Now, consider the set
   $$\tilde{\mc G}_{a}=\{ \mbf g \in \mbb C^{N_{R}\times 1}:R_{\eta,\sup}\leq \log\det(\mathbf{I}_{N_{R}}+\frac{P}{\sigma^2}\mathbf{g}\mathbf{g}^{H})\  \text{and} \ \lVert \mbf g \rVert \leq a \}   $$
   for some $a>0$ chosen sufficiently large such that
 \begin{align}
 \Big\{\mbf g \in \mbb C^{N_R\times 1}: \lVert \mbf g
 \rVert=a\Big\}\subseteq \tilde{\mc G}_{a}.
  \nonumber
 \end{align}
 Such an $a>0$ exists because
 \begin{align}
       \underset{a\rightarrow \infty}{\lim}  \underset{ \substack{\mbf g \\{\lVert\mbf g \rVert=a}}}{\min}\log\det(\mathbf{I}_{N_{R}}+\frac{P}{\sigma^2}\mathbf{g}\mathbf{g}^{H})=\infty. \nonumber
    \end{align}
 Since the set $\Big\{ \mbf g \in \mbb C^{N_R\times 1}:R_{\eta,\sup} \leq \log\det(\mathbf{I}_{N_{R}}+\frac{P}{\sigma^2}\mathbf{g}\mathbf{g}^{H})\Big\}$ is closed, it follows that $\tilde{\mc G}_{a}$ is a closed subset of $\mc B_a=\{ \mbf g \in \mbb C^{N_{R}\times 1}: \lVert \mbf g \rVert \leq a\}.$ By applying Theorem \ref{capacitycompoundchannels} for $N_T=1$, it follows that the compound capacity of $\tilde{\mc C}=\{W_{\mbf g}: \mbf g\in \tilde{\mc G}_{a}\}$ is equal to 
\begin{align}
\underset{\mbf g\in\tilde{\mc G}_{a}}{\min} \log\det(\mathbf{I}_{N_{R}}+\frac{P}{\sigma^2}\mathbf{g}\mathbf{g}^{H}). \nonumber
\end{align} \color{black}
  
  Let $\theta,\delta>0$. One can now use the same argument as in the MIMO case to prove the existence of a block-length $n_0$ and of a
code sequence $(\Gamma_{n})_{n=1}^\infty,$ where each code $\Gamma_n$ of block-length $n$ has the same size and the same encoder as  the code $\Gamma_{\tilde{\mc G}_{a},n}$ of block-length $n$ and a decoder adjusted to the actual gain $\mbf g,$ such that
    \[
        \frac{\log\lVert\Gamma_{n}\rVert}{n}\geq R_{\eta,\sup}-\delta
    \]
    and such that 
 \begin{align}
    &\mbf g \in \tilde{\mc G}_{a}\cup \mc B_a^{c}  \implies e(\Gamma_{n},\mbf g) \leq \theta     \nonumber 
    \end{align}
    for $n\geq n_0.$

   We have for $n\geq n_0$
    \begin{align}
       \mbb P \left[e(\Gamma_{n},\mbf G) \leq \theta\right]& \geq \mbb P \left[ \mbf G \in \tilde{\mc G}_{a}\cup \mc B_a^{c}  \right] \nonumber\\
       &\geq \mbb P \left[ \log\det(\mathbf{I}_{N_{R}}+\frac{P}{\sigma^2}\mathbf{G}\mathbf{G}^{H})\geq R_{\eta,\sup} \right] \nonumber \\
       &\geq 1-\eta. \nonumber
    \end{align}
    This completes the direct proof of the $\eta$-outage transmission capacity for $N_T=1.$
\color{black}
\subsubsection{Converse Proof}
We are going to show that for $N_T=1,$
$$C_\eta(P,W_\mbf G)\leq \sup \ \Big\{R: \mbb P\left[\log\det(\mathbf{I}_{N_{R}}+\frac{P}{\sigma^2}\mathbf{G}\mathbf{G}^{H})<R \right] \leq \eta\Big\}.$$
For this purpose, we introduce and prove the following lemma:
\begin{lemma}
\label{intermediatelemma}
For $N_T=1,$ it holds that
\begin{align*}
    u(\eta)&=\sup \ \Big\{R: \underset{\mathbf{Q}\in\mc Q_{P}}{\inf }\mbb P\left[f(\mbf G,\mbf Q)<R \right] \leq \eta\Big\} \\
    &\leq  \sup \ \Big\{R: \mbb P\left[\log\det(\mathbf{I}_{N_{R}}+\frac{P}{\sigma^2}\mathbf{G}\mathbf{G}^{H})<R \right] \leq \eta\Big\}.
\end{align*}
\end{lemma}
\begin{proof}
Notice first that for $N_T=1,$  it holds that for any $\mbf Q \in \mc Q_{P}$ and any $\mbf g \in \mbb C^{N_R\times 1}$
\begin{align*}
    f(\mbf g, \mbf Q)&=\log\det(\mathbf{I}_{N_{R}}+\frac{1}{\sigma^2}\mathbf{g}\mathbf{Q}\mathbf{g}^{H})\\&\leq \log\det(\mathbf{I}_{N_{R}}+\frac{P}{\sigma^2}\mathbf{g}\mathbf{g}^{H}).
\end{align*}
Therefore, for any $\mbf Q \in \mc Q_P$ and any $R \in \mbb R,$ we have 
\begin{align*}
     \mbb P\left[\log\det(\mathbf{I}_{N_{R}}+\frac{P}{\sigma^2}\mathbf{G}\mathbf{G}^{H})<R \right] \leq \mbb P \left[f(\mbf G, \mbf Q)<R\right].
\end{align*}
This implies that for any $\mbf Q \in \mc Q_P$ and any $R\in \mbb R$
\begin{align*}
     \mbb P\left[\log\det(\mathbf{I}_{N_{R}}+\frac{P}{\sigma^2}\mathbf{G}\mathbf{G}^{H})<R \right] \leq \underset{\mbf Q \in \mc Q_{P}}{\inf} \mbb P \left[f(\mbf G, \mbf Q)<R\right].
\end{align*}
It follows that 
\begin{align}
 \Big\{R:\ \underset{\mbf Q \in \mc Q_{P}}{\inf}\mbb P\left[f(\mbf G,\mbf Q)<R\right] \leq \eta\Big\} 
   &\subseteq   \Big\{R: \ \mbb P\left[\log\det(\mathbf{I}_{N_{R}}+\frac{P}{\sigma^2}\mathbf{G}\mathbf{G}^{H})<R\right]\leq \eta \Big\}  . \nonumber
\end{align}
This implies that
\begin{align}
u(\eta) &\leq \sup \ \Big\{R: \mbb P\left[\log\det(\mathbf{I}_{N_{R}}+\frac{P}{\sigma^2}\mathbf{G}\mathbf{G}^{H})<R \right] \leq \eta\Big\}.   \nonumber
\end{align}
\end{proof}
Now, from Theorem \ref{cetathmMIMO}, we know that
\begin{align}
    C_\eta(P,W_\mbf G)\leq u(\eta). \nonumber
\end{align}
By Lemma \ref{intermediatelemma}, it follows that for $N_T=1$
\begin{align}
    C_\eta(P,W_\mbf G) \leq \sup \ \Big\{R: \mbb P\left[\log\det(\mathbf{I}_{N_{R}}+\frac{P}{\sigma^2}\mathbf{G}\mathbf{G}^{H})<R \right] \leq \eta\Big\}. \nonumber
\end{align}
This completes the converse proof of the $\eta$-outage transmission capacity for $N_T=1.$
\color{black}    
	
\subsection{Alternative Proof of the outage transmission capacity for \texorpdfstring{$N_T=N_R=1$}{TEXT}}
\label{proofSISOcapacity}
In this section, we will show that the $\eta$-outage transmission capacity for the SISO case is equal to
 
 \begin{align}   
        C_\eta(P,W_\mbf G)=\log\left(1+\frac{P\gamma_0^2}{\sigma^2}\right),\nonumber
    \end{align}
    where 
     $$\gamma_0=\inf\{\gamma:\mbb P[\lvert \mbf G\rvert^2\geq \gamma]\leq1-\eta\}.$$
    
  Analogously to the proof of Lemma \ref{supremummaximum}, one can first show that for $N_T=N_R=1$
    \[
        \mbb P[|\mbf G|^2\geq \gamma_0]\leq 1-\eta.
    \]

\subsubsection{Direct Proof} We will show that for $N_T=N_R=1$
    \begin{equation}\label{eq:outage_geqSISO}
        C_\eta(P,W_\mbf G)\geq\log\left(1+\frac{P\gamma_0}{\sigma^2}\right).
    \end{equation}
  Let $s \in \mbb C $ such that $|s|^2=\gamma_0$ and let $\theta,\delta>0.$ It is well-known that there exists a code sequence $(\Gamma_{s,n})_{n=1}^\infty$ for the channel $W_s$ and a block-length $n_0$ such that for $n\geq n_0$, the rate of each code $\Gamma_{s,n}$ of block-length $n$ satisfies
    \[
        \frac{\log\lVert \Gamma_{s,n}\rVert}{n}\geq\log\left(1+\frac{P\gamma_0}{\sigma^2}\right)-\delta
    \]
    and such that
    \begin{align}
        e(\Gamma_{s,n},s)\leq \theta. \nonumber
    \end{align}
    
  For any $g$ with $|g|^2\geq \gamma_0,$ 
  the SISO Gaussian channel $W_s$ is degraded from the SISO Gaussian channel $W_{g}.$ Analogously to the MIMO case, it follows that there exists a code sequence $(\Gamma_{g,n})_{n=1}^{\infty}$ for $W_g$ such that  each code $\Gamma_{g,n}$ of block-length $n$ has the same encoder and the same size as $\Gamma_{s,n}$ but a different decoder adjusted to $g$ and such that for $n\geq n_0,$ $e(\Gamma_{g,n}, g)\leq \theta$. Here, we require channel state information at the receiver side (CSIR) so that the decoder can adjust its decoding strategy according to the channel state.
    So far, we have proved the existence of a
code sequence $(\Gamma_{n})_{n=1}^\infty$ and a block-length $n_0$ such that
    \[
        \frac{\log\lVert\Gamma_{n}\rVert}{n}\geq \log\left(1+\frac{P\gamma_0}{\sigma^2}\right)-\delta
    \]
    and such that 
 \begin{align}
    &\lvert g \rvert^2 \geq \gamma_0 \implies e(\Gamma_{n},g) \leq \theta     \nonumber 
    \end{align}
    for $n\geq n_0.$
    Now, for $n\geq n_0,$ we have
    \begin{align}
        \mbb P[e(\Gamma_n,\mbf G)\leq\theta]
        &\geq\mbb P[\lvert \mbf G\rvert^2\geq\gamma_0] \nonumber \\
        &\geq 1-\eta. \nonumber 
    \end{align}
    This implies \eqref{eq:outage_geqSISO} and completes the direct proof.
    \subsubsection{Converse proof}
      We will show that for $N_T=N_R=1$
    \begin{equation}\label{eq:outage_leq}
       C_\eta(P,W_\mbf G)\leq\log\left(1+\frac{P\gamma_0}{\sigma^2}\right). 
    \end{equation}
    Suppose this were not true. Then there exists an $\varepsilon>0$ such that for all $\theta,\delta>0$ there exists a code sequence $(\Gamma_n)_{n=1}^\infty$ satisfying 
    \begin{equation}\label{eq:larger}
        \frac{\log\lVert \Gamma_n\rVert}{n}\geq\log\left(1+\frac{P(\gamma_0+\varepsilon)}{\sigma^2}\right)-\delta
    \end{equation}
    and
    \begin{equation}\label{eq:converse_err}
        \mbb P[e(\Gamma_n,\mbf G)\leq\theta]\geq 1-\eta
    \end{equation}
    for sufficiently large $n$.  Since $\delta$ may be arbitrary, we may choose it in such a way that the right-hand side of \eqref{eq:larger} is strictly larger than $\log(1+(P\gamma_0)/\sigma^2)$. We define $\gamma_1$ to be the solution of the equation $$\log(1+(P\gamma_1)/\sigma^2)=\log\left(1+\frac{P(\gamma_0+\varepsilon)}{\sigma^2}\right)-\delta.$$  $\gamma_1$ is chosen such that the rate of the code sequence is greater than the capacity of the channel $W_g$ when $|g|^2<\gamma_1$. Therefore, even under the CSIR assumption, the strong converse for SISO Gaussian channels implies that for large $n,$ the error probability is greater than $\theta$ when  $|g|^2<\gamma_1.$ 
    It follows that 
       \begin{align}
        \mbb P[e(\Gamma_n,\mbf G)>\theta]
        &\geq\mbb P[\lvert \mbf G\rvert^2 < \gamma_1] \nonumber \\
        &>\eta, \nonumber
        \end{align}
        by the definition of $\gamma_0$, where we used that $\gamma_1>\gamma_0$ from the choice of $\delta$. This is a contradiction to \eqref{eq:converse_err}, and so \eqref{eq:outage_leq} must be true. This completes the converse proof.
\section{Proof of Theorem \ref{ccretathmMIMO}}
\label{proofoutagecrcapacity}
\subsection{Proof of the Lower Bound on the Outage CR Capacity}
\label{prooflowerboundcrcapacity}
\subsubsection{\text{If} \texorpdfstring{$l(\eta)=0$}{TEXT}} It is shown in \cite{part2} that when the terminals do not communicate over the channel, the CR capacity defined in \cite{part2} is equal to 
\begin{align}
    H_{0}=\underset{ \substack{U \\{\substack{U \circlearrow{X} \circlearrow{Y}\\ I(U;X)-I(U;Y) \leq 0}}}}{\max} I(U;X).  \nonumber
\end{align}
Hence, when the terminals do not communicate over the MIMO slow fading channel $W_\mbf G$, $H_0$ is also an achievable $\eta$-outage CR rate. Therefore, we have 
    \begin{align}
C_{\eta,CR}^{X,Y}(P,W_\mbf G) &\geq 
  \underset{ \substack{U \\{\substack{U \circlearrow{X} \circlearrow{Y}\\ I(U;X)-I(U;Y) \leq 0}}}}{\max} I(U;X)  \nonumber \\
  &=  \underset{ \substack{U \\{\substack{U \circlearrow{X} \circlearrow{Y}\\ I(U;X)-I(U;Y) \leq l(\eta)}}}}{\max} I(U;X).  \nonumber
\end{align}
\subsubsection{If \texorpdfstring{$l(\eta)>0$}{TEXT}} We extend the coding scheme provided in \cite{part2} to MIMO slow fading channels. By continuity, it suffices to show that 
$$ \underset{ \substack{U \\{\substack{U \circlearrow{X} \circlearrow{Y}\\ I(U;X)-I(U;Y) \leq R'}}}}{\max} I(U;X)  $$ is an achievable $\eta$-outage CR rate for every $R'<l(\eta).$
Let $U$ be a random variable satisfying $U \circlearrow{X} \circlearrow{Y}$ and $I(U;X)-I(U;Y) \leq R'$. Let the upper-bound  $0\leq \eta<1$ on the outage probability, from the CR generation perspective, be fixed arbitrarily. We are going to show that $H=I(U;X)$ is an achievable $\eta$-outage CR rate. Let $\alpha,\delta>0$. Without loss of generality, assume that the distribution of $U$ is a possible type for block-length $n$.
For any $\mu>0,$ we let
\begin{align}
N_{1}&=\lfloor 2^{n[I(U;X)-I(U;Y)+3\mu]} \rfloor \nonumber
\end{align}
and
\begin{align}
N_{2}&=\lfloor 2^{n[I(U;Y)-2\mu]}\rfloor. \nonumber
\end{align}For each pair $(i,j)$ with $1\leq i \leq N_{1}$ and $1\leq j \leq N_{2}$, we define a random sequence $\bs{U}_{i,j}\in\mathcal{U}^n$ of type $P_{U}$. Let $\mbf M=\bs{U}_{1,1},\hdots, \bs{U}_{N_{1},N_{2}}$  be the joint random variable of all $\bs{U}_{i,j}s.$ We define $\Phi_{\mbf M}$ as follows:
 Let $\Phi_{\mbf M}(X^n)=\bs{U}_{ij}$, if $\bs{U}_{ij}$ is jointly $UX$-typical with $X^n$ (either one if there are several). If no such $\bs{U}_{i,j}$ exists, then  $\Phi_{\mbf M}(X^n)$ is set to a constant sequence $\bs{u}_0$ different from all the  ${\bs{U}_{ij}}s$, jointly $UX$-typical with none of the realizations of $X^n$ and known to both terminals.
 
We further define the following two sets which depend on $\mbf M$:
\begin{align}
    S_{1}(\mbf M)&=\{(\bs{x},\bs{y}):(\Phi_{\mbf M}(\bs{x}),\bs{x},\bs{y}) \in \mathcal{T}_{U,X,Y}^{n}\} \nonumber
\end{align} and
\begin{align}
    S_{2}(\mbf M)=\{(\bs{x},\bs{y}):(\bs{x},\bs{y}) \in S_{1}(\mbf M) \ \text{s.t.} \ \exists \ \bs{U}_{i,\ell}\neq\bs{U}_{i,j}=\Phi(\bs{x}) \nonumber \\   \text{jointly} \ UY\text{-typical with} \ \bs{y} \ (\text{with the same first index} \ i)
\}.\nonumber
\end{align}
It is proved in \cite{part2} that 
\begin{align}
    \mathbb{E}_{\mbf M}\left[ \mbb P\left[(X^n,Y^n)\notin  S_{1}(\mbf M)\right]+\mbb P\left[(X^n,Y^n)\in  S_{2}(\mbf M)\right]\right]\leq \beta(n),
    \label{averagebeta}
\end{align}
where $\beta(n) \leq \frac{\alpha}{2}$ for sufficiently large $n$. 
We choose a realization $\mbf m=\bs{u}_{1,1},\hdots, \bs{u}_{N_1,N_2}$ satisfying:
\begin{align}
\mbb P\left[(X^n,Y^n)\notin  S_{1}(\mbf m)\right]+\mbb P\left[(X^n,Y^n)\in  S_{2}(\mbf m)\right]\leq \beta(n). \nonumber
\end{align} 
From \eqref{averagebeta}, we know that such a realization exists. We denote $\Phi_{\mbf m}$ by $\Phi.$
We assume that each $\bs{u}_{i,j}, i=1\hdots N_1, j=1\hdots N_2,$  is known to both terminals.  
This means that  $N_{1}$ codebooks $C_{i}, 1\leq i \leq N_{1}$, are known to both terminals, where each codebook contains $N_{2}$ sequences, $ \bs{u}_{i,j}, \ j=1,\hdots, N_2$. 

Let $\bs{x}$ be any realization of $X^n$ and $\bs{y}$ be any realization of $Y^n.$
 Let $f_1(\bs{x})=i$ if $\Phi(\bs{x})=\bs{u}_{i,j}$. Otherwise, if $\Phi(\bs{x})=\bs{u}_{0},$ then $f_1(\bs{x})=N_1+1.$
  Since $ C'<l(\eta)$, we choose $\mu$ to be sufficiently small such that
      \begin{align}
     \frac{\log \lVert f_1 \rVert}{n}&=\frac{\log(N_1+1)}{n} \nonumber \\
     &\leq l(\eta)-\mu',
     \label{inequalitylogfSISO}
      \end{align}
for some $\mu'>0,$
 The message $i^\star=f_1(\bs{x})$, with $i^\star\in\{1,\hdots,N_1+1\},$ is encoded to a sequence $\mbf t$ using a code sequence $(\Gamma^\star_n)_{n=1}^{\infty},$ where each code $\Gamma^\star_n$ of block-length $n$ is defined according to Definition \ref{defcode}, with rate $\frac{\log \lVert \Gamma^\star_n \rVert}{n}=\frac{\log \lVert f_1 \rVert}{n}$ satisfying \eqref{inequalitylogfSISO}
 and with error probability $e(\Gamma^\star_n,\mbf G)$ satisfying for sufficiently large $n$
 \begin{align}
     \mbb P\left[e(\Gamma^\star_n,\mbf G) \leq \theta \right] \geq 1-\eta,
     \label{erroroutage}
 \end{align}
where $\theta$ is  a positive constant satisfying $\theta\leq \frac{\alpha}{2}.$ 
  \color{black}
  Here, $\lVert f_1 \rVert$ refers to the cardinality of the set of messages $\{i^\star:i^\star=1,\hdots,N_1+1\}
$.
  Since $l(\eta)$ is an achievable $\eta$-outage transmission rate, we know that such a code sequence exists. The sequence $\mbf t$ is sent over the MIMO slow fading channel. Let $\mbf z$ be the corresponding channel output sequence. Terminal $B$ decodes the message $\tilde{i}^\star$ from the knowledge of $\mbf z.$
Let $\Psi(\bs{y},\mbf z)=\bs{u}_{\tilde{i}^\star,j}$ if $\bs{u}_{\tilde{i}^\star,j}$ and $\bs{y}$ are jointly $UY$-typical . If there is no such  $\bs{u}_{\tilde{i}^\star,j}$ or there are several, we set $\Psi(\bs{y},\mbf z)=\bs{u}_0$ (since $K$ and $L$ must have the same alphabet).
Now, we are going to show that the requirements in $\eqref{errorMIMOcorrelated},$  $\eqref{cardinalityMIMOcorrelated}$ and $\eqref{rateMIMOcorrelated}$ are satisfied.
Clearly, (\ref{cardinalityMIMOcorrelated}) is satisfied  for $c=H(X)+\mu+1$  because
{{\begin{align}
|\mathcal{K}|&=N_1 N_2+1 \nonumber \\
             &\leq  2^{n\left[I(U;X)+\mu\right]}+1 \nonumber \\
             &\leq2^{n\left[H(X)+\mu+1\right]}.\nonumber
\end{align}}}We define next for any $(i,j)\in \{1,\hdots,n\}\times\{1,\hdots,n\}$  the set
$$\mc S=\{ \bs{x}\in\mathcal{X}^{n} \ \text{s.t.} \ (\bs{u}_{i,j},\bs{x}) \ \text{jointly} \ UX\text{-typical}\}.$$
Then, it holds that 
\begin{align}
\mbb P[K=\bs{u}_{i,j}] &=\sum_{\bs{x}\in\mc S}\mbb P[K=\bs{u}_{i,j}|X^n=\bs{x}]P_{X}^n(\bs{x}) +\sum_{\bs{x}\in\mc S^c}\mbb P[K=\bs{u}_{i,j}|X^n=\bs{x}]P_{X}^n(\bs{x}) \nonumber \\
&\overset{(a)}{=}\sum_{\bs{x}\in\mc S}\mbb P[K=\bs{u}_{i,j}|X^n=\bs{x}]P_{X}^n(\bs{x}) \nonumber \\
&\leq \sum_{\bs{x}\in\mc S}P_{X}^n(\bs{x}) \nonumber \\
&=P_{X}^{n}(\{\bs{x}: (\bs{u}_{i,j},\bs{x}) \ \text{jointly} \ UX\text{-typical}\}) \nonumber\\
& = 2^{-nI(U;X)-\kappa(n)}, \nonumber
\end{align}
for some $\kappa(n)>0$ with $\underset{n\rightarrow \infty}{\lim} \frac{\kappa(n)}{n}=0$,
where $(a)$ follows because for  $(\bs{u}_{i,j},\mathbf{x})$ being not jointly $UX$-typical, we have $\mbb P[K=\bs{u}_{i,j}|X^n=\bs{x}]=0.$ This yields
{{\begin{align}
H(K)\geq nI(U;X)-\kappa'(n)
\nonumber \end{align}}}
for some $\kappa'(n)>0$ with $\underset{n\rightarrow \infty}{\lim} \frac{\kappa'(n)}{n}=0.$
Therefore, for sufficiently large $n,$ it holds that
\begin{align}
    \frac{H(K)}{n}>H-\delta. \nonumber
\end{align}
Thus, (\ref{rateMIMOcorrelated}) is satisfied.
\begin{remark}
It is to notice that for sufficiently large $n$ $$H(K)\approx \log\lvert \mc K \rvert \approx nI(U;X).$$
Therefore the random variable $K$ is nearly uniform for sufficiently large $n.$
It follows from Remark \ref{fastgleichentropyrate} that, for sufficiently large $n$, the random variable $L$ is also nearly uniform when the system is not in outage from the CR generation perspective. As result, when the system is not in outage and for sufficiently large $n$, $(K,L)$ is a pair of  nearly uniform random variables. This is the most convenient form of CR, as already mentioned in Remark \ref{remarkUCR}.
\end{remark} 
\quad Now, it remains to prove that \eqref{errorMIMOcorrelated} is satisfied. For this purpose, we define the following event:
\begin{align}
    \mathcal{D}_{\mbf m}= ``\Phi(X^n) \ \text{is equal to none of the} \  {\bs{u}_{i,j}}s". \nonumber
\end{align}
We denote its complement by $\mc D_{\mbf m}^{c}.$
We further define $I^\star=f_1(X^n)$ to be the random message generated by Terminal $A$ and  $\tilde{I}^\star$ to be the random message decoded by Terminal $B$. 
We have
\begin{align}
    \mbb P[K\neq L|\mbf G] \nonumber &=\mbb P[K\neq L|\mbf G,I^\star=\tilde{I}^\star]\mbb P[I^\star=\tilde{I}^\star|\mbf G] + \mbb P[K\neq L|\mbf G,I^\star\neq \tilde{I}^\star]\mbb P[I^\star\neq\tilde{I}^\star|\mbf G] \nonumber \\
        &\leq \mbb P[K\neq L|\mbf G,I^\star=\tilde{I}^\star]+ \mbb P[I^\star\neq\tilde{I}^\star|\mbf G].\nonumber
\end{align}
Here,
\begin{align}
    \mbb P[K\neq L|\mbf G,I^\star=\tilde{I}^\star] 
   &= \mbb P[K\neq L|\mbf G,I^\star=\tilde{I}^\star,\mathcal{D}_{\mbf m}]\mbb P[\mathcal{D}_{\mbf m}|\mbf G,I^\star=\tilde{I}^\star] + \mbb P[K\neq L|\mbf G,I^\star=\tilde{I}^\star,\mathcal{D}_{\mbf m}^c]\mbb P[\mathcal{D}_{\mbf m}^c|\mbf G,I^\star=\tilde{I}^\star] \nonumber \\
   &\overset{(a)}{=}\mbb P[K\neq L|\mbf G,I^\star=\tilde{I}^\star,\mathcal{D}_{\mbf m}^c]\mbb P[\mathcal{D}_{\mbf m}^c|\mbf G,I^\star=\tilde{I}^\star] \nonumber \\
   &\leq \mbb P[K\neq L|\mbf G,I^\star=\tilde{I}^\star,\mathcal{D}_{\mbf m}^c],\nonumber
\end{align}
where $(a)$ follows from $\mbb P[K\neq L|\mbf G,I^\star=\tilde{I}^\star,\mathcal{D}_{\mbf m}]=0,$ since conditioned on $\mbf G$, $I^\star=\tilde{I}^\star$ and $\mathcal{D}_{\mbf m}$, we know that $K$ and $L$ are both equal to $\bs{u}_0$.
It follows that
\begin{align}
    \mbb P[K\neq L|\mbf G] 
    &\leq \mbb P[K\neq L|\mbf G,I^\star=\tilde{I}^\star,\mathcal{D}_{\mbf m}^c]+ \mbb P[I^\star\neq\tilde{I}^\star|\mbf G] \nonumber \\
    &\leq \mbb P\left[(X^n,Y^n)\in  S_{1}^{c}(\mbf m)\cup S_{2}(\mbf m)\right]+\mbb P[I^\star\neq\tilde{I}^\star|\mbf G] \nonumber \\
    &\overset{(a)}{=}\mbb P\left[(X^n,Y^n)\notin  S_{1}(\mbf m)\right]+\mbb P\left[(X^n,Y^n)\in  S_{2}(\mbf m)\right] +\mbb P[I^\star\neq\tilde{I}^\star|\mbf G] \nonumber \\
    &\leq \beta(n)+ \mbb P[I^\star\neq\tilde{I}^\star|\mbf G],\nonumber
\end{align}
where $(a)$ follows because $S_{1}^{c}(\mbf m)$ and $S_{2}(\mbf m)$ are disjoint. 
It holds that
\begin{align}
    &\mbb P\left[I^\star\neq \tilde{I}^\star|\mbf G\right]\leq \theta \implies \mbb P[K\neq L|\mbf G] \leq \beta(n)+ \theta. \nonumber
\end{align}
Since, for sufficiently large $n,$ $\beta(n)+ \theta\leq \alpha$, it follows that
\begin{align}
    &\mbb P\left[I^\star\neq \tilde{I}^\star|\mbf G\right]\leq \theta  \implies \mbb P[K\neq L|\mbf G] \leq \alpha. \nonumber
\end{align}
From \eqref{erroroutage}, we know that
\begin{align}
   \mbb P\left[ \mbb P\left[I^\star\neq \tilde{I}^\star|\mbf G\right]\leq \theta\right] \geq 1-\eta. \nonumber
\end{align}
Thus
\begin{align}
    \mbb P\left[  \mbb P[K\neq L|\mbf G] \leq \alpha          \right] &\geq \mbb P\left[\mbb P\left[ I^\star\neq \tilde{I}^\star|\mbf G     \right]\leq \theta\right] \nonumber \\
    &\geq 1-\eta. \nonumber
\end{align} 
This completes the proof of the lower-bound on the $\eta$-outage CR capacity.
\subsection{Proof of the Upper Bound on the Outage CR Capacity}
Let $0\leq \eta <1.$ Let $H$ be any achievable $\eta$-outage CR rate. So, there exists a non-negative constant $c$ such that for every $\alpha>0$ and $\delta>0$ and for sufficiently large $n,$ there exists a permissible pair of random variables $(K,L)$ according to a fixed CR-generation protocol of block-length $n$ as introduced in Section \ref{systemmodel} such that
\begin{equation}
    \mbb P\left[\mbb P\left[K\neq L|\mbf G\right]\leq \alpha \right]\geq 1-\eta, 
    \label{errorMIMOcorrelatedrepeat}
\end{equation}
\begin{equation}
    |\mathcal{K}|\leq 2^{cn},
    \label{cardinalityMIMOcorrelatedrepeat}
\end{equation}
\begin{equation}
    \frac{1}{n}H(K)> H-\delta.
     \label{rateMIMOcorrelatedrepeat}
\end{equation}
We recall that the CR generation protocol consists of:
\begin{enumerate}
    \item A function $\Phi$ that maps $X^n$ into a random variable $K$ with alphabet $\mathcal{K}$ generated by Terminal $A.$
    \item A function $\Lambda$ that maps $X^n$ into the input sequence $\bs{T}^n \in \mbb C^{N_T\times n}$  satisfying the following power constraint
\begin{equation}
\frac{1}{n}\sum_{i=1}^{n}\bs{T}_{i}^H\bs{T}_{i}\leq P, \quad \text{almost surely}.   \nonumber
\end{equation}
    \item A function $\Psi$ that maps $Y^n$ and the  output sequence $\bs{Z}^n \in \mbb C^{N_R\times n}$ into a random variable $L$ with alphabet $\mathcal{K}$ generated by Terminal $B.$
\end{enumerate}
We are going to show that for any $\epsilon>0$
\begin{align}
    \frac{H(K)}{n} \leq \underset{ \substack{U \\{\substack{U \circlearrow{X} \circlearrow{Y}\\ I(U;X)-I(U;Y) \leq u(\eta)+\zeta(n,\alpha,\epsilon)}}}}{\max} I(U;X), \nonumber 
\end{align}
where $u(\eta)$ is defined in \eqref{Retasupu}
and where $\zeta(n,\alpha,\epsilon)=\frac{1}{n}+\alpha c+\epsilon.$ 
\color{black}
In our proof, we will use  the following lemma: 
\begin{lemma} (Lemma 17.12 in \cite{codingtheorems})
For arbitrary random variables $S$ and $R$ and sequences of random variables $X^{n}$ and $Y^{n}$, it holds that
\begin{align}
 I(S;X^{n}|R)-I(S;Y^{n}|R)  
 &=\sum_{i=1}^{n} I(S;X_{i}|X_{1},\dots, X_{i-1}, Y_{i+1},\dots, Y_{n},R) \nonumber \\ &\quad -\sum_{i=1}^{n} I(S;Y_{i}|X_{1},\dots, X_{i-1}, Y_{i+1},\dots, Y_{n},R) \nonumber \\
 &=n[I(S;X_{J}|V)-I(S;Y_{J}|V)],\nonumber
\end{align}
where $V=(X_{1},\dots, X_{J-1},Y_{J+1},\dots, Y_{n},R,J)$, with $J$ being a random variable independent of $R$,\ $S$, \ $X^{n}$ \ and $Y^{n}$ and uniformly distributed on $\{1 ,\dots, n \}$.
\label{lemma1}
\end{lemma}Let $J$ be a random variable uniformly distributed on $\{1,\dots, n\}$ and independent of $K$, $X^n$ and $Y^n$. We further define $U=(K,X_{1},\dots, X_{J-1},Y_{J+1},\dots, Y_{n},J).$ It holds that $U \circlearrow{X_J} \circlearrow{Y_J}.$ 

Notice  that
{{\begin{align}
H(K)&\overset{(a)}{=}H(K)-H(K|X^{n})\nonumber\\
&=I(K;X^{n}) \nonumber\\
&\overset{(b)}{=}\sum_{i=1}^{n} I(K;X_{i}|X_{1},\dots, X_{i-1}) \nonumber\\
&=n I(K;X_{J}|X_{1},\dots, X_{J-1},J) \nonumber\\
&\overset{(c)}{\leq }n I(U;X_{J}), \nonumber
\end{align}}}where $(a)$ follows because $K=\Phi(X^n)$ and $(b)$ and $(c)$ follow from the chain rule for mutual information.

We will show next that
\begin{align}
    I(U;X_J)-I(U;Y_J) \leq u(\eta)+\zeta(n,\alpha,\epsilon), \label{toshow}
\end{align}
where $\zeta(n,\alpha,\epsilon)=\frac{1}{n}+\alpha c+\epsilon$.
Applying Lemma \ref{lemma1} for $S=K$, $R=\varnothing$ with $V=(X_1,\hdots, X_{J-1},Y_{J+1},\hdots, Y_{n},J)$ yields
{{\begin{align}
I(K;X^{n})-I(K;Y^{n}) 
&=n[I(K;X_{J}|V)-I(K;Y_{J}|V)] \nonumber\\
&\overset{(a)}{=}n[I(KV;X_{J})-I(V;X_{J})-I(KV;Y_{J})+I(V;Y_{J})] \nonumber\\ 
&\overset{(b)}{=}n[I(U;X_{J})-I(U;Y_{J})], 
\label{UhilfsvariableMIMO1}
\end{align}}}where $(a)$ follows from the chain rule for mutual information and from the fact that $V$ is independent of $(X_{J},Y_{J})$ and $(b)$ follows from $U=(K,V)$. It results using (\ref{UhilfsvariableMIMO1}) that
{{\begin{align}
n[I(U;X_{J})-I(U;Y_{J})]
&=I(K;X^{n})-I(K;Y^{n}) \nonumber\\
&=H(K)-I(K;Y^{n})\nonumber \\ 
&=H(K|Y^n).
\label{star2MIMO2}
\end{align}}}

Next, to prove \eqref{toshow}, we will show that
\begin{align}
    \frac{H(K|Y^n)}{n}\leq u(\eta)+\zeta(n,\alpha,\epsilon). \nonumber
\end{align}

Let $\text{cov}(\bs{T}_i)=\mathbf{Q}_i$   for $i=1, \hdots,n,$ where $\bs{T}_{i}\in \mbb C^{N_T}, i=1, \hdots,n.$ We define
\begin{align}
   \mbf Q^\star=\frac{1}{n}\sum_{i=1}^{n} \mbf Q_{i}. \nonumber
\end{align}
By Lemma \ref{traceQ}, we know that $\text{tr}(\mbf Q^\star)\leq P$ and therefore $\mbf Q^\star \in \mc Q_{P}.$
Let
\begin{align}
&R(\mbf Q^\star)= \sup \Big\{R: \mbb P\left[f(\mbf G,\mbf Q^\star)< R\right]\leq \eta \Big\}. \nonumber
\end{align}
We recall that the function $f$ is defined in \eqref{fgQ}.
Since $\mbf Q^\star \in \mc Q_{P}$, Lemma \ref{supQ} implies that
\begin{align}
 R(\mbf Q^\star)\leq u(\eta).
 \label{comparerates}
 \end{align}
We consider for any $\epsilon>0$ the set
\begin{align}
    &\Omega=\Big\{\mathbf{g} \in \mathbb{C}^{N_{R}\times N_{T}}:   \mbb P\left[K\neq L|\mathbf{G}=\mathbf{g}\right]\leq \alpha \ \text{and} \ f(\mbf g,\mbf Q^\star) \leq   R(\mbf Q^\star)+\epsilon \Big\}, \nonumber
\end{align}
and define $\tilde{\mbf G}$ to be a  random matrix, independent of $X^n$,$Y^n$ and $\bs{\xi}^n$, with alphabet $\Omega$ such that for every Borel set $\seta \subseteq \mbb C^{N_R\times N_T},$ it holds that
\begin{align}
    \mbb P \left[ \tilde{\mbf G}\in \seta \right]=\mbb P \left[ \mbf G\in\seta|\mbf G\in\Omega\right]. \nonumber
\end{align}
In order to prove that such a $\tilde{\mbf G}$ is well-defined, it suffices show that $\mbb P \left[ \mbf G \in \Omega\right]>0.$ This is proved in what follows:
\begin{lemma}
\begin{align}
\mbb P\left[\mbf G \in \Omega\right]>0. \nonumber 
\end{align}
\label{probOmega}
\end{lemma}
\begin{proof}
From the definition of $R(\mbf Q^\star)$, we have
\begin{align}
    \eta &< \mbb P \left[f(\mbf G,\mbf Q^\star)<R(\mbf Q^\star)+\epsilon\right] \nonumber \\
    &\leq \mbb P\left[f(\mbf G,\mbf Q^\star)\leq R(\mbf Q^\star)+\epsilon \right]. \nonumber
\end{align}

Then, it holds that
\begin{align}
    \mbb P\left[f(\mbf G,\mbf Q^\star)\leq R(\mbf Q^\star)+\epsilon\right] =\eta_1, \nonumber
\end{align}
where $0\leq \eta<\eta_1\leq 1.$

 It follows using \eqref{errorMIMOcorrelatedrepeat} that
\begin{align}
    1-\eta 
    & \leq \mbb P\left[\mbb P\left[K\neq L|\mbf G\right]\leq \alpha\right] \nonumber \\
    &=\mbb P\left[ \mbb P\left[K\neq L\bigm|\mbf G\right]\leq \alpha\bigm|f(\mbf G,\mbf Q^\star)\leq R(\mbf Q^\star)+\epsilon \right] \mbb P\left[f(\mbf G,\mbf Q^\star)\leq R(\mbf Q^\star)+\epsilon\right] \nonumber \\
    &\quad+\mbb P\left[ \mbb P\left[K\neq L\bigm|\mbf G\right]\leq \alpha\bigm|f(\mbf G,\mbf Q^\star)> R(\mbf Q^\star)+\epsilon \right] \mbb P\left[f(\mbf G,\mbf Q^\star)> R(\mbf Q^\star)+\epsilon\right]\nonumber \\
    &=\eta_1 \ \mbb P\left[ \mbb P\left[K\neq L\bigm|\mbf G\right]\leq \alpha\bigm|f(\mbf G,\mbf Q^\star)\leq R(\mbf Q^\star)+\epsilon \right]\nonumber \\
    &\quad+(1-\eta_1) \ \mbb P\left[ \mbb P\left[K\neq L\bigm|\mbf G\right]\leq \alpha\bigm|f(\mbf G,\mbf Q^\star)> R(\mbf Q^\star)+\epsilon \right]
    \nonumber \\
    &\leq \eta_1 \ \mbb P\left[ \mbb P\left[K\neq L\bigm|\mbf G\right]\leq \alpha\bigm|f(\mbf G,\mbf Q^\star)\leq R(\mbf Q^\star)+\epsilon \right]+(1-\eta_1) \nonumber \\
    &\leq \mbb P\left[ \mbb P\left[K\neq L\bigm|\mbf G\right]\leq \alpha\bigm|f(\mbf G,\mbf Q^\star)\leq R(\mbf Q^\star)+\epsilon \right]+(1-\eta_1)\nonumber \\
    &< \mbb P\left[ \mbb P\left[K\neq L\bigm|\mbf G\right]\leq \alpha\bigm|f(\mbf G,\mbf Q^\star)\leq R(\mbf Q^\star)+\epsilon \right]+(1-\eta), \nonumber
\end{align}
where we used that $1-\eta_1<1-\eta.$ This means that
\begin{align}
    \mbb P\left[ \mbb P\left[K\neq L\bigm|\mbf G\right]\leq \alpha\bigm|f(\mbf G,\mbf Q^\star)\leq R(\mbf Q^\star)+\epsilon\right]>0. \nonumber
\end{align}
In addition, since $\eta_1>0$, we have
\begin{align}
    \mbb P\left[ \mbb P\left[K\neq L\bigm|\mbf G\right]\leq \alpha,f(\mbf G,\mbf Q^\star)\leq R(\mbf Q^\star)+\epsilon \right]>0. \nonumber
\end{align}
Thus
$$\mbb P\left[\mbf G \in \Omega\right]>0.$$
\end{proof}
Next, we fix the CR generation protocol and change the state distribution of the slow fading channel. We obtain the following new MIMO channel:
\begin{align}
    \tilde{\bs Z}_{i}=\tilde{\mbf G}\bs{T}_i+\bs{\xi}_i \quad i=1, \hdots,n,  \nonumber
\end{align}
where $\tilde{\bs Z}^n$ is the new output sequence.
We further define $\tilde{L}$ such that
\begin{align}
    \tilde{L}=\Psi(Y^n,\tilde{\bs Z}^n). \nonumber
\end{align}
Clearly, it holds for any $\mbf g \in \Omega$ that
\begin{align}
     \mbb P\left[K\neq \tilde{L}|\tilde{\mbf G}= \mbf g\right]\leq \alpha   \label{newerrorinequality}
\end{align}
and that
\begin{align}
    \log\det(\mathbf{I}_{N_{R}}+\frac{1}{\sigma^2}\mathbf{g}\mbf Q^\star\mathbf{g}^{H}) \leq R(\mbf Q^\star)+\epsilon.
    \label{absolutevalue}
\end{align}
Furthermore, since $\bs{\xi}_i \sim \mathcal{N}_{\mathbb{C}}(\bs{0}_{N_R},\sigma^2 \mathbf{I}_{N_R}), i=1, \hdots,n$, it follows for $i=1, \hdots,n$ that
\begin{align}
    I(\bs{T}_{i};\tilde{\bs Z}_{i}|\tilde{\mbf G}= \mbf g)&\leq \log\det(\mathbf{I}_{N_{R}}+\frac{1}{\sigma^2}\mathbf{g}\mathbf{Q}_{i} \mathbf{g}^{H})\quad \forall \mbf g\in\Omega.
    \label{mutualinfmax}
\end{align}
We recall that the goal is to prove that 
\begin{align}
    \frac{H(K|Y^n)}{n}\leq u(\eta)+\zeta(n,\alpha,\epsilon). \label{recall}
\end{align}

Now, we have
\begin{align}
    \frac{1}{n}H(K|Y^n)&=\frac{1}{n}H(K|\tilde{\mbf G},Y^n) \nonumber \\
            &=\frac{1}{n}H(K|\tilde{\mbf G},Y^n,\tilde{\bs Z}^n)+\frac{1}{n}I(K;\tilde{\bs Z}^n|\tilde{\mbf G},Y^n), \nonumber
\end{align}
where we used that $\tilde{\mbf G}$ is independent of $(K,Y^n).$ 
On the one hand, we have
\begin{align}
                                \frac{1}{n}H\left(K|\tilde{\bs Z}^n,\tilde{\mbf G},Y^n\right)  &\overset{(a)}{\leq } \frac{1}{n} H\left(K|\tilde{L},\tilde{\mbf G}\right) \nonumber \\    &\overset{(b)}{\leq } \frac{1}{n}\mbb E\left[1+\log|\mathcal{K}|\mbb P[K\neq \tilde{L}|\tilde{\mbf G}]\right]  \nonumber \\                          &=\frac{1}{n}+\frac{1}{n} \log|\mathcal{K}|\mbb E \left[ P[K\neq \tilde{L}|\tilde{\mbf G}]\right]  \nonumber \\
                                &\overset{(c)}{\leq } \frac{1}{n}+\frac{1}{n}\alpha \log|\mathcal{K}|  \nonumber \\                               &\overset{(d)}{\leq } \frac{1}{n}+\alpha \ c,  \nonumber
\end{align}
where (a) follows from $\tilde{L}=\Psi(Y^n,\tilde{\bs Z}^n)$, (b) follows from Fano's Inequality, (c) follows from \eqref{newerrorinequality}   and (d) follows from $\log|\mathcal{K}|\leq cn$ in \eqref{cardinalityMIMOcorrelatedrepeat}.
On the other hand, we have
 \begin{align} 
\frac{1}{n}I(K;\tilde{\bs Z}^n|\tilde{\mbf G},Y^{n}) &\leq \frac{1}{n} I(X^{n},K;\tilde{\bs Z}^n|\tilde{\mbf G},Y^{n}) \nonumber\\
& \overset{(a)}{\leq }\frac{1}{n} I(\bs{T}^n;\tilde{\bs Z}^n|\tilde{\mbf G},Y^{n})  \nonumber \\
& = \frac{1}{n} \left[h(\tilde{\bs Z}^n|\tilde{\mbf G},Y^{n})- h(\tilde{\bs Z}^n|\bs{T}^n,\tilde{\mbf G},Y^{n}) \right]\nonumber \\
& \overset{(b)}{=} \frac{1}{n} \left[h(\tilde{\bs Z}^n|\tilde{\mbf G},Y^{n})- h(\tilde{\bs Z}^n|\tilde{\mbf G},\bs{T}^n)\right] \nonumber \\
& \overset{(c)}{\leq }  \frac{1}{n} \left[h(\tilde{\bs Z}^n|\tilde{\mbf G})- h(\tilde{\bs Z}^n|\tilde{\mbf G},\bs{T}^n)\right] \nonumber \\
& =\frac{1}{n} I(\bs{T}^n;\tilde{\bs Z}^n|\tilde{\mbf G})  \nonumber \\
& \overset{(d)}{=} \frac{1}{n}\sum_{i=1}^{n} I(\tilde{\bs Z}_{i};\bs{T}^n|\tilde{\mbf G},\tilde{\bs Z}^{i-1}) \nonumber \\
& = \frac{1}{n}\sum_{i=1}^{n} h(\tilde{\bs Z}_{i}|\tilde{\mbf G},\tilde{\bs Z}^{i-1})-h(\tilde{\bs Z}_{i}|\tilde{\mbf G},\bs{T}^n,\tilde{\bs Z}^{i-1}) \nonumber \\
& \overset{(e)}{=} \frac{1}{n}\sum_{i=1}^{n} h(\tilde{\bs Z}_{i}|\tilde{\mbf G},\tilde{\bs Z}^{i-1})-h(\tilde{\bs Z}_{i}|\tilde{\mbf G},\bs{T}_{i}) \nonumber \\
& \overset{(f)}{\leq} \frac{1}{n} \sum_{i=1}^{n} h(\tilde{\bs Z}_{i}|\tilde{\mbf G})-h(\tilde{\bs Z}_{i}|\tilde{\mbf G},\bs{T}_{i}) \nonumber \\
& = \frac{1}{n}\sum_{i=1}^{n} I(\bs{T}_{i};\tilde{\bs Z}_{i}|\tilde{\mbf G}) \nonumber \\
&\overset{(g)}{\leq}\frac{1}{n} \sum_{i=1}^{n}\mbb E \left[\log\det(\mathbf{I}_{N_{R}}+\frac{1}{\sigma^2}\tilde{\mathbf{G}}\mathbf{Q}_{i} \tilde{\mathbf{G}}^{H})\right]\nonumber \\
&=\mbb E\left[\frac{1}{n}\sum_{i=1}^{n}\log\det(\mathbf{I}_{N_{R}}+\frac{1}{\sigma^2}\tilde{\mathbf{G}}\mathbf{Q}_{i} \tilde{\mathbf{G}}^{H})    \right] \nonumber \\
&\overset{(h)}{\leq} \mbb E\left[\log\det\left(\frac{1}{n}\sum_{i=1}^{n} \left[\mathbf{I}_{N_{R}}+\frac{1}{\sigma^2}\tilde{\mathbf{G}}\mathbf{Q}_{i} \tilde{\mathbf{G}}^{H}\right]\right)\right] \nonumber \\
&=\mbb E\left[\log\det\left(\mathbf{I}_{N_{R}}+\frac{1}{\sigma^2}\tilde{\mathbf{G}}\left(\frac{1}{n}\sum_{i=1}^{n}\mathbf{Q}_{i}\right)\tilde{\mathbf{G}}^{H}\right)\right] \nonumber \\
&=\mbb E\left[\log\det\left(\mathbf{I}_{N_{R}}+\frac{1}{\sigma^2}\tilde{\mathbf{G}}\mbf Q^\star\tilde{\mathbf{G}}^{H}\right)\right] \nonumber \\
& \overset{(i)}{\leq}  R(\mbf Q^\star)+\epsilon \nonumber \\
&\overset{(j)}{\leq} u(\eta)+\epsilon,\nonumber
  \end{align}where $(a)$ follows from the Data Processing Inequality because $Y^{n}\circlearrow{X^{n}K}\circlearrow{\tilde{\mbf G}\bs{T}^n}\circlearrow{\tilde{\bs Z}^{n}}$ forms a Markov chain, $(b)$ follows because $Y^{n}\circlearrow{X^{n}K}\circlearrow{\tilde{\mbf G}\bs{T}^n}\circlearrow{\tilde{\bs Z}^{n}}$ forms a Markov chain, $(c)(f)$ follow because conditioning does not increase entropy, $(d)$ follows from the chain rule for mutual information, $(e)$ follows because $\bs{T}_{1},\dots, \bs{T}_{i-1},\bs{T}_{i+1},\dots, \bs{T}_{n},\tilde{\bs Z}^{i-1} \circlearrow{\tilde{\mbf G},\bs{T}_{i}}\circlearrow{\tilde{\bs Z}_{i}}$ forms a Markov chain, $(g)$ follows from \eqref{mutualinfmax},
 $(h)$ follows from Jensen's Inequality since the function $\log\circ\det$ is concave on the set of Hermitian positive semi-definite matrices and since $\mbf I_{N_{R}}+\frac{1}{\sigma^2}\tilde{\mbf G}\mbf Q_{i}\tilde{\mbf G}^H$ is Hermitian positive semi-definite for $i=1, \hdots,n$ and $(i)$ follows from \eqref{absolutevalue} and $(j)$ follows from \eqref{comparerates}. This proves that for $0\leq \eta<1,$ \eqref{recall} is satisfied for 
 $\zeta(n,\alpha,\epsilon)=\frac{1}{n}+\alpha c+\epsilon >0.$ 
 
From (\ref{star2MIMO2}) and (\ref{recall}), we deduce that for $0\leq\eta<1$
{{\begin{align}
&I(U;X_{J})-I(U;Y_{J})  \leq u(\eta) +\zeta(n,\alpha,\epsilon), \nonumber
\end{align}}}where $U \circlearrow{X_{J}} \circlearrow{Y_{J}}.$ 

\color{black}Since the joint distribution of $X_{J}$ and $Y_{J}$ is equal to $P_{XY}$, $\frac{H(K)}{n}$ is upper-bounded by $I(U;X)$ subject to $I(U;X)-I(U;Y) \leq u(\eta) + \zeta(n,\alpha,\epsilon)$ with $U$ satisfying $U \circlearrow{X} \circlearrow{Y}$. As a result, it holds using \eqref{rateMIMOcorrelatedrepeat} that for sufficiently large $n$ and for every $\alpha,\delta,\epsilon>0,$ any achievable $\eta$-outage CR rate $H$ satisfies
\begin{align}
    H <\underset{ \substack{U \\{\substack{U \circlearrow{X} \circlearrow{Y}\\ I(U;X)-I(U;Y) \leq u(\eta)+\zeta(n,\alpha,\epsilon)}}}}{\max} I(U;X)+\delta.
    \nonumber
\end{align}
It follows that
\begin{align}
    H &\leq\underset{\alpha,\delta,\epsilon>0}{\inf} \ \underset{n\rightarrow\infty}{\lim}\left[\underset{ \substack{U \\{\substack{U \circlearrow{X} \circlearrow{Y}\\ I(U;X)-I(U;Y) \leq u(\eta)+\zeta(n,\alpha,\epsilon)}}}}{\max} I(U;X)+\delta\right] \nonumber \\
    &=\underset{ \substack{U \\{\substack{U \circlearrow{X} \circlearrow{Y}\\ I(U;X)-I(U;Y) \leq u(\eta)}}}}{\max} I(U;X). \nonumber
\end{align}
This completes the proof of the upper-bound on the $\eta$-outage CR capacity.
\section{Proof of Theorem \ref{capacitycompoundchannels}}
\label{proofcompoundcapacity}
   The goal is to prove that the capacity of  $\mathcal{C}=\{ W_{\mbf g}: \mbf g\in \mc G_{a}\}$ is  $$\underset{\mbf Q \in \mc Q_{P}}{\max}\underset{\mbf g \in \mc G_{a}}{\min} \log\det(\mbf I_{N_{R}}+\frac{1}{\sigma^2}\mbf g \mbf Q \mbf g^H).$$ In our proof, we follow the strategy of \cite{discretetimegaussian}\footnote{In\cite{discretetimegaussian}, the focus was on compound real Gaussian channels with square channel matrix whose operator norm is upper-bounded by $a$ and with noise covariance matrix satisfying further conditions.}.
\subsection{\text{Direct Proof of Theorem} \ref{capacitycompoundchannels} for finite \texorpdfstring{$\mc G_{a}$}{TEXT}}
We prove first the direct part of Theorem \ref{capacitycompoundchannels} for finite $\mc G_a.$ This result will be later extended for infinite $\mc G_a$ using an approximation inequality.
\begin{theorem}
\label{capacityfinitestate}
Let $\mc G_{a}$ be any finite subset of $\mc B_{a}.$
We define the compound channel 
$$ \mathcal{C}'=\{ W_{\mbf g}: \mbf g\in \mc G_{a}\}.$$
An achievable rate for $\mathcal{C}'$  is 
$$ \underset{\mbf Q \in \mc Q_{P}}{\max}\underset{\mbf g\in\mc G_{a}}{\min} \log\det(\mathbf{I}_{N_{R}}+\frac{1}{\sigma^2}\mathbf{g}\mathbf{Q}\mathbf{g}^{H}).  $$
\end{theorem}
\subsubsection{Auxiliary Lemmas}
In order to prove Theorem \ref{capacityfinitestate}, we introduce the following lemmas first.
\begin{lemma}{(Feinstein’s Lemma with Input Constraint\color{black})}\\~\\
For any channel  $W$ with input set $\mc T$ and  output set $\mc Z,$ with random input $T$ distributed according to $p(t)$ and with corresponding  random channel output $Z$ distributed according to $q(z)$ and for any integer $\tau \geq 1$, real number $\alpha>0$, and measurable subset $E$ of $\mc T$, there exists a code with size $\tau$, maximum error probability $\epsilon$ and block-length $n=1$, whose codewords are contained in the set $E,$ where $\epsilon$ satisfies
$$\epsilon=\tau2^{-\alpha}+\mbb P\left[ i(T,Z)\leq \alpha\right]+\mbb P\left[T \notin E \right],        $$
where
$$i(T,Z)=\log\frac{W(Z|T)}{q(Z)}.$$ 

\end{lemma}
\begin{proof}
As stated in \cite{discretetimegaussian}, the proof is the same as the one for Theorem 2 in \cite{errorbound} or Lemma 8.2.1 in \cite{informationbook}. 
\end{proof}
For any $\mbf g \in \mc G_a,$ we assume that the random input sequence $\bs{T}^n$ of $W_\mbf g$ is distributed according to $p(\bs{t}^n)$ and that the corresponding random channel output sequence $\bs{Z}^n$ is distributed according  to $q(\bs{z}^n).$
We define for any $\bs{t}^n\in \mbb C^{N_T\times n},$
any $\bs{z}^n \in \mbb C^{N_R\times n}$ and any $\mbf g \in \mc G_a$
$$i_{\mbf g}(\bs{t}^n, \bs{z}^n)=\log \frac{W_{\mbf g}(\bs{z}^n| \bs{t}^n) }{q( \bs{z}^n)}.$$
\begin{lemma}
\label{existenceerror}
For any real numbers $\alpha>0$, $\delta>0$, and any integer $\tau\geq 1$, there exists a code $\Gamma_n$ for $\mathcal{C}'$ with size $\lVert \Gamma_n\rVert=\tau$, block-length $n$ and with codewords contained in $E_n=\{ \bs{t}^n=(\bs{t}_1,\hdots,\bs{t}_n)\in \mbb C^{N_T\times n}: \frac{1}{n}\sum_{i=1}^{n}\lVert \bs{t}_i \rVert^2\leq P\}$ such that for all $\mbf g \in \mc G'_{a}$
\begin{align}
    e(\Gamma_n,\mbf g) \leq |\mc G_{a}|\tau2^{-\alpha}+|\mc G_{a}|^22^{-\delta}+|\mc G_{a}|\mbb P\left[\bs{T}^n\notin E_n\right]+\sum_{\mbf g \in \mc G_{a}}\mbb P\left[i_{\mbf g}(\bs{T}^n,\bs{Z}^n)\leq \alpha+\delta\right]. \nonumber
\end{align}

\end{lemma}
\begin{proof}
The proof is a simple modification of that of Lemma 3 in \cite{capacityofclassofchannels}. It is based on an application of Feinstein's lemma.
\end{proof}

\begin{lemma}
\label{abweichungmean} 
Let $W_{\mbf g}$ be a fixed channel with $\mbf g \in \mathbb{C}^{N_R\times N_T}$. Let $\bs{T}^n \in \mbb C^{N_T\times n}$ and $\bs{Z}^n \in \mbb C^{N_R\times n}$ be the random input and output sequence, respectively. We further assume that the $\bs{T}_is$ are i.i.d., where each $\bs{T}_i \in \mbb C^{N_T}$ is Gaussian distributed with mean $\bs{0}_{N_T}$  and with a non-singular covariance matrix $\mbf Q.$
Then for any $\delta>0$
\begin{align} 
    \mbb P\left[i_{\mbf g}(\bs{T}^n,\bs{Z}^n)\leq \mbb E\left[i_{\mbf g}(\bs{T}^n,\bs{Z}^n)\right]-n\delta\right] \leq 2^{\left\{-\frac{n N_R}{2\ln(2)}\left[\left(1+\frac{(\ln(2)\delta)^2}{N_{R}^2}\right)^\frac{1}{2}-1\right] \right\}}.  \nonumber 
\end{align}
\end{lemma}
\begin{proof}
Since $(\bs{T}_i,\bs{Z}_i), i=1, \hdots,n,$ are i.i.d., we introduce $(\bs{T},\bs{Z})$ such that $(\bs{T},\bs{Z})$ has the same joint distribution as each of the $(\bs{T}_i,\bs{Z}_i).$
Now 
\begin{align}
    \mbb E[i_{\mbf g}(\bs{T}^n, \bs{Z}^n)]&=n \mbb E\left[i_{\mbf g}(\bs{T},\bs{Z})\right] \nonumber \\
    &=n \log\det\left(\mbf I_{N_{R}}+\frac{1}{\sigma^2}\mbf g\mbf Q \mbf g^H\right). \nonumber
\end{align}
Let
\begin{align}
    \mbf \Theta=\mbf g \mbf Q \mbf g^{H}+\sigma^2\mbf I_{N_R} \nonumber
\end{align}
be the covariance matrix of $\bs{Z}.$
Here, $\mbf \Theta$ is positive definite and therefore non-singular. 
We further define
\begin{align}
   \phi_i= -\frac{1}{\sigma^2}\left(\bs{Z}_i-\mbf g \bs{T}_i\right)^H\left(\bs{Z}_i-\mbf g \bs{T}_i\right)+\bs{Z}_i^H\mbf \Theta^{-1}\bs{Z}_i.
    \nonumber 
\end{align}
Since the $\bs{\phi}_is$ are i.i.d., we define $\bs{\phi}$ to be a random variable with the same distribution as each of the $\bs{\phi}_i$ as follows:
\begin{align}
    \phi=-\frac{1}{\sigma^2}\left(\bs{Z}-\mbf g \bs{T}\right)^H\left(\bs{Z}-\mbf g \bs{T}\right)+\bs{Z}^H\mbf \Theta^{-1}\bs{Z}. \nonumber
\end{align}
Since  $ i_{\mbf g}(\bs{T}_i,\bs{Z}_i)=\log\det\left(\mbf I_{N_{R}}+\frac{1}{\sigma^2}\mbf g\mbf Q \mbf g^H\right)+\frac{\bs{\phi}_i}{\ln(2)}, i=1,\hdots,n,$ it follows that
\begin{align}
    \mbb P\left[i_{\mbf g}(\bs{T}^n,\bs{Z}^n)\leq \mbb E\left[i_{\mbf g}(\bs{T}^n,\bs{Z}^n )\right]-n\delta\right] 
   &=\mbb P\left[ \sum_{i=1}^{n}\frac{\phi_i}{\ln(2)} \leq -n\delta\right] \nonumber \\
   &=\mbb P\left[ -(\ln(2)n\delta+\sum_{i=1}^{n}\phi_i)\geq 0 \right] \nonumber \\
   &\leq \mbb E\left[\exp(-\beta(\ln(2)n\delta+\sum_{i=1}^{n}\bs{\phi}_i))\right] \nonumber \\
   &=\exp(-\beta n \ln(2) \delta)\mbb E\left[\exp(-n\beta\bs{\phi})\right]\quad \forall \beta\geq 0. \nonumber
\end{align}

Let $\zeta(\beta)=\mbb  E\left[ \exp(-\beta\phi)   \right]$ so that
\begin{align}
    \mbb P\left[i_{\mbf g}(\bs{T}^n,\bs{Z}^n)\leq \mbb E\left[i_{\mbf g}(\bs{T}^n,\bs{Z}^n)\right]-n\delta\right]\leq \left(\exp(-\ln(2)\beta\delta)\zeta(\beta)\right)^n.
    \label{probj}
\end{align}
In order to compute $\zeta(\beta),$ we introduce the Gaussian random vector $\bs{W}=[\bs{T},\bs{Z}]^T$ of dimension $N_{T}+N_{R}.$ Since $\bs{T}$ and $\bs{Z}$ have mean zero, $\bs{W}$ has also mean zero and its covariance matrix $\mbf O $ can be written as:
\begin{align}\mbf O=
\begin{pmatrix}
 \begin{matrix}
  \mbf Q & \mbf Q \mbf g^{H} \\
  \mbf g \mbf Q & \mbf \Theta
  \end{matrix}
  \end{pmatrix}. \nonumber
\end{align}
We further define:
\begin{align}
    \mbf \Lambda=
    \begin{pmatrix}
    \begin{matrix}
    \mbf 0 & \mbf 0 \\
  \mbf 0 & \mbf \Theta^{-1}
    \end{matrix}
    \end{pmatrix} \nonumber
\end{align}
and 
\begin{align}
    \mbf \Phi=
    \begin{pmatrix}
    \begin{matrix}
      \frac{1}{\sigma^2}\mbf g^{H}\mbf g & -\frac{1}{\sigma^2}\mbf g^{H} \\
  -\frac{1}{\sigma^2}\mbf g & \frac{1}{\sigma^2}\mbf I_{N_{R}}
    \end{matrix}
    \end{pmatrix}. \nonumber
\end{align}
We can then write
\begin{align}
   \bs{\phi}=\bs{W}^{H}\bs{\Lambda}\bs{W}-\bs{W}^{H}\mbf \Phi\bs{W}. \nonumber
\end{align}
Indeed
\begin{align}
    \bs{W}^{H}\mbf \Lambda\bs{W}&=(\bs{T}^H\bs{Z}^H)      \begin{pmatrix}
    \begin{matrix}
    \mbf 0 & \mbf 0 \\
  \mbf 0 & \mbf \Theta^{-1}
    \end{matrix}
    \end{pmatrix}\begin{pmatrix}\begin{matrix}
      \bs{T} \\
  \bs{Z}
    \end{matrix}\end{pmatrix} \nonumber \\
    &=(\mbf 0 \ \bs{Z}^H\mbf \Theta^{-1})\begin{pmatrix}\begin{matrix}
      \bs{T} \\
  \bs{Z}
    \end{matrix}\end{pmatrix} \nonumber \\
    &=\bs{Z}^H\mbf \Theta^{-1}\bs{Z}, \nonumber
\end{align}
and
\begin{align}
    \bs{W}^{H}\mbf \Phi\bs{W}=
    &(\bs{T}^H\bs{Z}^H)        \begin{pmatrix}
    \begin{matrix}
      \frac{1}{\sigma^2}\mbf g^{H}\mbf g & -\frac{1}{\sigma^2}\mbf g^{H} \\
  -\frac{1}{\sigma^2}\mbf g & \frac{1}{\sigma^2}\mbf I_{N_{R}}
    \end{matrix}
    \end{pmatrix}\begin{pmatrix}\begin{matrix}
      \bs{T} \\
  \bs{Z}
    \end{matrix}\end{pmatrix}  \nonumber \\
    &=(\frac{1}{\sigma^2} \bs{T}^{H} \mbf g^{H} \mbf g -\frac{1}{\sigma^2}\bs{Z}^{H}\mbf g \ , \ -\frac{1}{\sigma^2}\bs{T}^{H}\mbf g^{H}+\frac{1}{\sigma^2}\bs{Z}^{H})\begin{pmatrix}\begin{matrix}
      \bs{T} \\
  \bs{Z}
    \end{matrix}\end{pmatrix} \nonumber \\
    &=\frac{1}{\sigma^2}\left[\bs{T}^{H}\mbf g^{H} \mbf g \bs{T}-\bs{Z}^{H}\mbf g \bs{T}-\bs{T}^{H}\mbf g^{H} \bs{Z}+\bs{Z}^{H}\bs{Z} \right] \nonumber \\
    &=\frac{1}{\sigma^2}(\bs{Z}-\mbf g \bs{T})^{H}(\bs{Z}-\mbf g \bs{T}). \nonumber
\end{align}
Since $\mbf Q$ is non-singular and  the matrix $\mbf\Theta-\mbf g \mbf Q \mbf Q^{-1}\mbf Q\mbf g^{H}=\sigma^2\mbf I_{N_R}$ is non-singular, it follows by applying the inversion rule for the block-matrix $\mbf O$ \cite{matrixinversion}  that 
\begin{align}\mbf O^{-1}=
\begin{pmatrix}
 \begin{matrix}
  \mbf Q^{-1}+\frac{1}{\sigma^2}\mbf g^{H}\mbf g & -\frac{1}{\sigma^2}\mbf g^{H} \\
  -\frac{1}{\sigma^2}\mbf g  & \frac{1}{\sigma^2}\mbf I_{N_{R}}
  \end{matrix}
  \end{pmatrix}. \nonumber
\end{align}

Now, let $\mbf  M(\beta)=\mbf O^{-1}+\beta(\mbf \Lambda-\mbf \Phi)  \in \mbb C^{(N_{T}+N_{R})\times (N_{T}+N_{R})}.$ It follows that
\begin{align}
    \zeta(\beta)&= \mbb E\left[\exp(-\beta\bs{\phi})  \right]\nonumber \\
    &=\frac{\int\exp(-\bs{w}^H\mbf O^{-1}\bs{w})\times
    \exp\left[-\beta(\bs{w}^H\bs{\Lambda}\bs{w}-\bs{w}^H\bs{\Phi}\bs{w})\right]d\bs{w}}{\pi^{N_{T}+N_{R}}\det(\mbf O)}  \nonumber \\
    &=\frac{1}{\pi^{N_{T}+N_{R}}\det(\mbf O)}\int \exp(-\bs{w}^H\mbf M(\beta) \bs{w}) d\bs{w} \nonumber \\
    &=\det(\mbf M(\beta) \mbf O)^{-1}, \nonumber
\end{align}
where the integral is a $(N_{T}+N_{R})$-fold integral over $\mbb C^{N_{T}+N_{R}}.$
 Here, $\mbf M(\beta)$ is positive definite  for $0\leq\beta<\beta_0$ for some $\beta_0\geq 1.$
Indeed, it holds that
\begin{align}\mbf M(\beta)=
\begin{pmatrix}
 \begin{matrix}
  \mbf Q^{-1}+\frac{1}{\sigma^2}(1-\beta)\mbf g^{H}\mbf g & -\frac{1}{\sigma^2}(1-\beta)\mbf g^{H} \\
  -\frac{1}{\sigma^2}(1-\beta)\mbf g  & \beta \mbf \Theta^{-1}+\frac{1}{\sigma^2}(1-\beta)\mbf I_{N_{R}}
  \end{matrix}
  \end{pmatrix}, \quad \beta\geq 0.
  \nonumber \end{align}
  Notice that 
  $$ \mbf M(\beta)=\mbf \beta M(1)+(1-\beta)\mbf M(0) $$
  and that $\mbf M(0)$ and $\mbf M(1)$ are both positive definite. From the convexity of the set of positive definite Hermitian matrices, it follows for all $\beta\in(0,1)$ that
  $\beta \mbf M(1)+(1-\beta)\mbf M(0)$ is positive definite. This proves that $\mbf M(\beta)$ is positive definite for $0\leq \beta\leq 1.$

By substituting $\mbf \Lambda,$ $\mbf \Phi$ and $\mbf O$, and by using the fact that
$\mbf \Theta=\mbf g \mbf Q \mbf g^H+\sigma^2 \mbf I_{N_{R}},$ we obtain
\begin{align}
        \mbf M(\beta)\mbf O
        &= \mbf I_{N_{T}+ N_{R}}+\beta(\mbf \Lambda-\mbf \Phi)\mbf O \nonumber \\
        &= \mbf I_{N_{T}+ N_{R}}+\beta  \begin{pmatrix}
    \begin{matrix}
       \frac{-1}{\sigma^2}\mbf g^{H} \mbf g& \frac{1}{\sigma^2}\mbf g^{H} \\
  \frac{1}{\sigma^2}\mbf g&\mbf \Theta^{-1}-\frac{1}{\sigma^2}\mbf I_{N_{R}}
    \end{matrix}
    \end{pmatrix}\begin{pmatrix}
 \begin{matrix}
  \mbf Q & \mbf Q \mbf g^{H} \\
  \mbf g \mbf Q & \mbf \Theta
  \end{matrix}
  \end{pmatrix}
  \nonumber \\
  &=\mbf I_{N_{T}+ N_{R}} +\beta \begin{pmatrix}
    \begin{matrix}
       \frac{-1}{\sigma^2}\mbf g^{H} \mbf g \mbf Q+\frac{1}{\sigma^2}\mbf g^{H}\mbf g \mbf Q& -\frac{1}{\sigma^2}\mbf g^{H}\mbf g \mbf Q \mbf g^{H}+\frac{1}{\sigma^2}\mbf g^{H}\mbf \Theta \\
  \frac{1}{\sigma^2}\mbf g\mbf Q+\mbf \Theta^{-1}\mbf g\mbf Q-\frac{1}{\sigma^2}\mbf g \bf Q&\frac{1}{\sigma^2}\mbf g \mbf Q \mbf g^{H}+\mbf I_{N_{R}}-\frac{1}{\sigma^2}\mbf \Theta
    \end{matrix}
    \end{pmatrix} \nonumber \\
    &=\begin{pmatrix}
    \begin{matrix}
      \mbf I_{N_{T}} & \beta \mbf g^{H} \\
  \beta \mbf \Theta^{-1} \mbf g \mbf Q & \mbf I_{N_{R}}
    \end{matrix}
    \end{pmatrix}. \nonumber
\end{align}
As a result, we obtain using the determinant rule for block-matrices
\begin{align}
    \det(\mbf M(\beta) \mbf O)
    &=  \det(\mbf I_{N_{R}}-\beta^2 \mbf \Theta^{-1}\mbf g \mbf Q \mbf g^{H}) \nonumber \\
    &=\det(\mbf \Theta^{-1})\det(\mbf \Theta-\beta^2 \mbf g \mbf Q \mbf g^{H}) \nonumber \\
    &=\det(\mbf \Theta^{-1})\det(\sigma^2\mbf I_{N_{R}}+(1-\beta^2) \mbf g \mbf Q \mbf g^{H}) \nonumber \\
    &=\sigma^{2N_{R}}\frac{\det(\mbf I_{N_{R}}+(1-\beta^2) \frac{1}{\sigma^2}\mbf g \mbf Q \mbf g^{H})}{\det(\mbf \Theta)}, \nonumber 
\end{align}
where
\begin{align}
    \det(\mbf \Theta)=\sigma^{2N_{R}}\det(\mbf I_{N_{R}}+ \frac{1}{\sigma^2}\mbf g \mbf Q \mbf g^{H}). \nonumber
\end{align}
We define $\lambda_1,\hdots, \lambda_{N_{R}}$ to be the eigenvalues of the positive semi-definite matrix $\frac{1}{\sigma^2}\mbf g \mbf Q \mbf g^{H}.$
Then it holds that
\begin{align}
    \det(\mbf I_{N_{R}}+ \frac{1}{\sigma^2}\mbf g \mbf Q \mbf g^{H})=\prod_{i=1}^{N_{R}} (1+\lambda_i) \nonumber
\end{align}
and
\begin{align}
    \det(\mbf I_{N_{R}}+(1-\beta^2) \frac{1}{\sigma^2}\mbf g \mbf Q \mbf g^{H}) =\prod_{i=1}^{N_{R}} (1+(1-\beta^2)\lambda_i). \nonumber
\end{align}
This yields
\begin{align}
    \det(\mbf M(\beta)\mbf O)=\prod_{i=1}^{N_{R}} \frac{1+(1-\beta^2)\lambda_i}{1+\lambda_i}
    &=\prod_{i=1}^{N_{R}}\left(1-\beta^2\frac{\lambda_i}{1+\lambda_i}\right) \nonumber
\end{align}
such that 
\begin{align}
    \zeta(\beta)= \prod_{i=1}^{N_{R}}\left(1-\beta^2\frac{\lambda_i}{1+\lambda_i}\right)^{-1}, \quad 0\leq \beta <\beta_0. \nonumber
\end{align}
Then, we have
\begin{align}
    \zeta(\beta)\leq \frac{1}{(1-\beta^2)^{N_{R}}}, \quad 0 \leq \beta <\beta_0 \nonumber
\end{align}
and hence
\begin{align}
    \left(\exp(-\ln(2)\delta\beta)\zeta(\beta)\right)^{\frac{1}{N_R}} \leq \frac{\exp(-\frac{\ln(2)\delta\beta}{N_R})}{1-\beta^2}, \quad 0\leq \beta <\beta_0. \nonumber
\end{align}
Now if we put
\begin{align}
    \beta=\frac{N_R}{\ln(2)\delta}\left[ -1+\left(1+\frac{(\ln(2)\delta)^2}{N_R^2} \right)^{\frac{1}{2}}\right], \nonumber
\end{align}
it follows that $0<\beta<1$ and it holds that
\begin{align}
    \exp(-\frac{\ln(2)\delta\beta}{N_{R}})=\exp\left(-´\left[-1+\left(1+\frac{(\ln(2)\delta)^2}{N_R^2}     \right)^{1/2}     \right] \right)   \nonumber
\end{align}
and that
\begin{align}
\frac{1}{1-\beta^2}&=\frac{1}{1-\left(\frac{N_R}{\ln(2)\delta}\right)^2\left[1-2\sqrt{1+\left(\frac{\ln(2)\delta}{N_R}\right)^2}+1+\left(\frac{\ln(2)\delta}{N_R}\right)^2\right]} \nonumber \\
&=\frac{1}{2}\frac{\left(\frac{\ln(2)\delta}{N_R}\right)^2}{\sqrt{1+\left(\frac{\ln(2)\delta}{N_R}\right)^2}-1} \nonumber \\
&=\frac{1}{2}\left(\sqrt{1+\left(\frac{\ln(2)\delta}{N_{R}}\right)^2}+1\right) \nonumber \\
&=\left(1+\frac{1}{2}\left[-1+\left(1+\frac{(  \ln(2)\delta)^2}{N_R^2}     \right)^{1/2}     \right]        \right). \nonumber
\end{align}
This implies that
\begin{align}
    (1-\beta^2)^{-1} \exp\left(-\frac{\ln(2)\delta\beta}{N_R}\right)=\left(1+\frac{1}{2}\left[-1+\left(1+\frac{(\ln(2)\delta)^2}{N_{R}^2}     \right)^{1/2}     \right]        \right) \exp\left(-\left[-1+\left(1+\frac{(\ln(2)\delta)^2}{N_{R}^2}     \right)^{1/2}     \right]  \right).    \nonumber 
\end{align}
Since $(1+\frac{1}{2}x)\exp(-x)\leq \exp(-\frac{x}{2}) \ \text{for} \ x\geq 0,$ we have
\begin{align}
\exp(-\ln(2)\delta\beta) \zeta(\beta) \leq \exp\left(-\frac{N_R}{2}\left[\left(1+\frac{(\ln(2)\delta)^2}{N_R^2}     \right)^{1/2}-1     \right]\right). \nonumber
\end{align}
It follows from \eqref{probj} that
\begin{align}
   \mbb P\left[i_{\mbf g}(\bs{T}^n,\bs{Z}^n)\leq \mbb E\left[i_{\mbf g}(\bs{T}^n,\bs{Z}^n)\right]-n\delta\right] &\leq \exp\left(-\frac{nN_R}{2}\left[\left(1+\frac{(\ln(2)\delta)^2}{N_R^2}     \right)^{1/2}-1     \right]\right) \nonumber \\
   &= 2^{\left\{-\frac{n N_R}{2\ln(2)}\left[\left(1+\frac{(\ln(2)\delta)^2}{N_{R}^2}\right)^\frac{1}{2}-1\right] \right\}}. \nonumber 
\end{align}
This completes the proof of the lemma.
\end{proof}
\begin{lemma}
\label{upperboundprobsum}
Let $\bs{X}_i,$ $i=1, \hdots,n$ be i.i.d. $N$-dimensional complex Gaussian random vectors with mean $\bs{0}_N$ and covariance matrix $\mbf O$ whose trace is smaller than or equal to $M$.
Then, for any $\delta>0$
\begin{align}
    \mbb P\left[\sum_{i=1}^{n} \lVert \bs{X}_i\rVert^2\geq n(M+\delta)\right]\leq \left[(1+\frac{\delta}{M})2^{-\frac{\delta}{\ln(2)M}}\right]^{n}, \nonumber
\end{align}
where 
\begin{align}
    \lVert \bs{X}_i \rVert^2=\sum_{j=1}^{N} |X_i^j|^2 \nonumber
\end{align}
and
\begin{align}
    \bs{X}_i=(X_i^1,\hdots,X_i^N)^{T}. \nonumber
\end{align}
\end{lemma}
\begin{proof}
Let $\bs{X}$ be a random vector with the same distribution as each of the $\bs{X}_i$. Then
\begin{align}
\mbb P\left[ \sum_{i=1}^{n} \lVert \bs{X}_i\rVert^2 \geq n(M+\delta)      \right] 
&=\mbb P\left[ \sum_{i=1}^{n}\lVert \bs{X}_i^2\rVert-n(M+\delta)\geq 0       \right] \nonumber \\
    &\leq \mbb E \left[\exp\left(\beta\left(\sum_{i=1}^{n}\lVert \bs{X}_{i}\rVert^2-n(M+\delta\right)    \right)            \right]\nonumber \\
    &=\left[\exp(-[M+\delta]\beta)\mbb E\left[ \exp\left(\beta \lVert \bs{X} \rVert^2\right)\right]\right]^n,
    \label{eqprob}
\end{align}
where we used that the $\bs{X}_is$ are i.i.d..
By a standard calculation which follows below, one can show that
\begin{align}
    \mbb E\left[ \exp(\beta \lVert \bs{X}\rVert^2    \right]&=\mbb E \left[\exp(\beta\bs{X}^H\bs{X})           \right] \nonumber \\
    &=\prod_{j=1}^{N}(1-\beta\mu_j)^{-1} \quad \beta<\beta_0, \label{standardcalculation}
\end{align}
where $\mu_1,\hdots,\mu_{N}$ are the eigenvalues of $\mbf O$, and for $\beta_0=\frac{1}{M}\leq \frac{1}{\mu_1+\hdots+\mu_N}<\underset{j\in \{1,\hdots,N\}}{\min}\frac{1}{\mu_j}$  so that all the factors are positive, whether $\mbf O$ is non-singular or singular.
To prove \eqref{standardcalculation}, we let $r$ be the rank of $\mbf O.$ It holds that $r\leq N$. We make use of the spectral decomposition theorem to express $\mbf O$ as $\mbf S_{\mbf O}^{\star} \Lambda^{\star}{S_{\mbf O}^{\star}}^{H} $, where $\Lambda^{\star}$ is a diagonal matrix whose first $r$ diagonal elements are positive and where the remaining diagonal elements are equal to zero.
Next, we let $\mbf V^{\star}=\mbf S_{\mbf O}^{\star} {\Lambda^{\star}}^{\frac{1}{2}}$ and remove the $N-r$ last columns of $\mbf V^{\star}$, which are null vectors to obtain the matrix $\mbf V.$ Then, it can be verified that $\mbf O=\mbf V \mbf V^{H}.$
We can write $\bs{X}=\mbf  V \bs{U}^\star$
where $\bs{U}^\star \sim\mathcal{N}_{\mbb C}(\bs{0},\mbf I_{r}).$
As a result:
\begin{align}
    \bs{X}^H\bs{X}={(\bs{U}^\star)}^{H}\mbf V^{H}\mbf V \bs{U}^\star. \nonumber
\end{align}
Let $\mbf S$ be a unitary matrix which diagonalizes $\mbf V^{H} \mbf V$ such that $\mbf S^{H} \mbf V^{H} \mbf V \mbf S= \text{Diag}(\mu_1,\hdots,\mu_r)$ with $\mu_1,\hdots,\mu_r$ being the positive eigenvalues of $\mbf O=\mbf V \mbf V^{H}$ in decreasing order, as mentioned above.
One defines $\bs{U}=\mbf S^{H} \bs{U}^{\star}.$ We have
\begin{align}
    \text{cov}(\bs{U})&=\mbf S^{H} \text{cov}(\bs{U}^{\star}) \mbf S \nonumber \\
    &=\mbf S^{H} \mbf S \nonumber \\
    &=\mbf I_{r}. \nonumber
\end{align}
Therefore, it holds that $\bs{U}=(U_1,\hdots,U_r)^{T} \sim \mathcal{N}(\bs{0},\mbf I_r).$
Since $\mbf S$ is unitary, it follows that
\begin{align}
    \bs{X}^{H}\bs{X}&=\left((\mbf S^{H})^{-1}\bs{U}\right)^{H} \mbf V^{H} \mbf V (\mbf S^{H})^{-1} \bs{U} \nonumber \\
    &=\bs{U}^{H}\mbf S^{H} \mbf V^{H} \mbf V \mbf S \bs{U} \nonumber \\
    &=\bs{U}^{H}\text{Diag}(\mu_1,\hdots,\mu_r)\bs{U} \nonumber \\
    &=\sum_{j=1}^{r} \mu_j |U_{j}|^{2}.\nonumber
\end{align}
Then, we have
\begin{align}
    \mbb E\left[ \exp(\beta \lVert \bs{X}\rVert^2)    \right] &= \mbb E \left[\prod_{j=1}^{r}\exp(\frac{1}{2}\beta\mu_j 2|U_j|^2  )             \right] \nonumber \\
    \nonumber \\
    &=\prod_{j=1}^{r} \mbb E \left[ \exp(\frac{1}{2}\beta\mu_j 2|U_j|^2  )       \right] \nonumber \\
    &=\prod_{j=1}^{N} (1-\beta\mu_j)^{-1},             \nonumber 
\end{align}
where we used that all the $U_j$s are independent, that $ \forall j \in \{1,\hdots,r\},  2|U_{j}|^2$ is chi-square distributed with $k=2$ degrees of freedom and with moment generating function equal to $\mbb E\left[\exp(2t|U_j|^2)\right]=(1-2t)^{-k/2}$ for $t<\frac{1}{2}$ and that $\forall j \in \{1,\hdots,r\}$ and for $\beta<\beta_0,$ $\frac{1}{2}\beta\mu_j<\frac{1}{2}$. This completes the proof of \eqref{standardcalculation}.

Now, it holds that
$$ \prod_{i=1}^{N}(1-\beta\mu_i)\geq 1-\beta(\mu_1+\hdots+\mu_{N})\geq 1-\beta M.         $$ This yields
\begin{align}
    &\exp(-(M+\delta)\beta)\mbb E\left[ \exp(\beta \lVert \bs{X}\rVert^2    \right] \leq \frac{\exp(-(M+\delta)\beta)}{1-\beta M}, \nonumber
\end{align}
where $0<\beta<\frac{1}{M}=\beta_0.$
Putting $\beta=\frac{\delta}{M(\delta+M)}<\frac{1}{M}$ yields
\begin{align}
    \exp(-(M+\delta)\beta)\mbb E\left[ \exp(\beta \lVert \bs{X}\rVert^2    \right]&\leq (1+\frac{\delta}{M})\exp(-\frac{\delta}{M})\nonumber \\
    &=(1+\frac{\delta}{M})2^{(-\frac{\delta}{\ln(2)M})},
\nonumber \end{align}
which combined with \eqref{eqprob} proves the lemma.
\end{proof}
\begin{lemma}
\label{existence} 
Let $\epsilon>0$ be fixed arbitrarily. Let $\mc S_a$ be any closed subset of $\mc B_a.$ Then, there exists a non-singular $\mbf Q \in \mc Q_{P}$ such that
\begin{enumerate}
    \item $\mathrm{tr}(\mbf Q)<P$
    \item$\log\det(\mbf I_{N_R}+\mbf g \mbf Q \mbf g^H)
    \geq \underset{\mbf Q \in \mc Q_{P}}{\max}\underset{\mbf g \in \mc S_{a}}{\min}\log\det(\mbf I_{N_{R}}+\frac{1}{\sigma^2}\mbf g \mbf Q \mbf g^H)-\epsilon \quad \text{for all} \ \mbf g \in \mc S_a.$
\end{enumerate}
\end{lemma}
\begin{proof}
Notice first that the set  $\mathcal{L}=\mc S_{a} \times \mc Q_{P}$ is a  compact set, because the conditions on the matrices $\mbf g \in \mc S_{a}\subset \mc B_a$ and on the positive semi-definite matrices $\mbf Q\in \mc Q_{P}$ guarantee that $\mc L$ is bounded and closed in $\mbb C^{N_{R}\times N_{T}} \times \mbb C^{N_{T}\times N_{T}}$. Now the function $f(\mbf g,\mbf Q)=\log\det(\mbf I_{N_{R}}+\frac{1}{\sigma^2}\mbf g \mbf Q \mbf  g^H)$ is uniformly continuous on $\mathcal{L}$ \color{black}. One can find a non-singular $\mbf Q_{0} \in \mc Q_{P}$ such that
\begin{align}
& \log\det(I_{N_R}+\frac{1}{\sigma^2}\mbf g \mbf Q_{0} \mbf g^H) \geq \underset{\mbf Q\in \mc Q_{P}}{\max}\underset{\mbf g \in \mc S_{a}}{\min}  \log\det(I_{N_R}+\frac{1}{\sigma^2}\mbf g \mbf Q \mbf g^H) -\frac{\epsilon}{2} \quad \forall \mbf g \in \mc S_{a}.\nonumber
\end{align}
If $\text{tr}(\mbf Q_0)<P$, the proof is complete. If $\text{tr}(\mbf Q_0)=P$, we can find, by the uniform continuity of $f$ on $\mathcal{L}$, a number $\delta>0$ such that
$\lvert f(\mbf g,\mbf Q)-f(\mbf g,\mbf Q_0)\rvert\leq \frac{\epsilon}{2}$ for all $\mbf g$ if $\lVert  \mbf Q- \mbf Q_0 \rVert \leq \delta$. We can then change $\mbf Q_0$ into a non-singular $\mbf Q_1$ in such a way that $\lVert \mbf Q_{1}-\mbf Q_{0} \rVert \leq \delta$ and
$\text{tr}(\mbf Q_{1})<\text{tr}(\mbf Q_{0})=P$.
 $\mbf Q_{1}$ satisfies the conditions of the lemma. This completes the proof of the lemma.
\end{proof}
\subsubsection{Proof of Theorem \ref{capacityfinitestate}}
\begin{proof}
Now that we proved the lemmas, we fix $R$ to be any positive number strictly less than $\underset{\mbf Q \in \mc Q_{P}}{\max}\underset{\mbf g\in\mc G_{a}}{\min} \log\det(\mathbf{I}_{N_{R}}+\frac{1}{\sigma^2}\mathbf{g}\mathbf{Q}\mathbf{g}^{H})$ and put $2\theta=\underset{\mbf Q \in \mc Q_{P}}{\max}\underset{\mbf g\in\mc G_{a}}{\min} \log\det(\mathbf{I}_{N_{R}}+\frac{1}{\sigma^2}\mathbf{g}\mathbf{Q}\mathbf{g}^{H})-R.$\\~\\
By Lemma \ref{existence}, \color{black} one can find a non-singular $\mbf Q_1 \in \mc Q_{P}$ \color{black} such that $\text{tr}(\mbf Q_{1})=P-\beta, \quad \beta>0$ and
\begin{align}
\mbb E[i_{\mbf g}(\bs{T},\bs{Z})]=\log\det(\mbf I_{N_{R}}+\frac{1}{\sigma^2}\mbf g \mbf Q_{1} \mbf g^{H})\geq R+\frac{\theta}{2} \quad \forall \mbf g \in \mc G_{a},
\label{meanJ}
\end{align}
where $\bs{T}$ and $\bs{Z}$ represent the random input and output of $W_\mbf g,$ respectively, and where $\bs{T}\sim\mc N_{\mbb C}(\mbf 0_{N_T},\mbf Q_1).$
 Let $E_n$ be the set of all input sequences $\bs{t}^n$ satisfying $\frac{1}{n}\sum_{i=1}^{n}\lVert \bs{t}_i\rVert^2\leq P.$   $ \text{For any} \ \mbf g \in \mc G_{a},$  we define 
$\bs{T}_i,i=1, \hdots,n,$ to be the i.i.d. random inputs of $ W_\mbf g,$ each normally distributed with mean $\mbf 0_{N_T}$ and covariance matrix $\mbf Q_1$. Let $\hat{P}=P-\beta$ and $\hat{\beta}=\frac{\beta}{\ln(2)\hat{P}}-\log(1+\frac{\beta}{\hat{P}})>0$.
Then, by Lemma \ref{upperboundprobsum}, it holds that
\begin{align}
    \mbb P \left[\sum_{i=1}^{n}\lVert \bs{T}_i\rVert^2 \geq n(\hat{P}+\beta)  \right] &\leq \left[(1+\frac{\beta}{\hat{P}})2^{(-\frac{\beta}{\ln(2)\hat{P}})}\right]^{n} \nonumber \\
    &=2^{\left(-n\frac{\beta}{\ln(2)\hat{P}}+n\log(1+\frac{\beta}{\hat{P}})\right)} \nonumber \\
    &= 2^{-n\hat{\beta}}.\nonumber
\end{align}

As a result, we have
\begin{align}
    \mbb P\left[ \bs{T}^n \notin E_n\right]&= \mbb P \left[\sum_{i=1}^{n}\lVert \bs{T}_i\rVert^2 > nP \right]\nonumber \\
    &\leq  \mbb P \left[\sum_{i=1}^{n}\lVert \bs{T}_i\rVert^2 \geq n(\hat{P}+\beta)  \right] \nonumber \\
    &\leq 2^{-n\hat{\beta}}.\nonumber
\end{align}

Now define $\tau=\lfloor2^{nR}\rfloor$, $\alpha=n(R+\frac{\theta}{8})$ and $\delta=\frac{n\theta}{8}.$
It follows from Lemma \ref{existenceerror} that there exists a code $\Gamma_n$  for $\mathcal{C}'$ with size $\lVert \Gamma_n \rVert=\tau$ and block-length $n$ such that for all $\mbf g \in \mc G_{a}$ 
\begin{align}
    e(\Gamma_n,\mbf g)
    &\leq |\mc G_{a}|2^{nR}2^{-n(R+\frac{\theta}{8})}+|\mc G_{a}|^22^{-n\frac{\theta}{8}}+|\mc G_{a}|2^{-n\hat{\beta}}+\sum_{\mbf g\in\mc G_{a}} \mbb P\left[i_{\mbf g}(\bs{T}^n,\bs{Z}^n)\leq n(R+\frac{\theta}{4})\right].
    \label{epsilonn}
\end{align}
Since $\mbb E\left[ i_{\mbf g}(\bs{T}^n,\bs{Z}^n)\right]=n\mbb E\left[i_{\mbf g}(\bs{T},\bs{Z})\right], $ it follows from $\eqref{meanJ}$ using Lemma \ref{abweichungmean} that
\begin{align}
    \mbb P\left[i_{\mbf g}(\bs{T}^n,\bs{Z}^n)\leq n(R+\frac{\theta}{4})\right] 
    &=\mbb P\left[i_{\mbf g}(\bs{T}^n,\bs{Z}^n)\leq n(R+\frac{\theta}{2})-n\frac{\theta}{4}\right] \nonumber \\
    &\leq \mbb P\left[i_{\mbf g}(\bs{T}^n,\bs{Z}^n)\leq \mbb E\left[ i_{\mbf g}(\bs{T}^n,\bs{Z}^n)\right]-n\frac{\theta}{4}\right] \nonumber \\
    &\leq 2^{-\frac{nN_{R}}{2\ln(2)}\left[\left(1+\frac{(\ln(2)\theta)^2}{(4N_{R})^2}\right)^{\frac{1}{2}}-1\right]}. \nonumber
\end{align}
Then, it follows using \eqref{epsilonn} that for all $\mbf g \in \mc G_{a}$
\begin{align}
      e(\Gamma_n,\mbf g) &\leq (\lvert \mc G_{a}\rvert+\lvert \mc G_{a} \rvert^2)2^{-\frac{n\theta}{8}}+\lvert \mc G_{a}\rvert 2^{-n\hat{\beta}}  +\lvert \mc G_{a} \rvert 2^{-\frac{nN_{R}}{2\ln(2)}\left[\left(1+\frac{(\ln(2)\theta)^2}{(4N_{R})^2}\right)^{\frac{1}{2}}-1\right]}.\nonumber
\end{align}
The limit of the last upper-bound is equal to 0 as $n\rightarrow \infty.$ Since $R$ is any number strictly less than $\underset{\mbf Q \in \mc Q_{P}}{\max}\underset{\mbf g\in\mc G_{a}}{\min} \log\det(\mathbf{I}_{N_{R}}+\frac{1}{\sigma^2}\mathbf{g}\mathbf{Q}\mathbf{g}^{H}),$  Theorem \ref{capacityfinitestate} is proved.
\end{proof} 
\subsection{\text{Direct Proof of Theorem} \ref{capacitycompoundchannels} for infinite \texorpdfstring{$\mc G_{a}$}{TEXT}}
Now, we proceed with the proof of the direct part of Theorem \ref{capacitycompoundchannels} for infinite $\mc G_a.$
In the proof, we will make use of  Theorem \ref{capacityfinitestate}. We will additionally establish an approximation inequality and a probabilistic bound on the output power. This is done in the following auxiliary lemmas.
\subsubsection{Auxiliary Lemmas}
\begin{lemma}
\label{approximation}
Let $W_{\mbf g}$ and $W_{\hat{\mbf g}} $ be two channels such that $\mbf g,\hat{\mbf g} \in \mc G_{a}$ and let $\bs{t}^n$ be an input
$n$-sequence of vectors $\bs{t}_i$ such that $\frac{1}{n}\sum_{i=1}^{n}\lVert \bs{t}_i\rVert^2\leq P$ and let $\bs{z}^n$ be an output $n$-sequence of vectors $\bs{z}_i$ such that $\frac{1}{n} \sum_{i=1}^{n}\lVert \bs{z}_{i} \rVert^2\leq \rho, \  \rho>0.$
Then, it holds that 
\begin{align}
\frac{W_{\mbf g}(\bs{z}^n|\bs{t}^n)}{ W_{\hat{\mbf g}}(\bs{z}^n|\bs{t}^n)}\leq 2^{\frac{2n}{\ln(2)\sigma^2}\left[\sqrt{P\rho}+aP\right]\lVert \mbf{g}-\hat{\mbf g}\rVert}.\nonumber
\end{align}
\begin{proof}
$\forall i \in \{1,\hdots,n\},$ we have
\begin{align}
&\frac{W_{\mbf g}(\bs{z}_i|\bs{t}_i)}{ W_{\hat{\mbf g}}(\bs{z}_i|\bs{t}_i)} =\exp\left(-\frac{1}{\sigma^2}\left[(\bs{z}_i-\mbf g \bs{t}_i)^{H}(\bs{z}_i-\mbf g \bs{t}_i)-(\bs{z}_i- \hat{\mbf g} \bs{t}_i)^{H}(\bs{z}_i-\hat{\mbf g} \bs{t}_i)\right]\right), \nonumber
\end{align}
where
\begin{align}
   & -´\frac{1}{\sigma^2}\left[ \ \left(\bs{z}_i-\mbf g \bs{t}_i\right)^{H}\left(\bs{z}_i-\mbf g \bs{t}_i\right)-\left(\bs{z}_i- \hat{\mbf g} \bs{t}_i\right)^{H}\left(\bs{z}_i-\hat{\mbf g} \bs{t}_i\right) \right] \nonumber \\
    & \leq´\frac{1}{\sigma^2}\left| \left(\bs{z}_i-\mbf g \bs{t}_i\right)^{H}\left(\bs{z}_i-\mbf g \bs{t}_i\right)-\left(\bs{z}_i- \hat{\mbf g} \bs{t}_i\right)^{H}\left(\bs{z}_i-\hat{\mbf g} \bs{t}_i\right)\right| \nonumber \\
    &=\frac{1}{\sigma^2}\left| \bs{z}_i^{H}\left(\hat{\mbf g}-\mbf g\right)\bs{t}_i+\left(\bs{z}_i^{H}\left(\hat{\mbf g}-\mbf g\right)\bs{t}_i\right)^{H}      +\lVert \mbf g \bs{t}_i \rVert^2-\lVert \hat{\mbf g} \bs{t}_i \rVert^2 \right| \nonumber \\
    &\leq \frac{1}{\sigma^2}  \left|  \bs{z}_i^{H}\left(\hat{\mbf g}-\mbf g\right)\bs{t}_i+\left(\bs{z}_i^{H}\left(\hat{\mbf g}-\mbf g\right)\bs{t}_i\right)^{H} \right| +\frac{1}{\sigma^2} \left| \ \lVert \mbf g \bs{t}_i \rVert^2-\lVert \hat{\mbf g} \bs{t}_i \rVert^2 \right|             \nonumber\\
    &\leq \frac{1}{\sigma^2} \left[ 2\lVert \hat{\mbf g}-\mbf g \rVert \lVert \bs{t}_i\rVert  \lVert \bs{z}_i \rVert  +\left| \lVert \mbf g \bs{t}_i \rVert^2-\lVert \hat{\mbf g} \bs{t}_i \rVert^2  \right|   \right] \nonumber \\
    &=  \frac{2}{\sigma^2} \lVert \hat{\mbf g}-\mbf g \rVert \lVert \bs{t}_i\rVert  \lVert \bs{z}_i \rVert  +\frac{1}{\sigma^2} \left|\lVert \mbf g \bs{t}_i \rVert-\lVert \hat{\mbf g} \bs{t}_i \rVert  \right|\left(\lVert \mbf g \bs{t}_i \rVert+\lVert \hat{\mbf g} \bs{t}_i \rVert \right) \nonumber \\
    &\leq\frac{1}{\sigma^2}\left[ 2\lVert \hat{\mbf g}-\mbf g \rVert  \lVert \bs{t}_i\rVert  \lVert \bs{z}_i \rVert+\lVert\left(\mbf g - \mbf {\hat{g}}\right)\bs{t}_i\rVert  \left(\lVert \mbf g \bs{t}_i \rVert+\lVert \hat{\mbf g} \bs{t}_i \rVert\right)  \right]               \nonumber \\
    &\leq \frac{1}{\sigma^2}\left[ 2\lVert \hat{\mbf g}-\mbf g \rVert \lVert \bs{t}_i\rVert  \lVert \bs{z}_i \rVert+\lVert\left(\mbf g - \mbf {\hat{g}}\right)\bs{t}_i\rVert  \left(\lVert \mbf g \rVert \lVert \bs{t}_i \rVert+\lVert \hat{\mbf g} \rVert \lVert \bs{t}_i \rVert\right)  \right]               \nonumber \\
    & \leq  \frac{1}{\sigma^2} \left[  2\lVert \hat{\mbf g}-\mbf g \rVert  \lVert \bs{t}_i\rVert  \lVert \bs{z}_i \rVert +2 a \lVert \bs{t}_i \rVert  \lVert \hat{\mbf g} -\mbf g  \rVert  \lVert \bs{t}_i \rVert \right] \nonumber \\
    &=\frac{2}{\sigma^2}\lVert \hat{\mbf g} - \mbf g  \rVert \left[ \lVert \bs{t}_i\rVert \lVert \bs{z}_i \rVert + a\lVert \bs{t}_i \rVert^2   \right],\nonumber
\end{align}
where we used that $\lVert \mbf g \rVert \leq a$ and $\lVert \hat{\mbf g} \rVert \leq a$  for $\mbf g,\hat{\mbf g} \in \mc G_{a}\subset\mc B_a.$ 

Now
\begin{align}
    \frac{W_{\mbf g}(\bs{z}^n|\bs{t}^n)}{ W_{\hat{\mbf g}}(\bs{z}^n|\bs{t}^n)}&\overset{(a)}{=}\prod_{i=1}^{n}\frac{W_{\mbf g}(\bs{z}_i|\bs{t}_i)}{ W_{\hat{\mbf g}}(\bs{z}_i|\bs{t}_i)} \nonumber\\
    &\leq\exp\left(\sum_{i=1}^{n}\frac{2}{\sigma^2}\lVert \hat{\mbf g} - \mbf g  \rVert\left[  \lVert \bs{t}_i\rVert \lVert \bs{z}_i \rVert + a\lVert \bs{t}_i \rVert^2        \right]\right) \nonumber \\
    &=\exp\left(\frac{2}{\sigma^2}\lVert \hat{\mbf g} - \mbf g  \rVert\left[ \sum_{i=1}^{n} \lVert \bs{t}_i\rVert \lVert \bs{z}_i \rVert + a \sum_{i=1}^{n}\lVert \bs{t}_i \rVert^2        \right]\right) \nonumber \\
    &\overset{(b)}{\leq} \exp\left(\frac{2}{\sigma^2}\lVert \hat{\mbf g} - \mbf g  \rVert\left[ \sqrt{\sum_{i=1}^{n} \lVert \bs{t}_i\rVert^2} \sqrt{\sum_{i=1}^{n}\lVert \bs{z}_i \rVert^2} + a \sum_{i=1}^{n}\lVert \bs{t}_i \rVert^2        \right]          \right) \nonumber \\
    &\overset{(c)}{\leq} \exp\left(\frac{2n}{\sigma^2}\left[\sqrt{P\rho}+aP\right]\lVert \mbf{g}-\hat{\mbf g}\rVert\right),\nonumber
    & \nonumber \\&= 2^{\frac{2n}{\ln(2)\sigma^2}\left[\sqrt{P\rho}+aP\right]\lVert \mbf{g}-\hat{\mbf g}\rVert},
\nonumber \end{align}
where $(a)$ follows because the channels $W_\mbf g$ and $W_{\hat{\mbf g}}$ are memoryless, $(b)$ follows from Cauchy-Schwarz's inequality and $(c)$ follows because we require that $\frac{1}{n}\sum_{i=1}^{n}\lVert \bs{t}_i\rVert^2\leq P$ and that $\frac{1}{n}\sum_{i=1}^{n}\lVert \bs{z}_i\rVert^2\leq \rho.$ This completes the proof of the lemma.
\end{proof}
\end{lemma}
\begin{lemma} 
\label{lemmadelta}
Let $\mbf g \in \mc G_{a}.$ Let $\bs{t}^n=(\bs{t}_1,\hdots,\bs{t}_n)$ be any $n$-input sequence of $W_\mbf g$ satisfying $\frac{1}{n}\sum_{i=1}^{n} \lVert \bs{t}_i \rVert^2 \leq P.$ Let $\bs{z}^n=(\bs{z}_1,\hdots,\bs{z}_n)$ be the $n$-output sequence. It holds that
\begin{align}
W_{\mbf g}\left( \sum_{i=1}^{n}\lVert \bs{z}_i\rVert^2 \geq n(2a^2P+2N_{R}\sigma^2+2)|\bs{t}^n\right) \leq\left[\left(1+\frac{1}{\sigma^2 N_{R}}\right) 2^{-\frac{1}{\ln(2)\sigma^2N_{R}}}\right]^{n}.\nonumber
\end{align}
\begin{proof}
We have
\begin{align}
\sum_{i=1}^{n} \lVert \bs{z}_i \rVert^2 
    &=\sum_{i=1}^{n} \lVert \mbf g \bs{t}_i+\bs{\xi}_i \rVert^2 \nonumber \\
    &\leq 2\sum_{i=1}^{n} \left( \lVert \bs{\xi}_i \rVert^2+\lVert \mbf g \bs{t}_i \rVert^2         \right) \nonumber \\
    &\leq 2\sum_{i=1}^{n} \left( \lVert \bs{\xi}_i \rVert^2+\lVert \mbf g \rVert^2 \lVert \bs{t}_i \rVert^2         \right) \nonumber \\
    &\leq 2 \sum_{i=1}^{n}\lVert \bs{\xi}_i \rVert^2 +2a^2 n P.\nonumber
\end{align}
Hence,
\begin{align}
    W_{\mbf g}\left( \sum_{i=1}^{n} \lVert \bs{z}_i \rVert^2 \geq n(2a^2P+2N_{R}\sigma^2+2)|\bs{t}^n\right)&\leq \mbb P \left[ 2  \sum_{i=1}^{n} \lVert \bs{\xi}_i \rVert^2 +2a^2 n P \geq n(2a^2P+2N_{R}\sigma^2+2)\right]\nonumber \\
    &= \mbb P\left[ \sum_{i=1}^{n}\lVert \bs{\xi}_i \rVert^2 \geq n(N_{R}\sigma^2+1)\right]
    \nonumber \\
    &= \mbb P\left[ \sum_{i=1}^{n}\lVert \bs{\xi}_i \rVert^2 \geq n(\text{tr}(\sigma^2\mbf{I}_{N_{R}})+1)\right] \nonumber \\
    &\leq \left[\left(1+\frac{1}{\sigma^2 N_{R}}\right) 2^{-\frac{1}{\ln(2)\sigma^2N_{R}}}\right]^{n},
\nonumber \end{align}
where we used Lemma \ref{upperboundprobsum} in the last step. This completes the proof of the lemma.
\end{proof}
\end{lemma}
\subsubsection{Direct Proof of Theorem \ref{capacitycompoundchannels}}
Now that we proved the lemmas, we fix $R$ to be any positive number strictly less than $\underset{\mbf Q \in \mc Q_{P}}{\max}\underset{\mbf g \in \mc G_{a}}{\min} \log\det(\mbf I_{N_{R}}+\frac{1}{\sigma^2}\mbf g \mbf Q \mbf g^H)$ and put $2\theta=\underset{\mbf Q \in \mc Q_{P}}{\max}\underset{\mbf g \in \mc G_{a}}{\min} \log\det(\mbf I_{N_{R}}+\frac{1}{\sigma^2}\mbf g \mbf Q \mbf g^H)-R$.
By Lemma \ref{existence}, one can find a non-singular $\mbf Q_1 \in \mc Q_{P}$ such that $\text{tr}(\mbf Q_1)=P-\beta,\ \beta>0$, and 
\begin{align}
 \mbb E\left[i_{\mbf g} \left(\bs{T},\bs{Z}\right)\right]=\log\det(\mbf I_{N_{R}}+\frac{1}{\sigma^2}\mbf g \mbf Q_{1} \mbf g) \geq R+\theta \quad \forall \mbf g \in \mc G_{a},
 \label{inequalityrate}
\end{align} 
with $\bs{T}$ and $\bs{Z}$ being the random input and output of $W_\mbf g,$ respectively, where $\bs{T} \sim \mathcal{N}_{\mbb C}\left(\bs{0}_{N_{T}},\mbf Q_1\right).$
We now pick a finite subset $\mc G_{a}^{\prime}$ of $\mc G_{a}$ such that for every $\mbf g \in \mc G_{a}$, there is a $\hat{\mbf g} \in \mc G_{a}^{\prime}$ satisfying $\lVert \mbf g - \hat{\mbf g}\rVert \leq \nu.$
This can be done because the set $\mc G_{a}$ is compact. By inequality \eqref{inequalityrate} and since 
\begin{align}
\underset{\mbf Q \in \mc Q_{P}}{\max}\underset{\mbf g \in \mc G_{a}^{\prime}}{\min} \log\det(\mbf I_{N_{R}}+\frac{1}{\sigma^2}\mbf g \mbf Q \mbf g^H) \geq \underset{\mbf Q \in \mc Q_{P}}{\max}\underset{\mbf g \in \mc G_{a}}{\min} \log\det(\mbf I_{N_{R}}+\frac{1}{\sigma^2}\mbf g \mbf Q \mbf g^H), \nonumber
\end{align}
it follows that
$$\underset{\mbf Q \in \mc Q_{P}}{\max}\underset{\mbf g \in \mc G_{a}^{\prime}}{\min} \log\det(\mbf I_{N_{R}}+\frac{1}{\sigma^2}\mbf g \mbf Q \mbf g^H)\geq R+\theta.
$$
Hence, the calculations of Theorem \ref{capacityfinitestate}  imply that there exists a code $\Gamma_n$ for $\mathcal{C}'$ with block-length $n$, size $\lVert \Gamma_n \rVert=\lfloor 2^{nR} \rfloor$ such that the codewords $\bs{t}^n=(\bs{t}_1,\hdots,\bs{t}_n)$ satisfy $\frac{1}{n}\sum_{i=1}^{n}\lVert \bs{t}_i \rVert^2\leq P$ and such that for all $\hat{\mbf g} \in \mc G'_{a}$
    \begin{align}
     e(\Gamma_n,\hat{\mbf g}) & \leq (|\mc G_{a}^{\prime}|+|\mc G_{a}^{\prime}|^2)2^{-\frac{n\theta}{8}}+|\mc G_{a}^{\prime}|2^{-n\hat{\beta}} +\lvert \mc G_{a}^{\prime} \rvert 2^{-\frac{nN_{R}}{2\ln(2)}\left[\left(1+\frac{(\ln(2)\theta)^2}{(4N_{R})^2}\right)^{\frac{1}{2}}-1\right]},        \label{epsilonstrich}
    \end{align}
    where $\hat{\beta}=\frac{\beta}{\ln(2)(P-\beta)}-\log(1+\frac{\beta}{P-\beta})$ and where $\beta$ is independent of $n$.

We now consider the use of codewords and decoding sets belonging to the  code $\Gamma_n$ for $\mathcal{C}'$ with the larger compound channel $\mathcal{C}$.
Let $\mbf g \in \mc G_{a}$ and $\hat{\mbf g} \in \mc G_{a}^{\prime}$ such that $\lVert \mbf g -\hat{\mbf g} \rVert \leq \nu.$
Let $\bs{t}^n$ be \textit{any} codeword of $\Gamma_n$ and $B$ the corresponding decoding set. Let $F=\{\bs{z}^n=(\bs{z}_1,\hdots,\bs{z}_n):\frac{1}{n}\sum_{i=1}^{n}\lVert \bs{z}_i \rVert^2\leq \rho\},$ where
$\rho=2a^2P+2N_{R}\sigma^2+2.$
Then 
 \begin{align}
     W_{\mbf g}(B^{c}|\bs{t}^n) &=W_{\mbf g}\left(\left(B^{c}\cap F\right) \cup \left(B^c\cap F^c\right)|\bs{t}^n\right) \nonumber \\
     &\leq W_{\mbf g}(B^{c}\cap F|\bs{t}^n)+W_{\mbf g}(F^c|\bs{t}^n).\nonumber
 \end{align}
 By Lemma \ref{lemmadelta}, it holds that
 \begin{align}
     W_{\mbf g}(F^{c}|\bs{t}^n)&\leq \left[ \left(1+\frac{1}{N_{R}\sigma^2}\right)2^{-\frac{1}{\ln(2)N_{R}\sigma^2}}   \right]^{n} \nonumber \\
     &=2^{-n \left( \frac{1}{\ln(2)N_{R}\sigma^2}-\log\left(1+\frac{1}{N_{R}\sigma^2} \right)\right)}.\nonumber
 \end{align}
By Lemma \ref{approximation}, it holds that
\begin{align}
    &W_{\mbf g}(B^{c}\cap F|\bs{t}^n) \leq 2^{\frac{2n}{\ln(2)\sigma^2}\left[\sqrt{P\rho}+aP\right] \nu }  W_{\hat{\mbf g}}(B^{c}\cap F|\bs{t}^n). \nonumber
\end{align}
Now
\begin{align}
    W_{\hat{\mbf g}}(B^c\cap F|\bs{t}^n)\leq W_{\hat{\mbf g}}(B^c|\bs{t}^n)\leq e(\Gamma_n,\hat{\mbf g}).\nonumber
\end{align}
This implies using \eqref{epsilonstrich} that for all $\mbf g \in \mc G_{a}$
\begin{align}
 \ e(\Gamma_n,\mbf g) &\leq 2^{-n \left( \frac{1}{\ln(2)N_{R}\sigma^2}-\log\left(1+\frac{1}{N_{R}\sigma^2} \right)\right)}+ (|\mc G_{a}^{\prime}|+|\mc G_{a}^{\prime}|^2) 
    2^{-n\left(\frac{\theta}{8}-\frac{2}{\ln(2)\sigma^2}\left[\sqrt{P\rho}+aP\right] \nu\right)}\nonumber \\
    &\quad+|\mc G_{a}^{\prime}|2^{-n\left(\hat{\beta}-\frac{2}{\ln(2)\sigma^2}\left[\sqrt{P\rho}+aP\right] \nu \right)}  +|\mc G_{a}^{\prime}|2^{-n\left[ c_1 -c_{2}\nu       \right]},
    \label{exponentials}
\end{align}
where
\begin{align}
    c_1=\frac{N_{R}}{2\ln(2)}\left[\left(1+\frac{(\ln(2)\theta)^2}{(4N_{R})^2}\right)^{\frac{1}{2}}-1\right] \nonumber
\end{align}
and
\begin{align}
   c_2= \frac{2}{\ln(2)\sigma^2}\left[\sqrt{P\rho}+aP\right].\nonumber 
\end{align}
The exponentials in $\eqref{exponentials}$ are all of the form $2^{-n(K_{1}-K_2\nu)}$ where $K_1$ and $K_2$ do not depend on $n$ and where $K_1$ is positive and $K_2$ is non-negative.
For $\nu$ sufficiently small, it holds that $K_1-K_2\nu>0,$ which yields 
$\underset{n\rightarrow\infty}{\lim} e(\Gamma_n,\mbf g)=0. $
This proves that $\underset{\mbf Q \in \mc Q_{P}}{\max}\underset{\mbf g \in \mc G_{a}}{\min} \log\det(\mbf I_{N_R}+\frac{1}{\sigma^2}\mbf g \mbf Q \mbf g^H)$ is an achievable rate for $\mathcal{C}.$ This completes the direct proof of Theorem \ref{capacitycompoundchannels}.
\subsection{Converse Proof of Theorem \ref{capacitycompoundchannels}}
Let $R$ be any achievable rate for $\mc C=\{W_\mbf g: \mbf g \in \mc G_{a}\}.$  So, for every $\theta,\delta>0,$ there exists a code sequence $(\Gamma_n)_{n=1}^\infty$  such that for all $\mbf g \in \mc G_{a}$
    \[
        \frac{\log\lVert \Gamma_n\rVert}{n}\geq R-\delta
    \]
    and 
  \begin{align}
      e(\Gamma_n,\mbf g)\leq\theta, \label{maxerror}
    \end{align}
    for sufficiently large $n$.
Notice that from \eqref{maxerror}, it follows that the average error
probability is also bounded from above by $\theta.$ The uniformly-distributed message is modeled by $W$ and the random decoded message is modeled by $\hat{W}.$  The set of messages is denoted by $\mathcal{W}.$    
For any $\mbf g \in \mc G_{a},$ the uniformly-distributed message $W$ is mapped to the random input sequence of the channel $W_\mbf g$, denoted by $\bs{T}^n=(\bs{T}_1,\hdots,\bs{T}_n).$ The corresponding random output sequence is denoted by $\bs{Z}^n=(\bs{Z}_1,\hdots,\bs{Z}_n).$ The covariance matrix of each input $\bs{T}_i$ is denoted by $\mbf Q_i.$  We define $\mbf Q^\star$ such that
$\mbf Q^\star=\frac{1}{n}\sum_{i=1}^{n}\mbf Q_i.$
By using $\Gamma_n$ as a transmission-code for $\mc C,$ it follows that
\begin{align}
    \mbb P\left[ W\neq \hat{W}\right]\leq  \theta. \nonumber
\end{align}
We have
\begin{align}
    H(W)&= \log |\mathcal{W}| \nonumber \\
        &=    \log \lVert \Gamma_n \rVert \nonumber \\
        &\geq n(R-\delta).
        \label{entropyconverse2}
\end{align}
On the one hand, as shown in \eqref{applyfano}, we obtain by applying Fano's inequality
\begin{align}
   H(W)\leq \frac{1+I(W;\hat{W})}{1-\theta}. \label{entropyboundconverse2}
\end{align}
On the other hand, as shown in \eqref{mutinfconverse}, it holds that
\begin{align}
    \frac{1}{n} I(W;\hat{W})\leq \log\det(\mbf I_{N_{R}}+\frac{1}{\sigma^2}\mbf g \mbf Q^{\star} \mbf g^{H}  ).
    \label{mutinfconverse2}
\end{align} 
As a result, it follows from \eqref{entropyconverse2}, \eqref{entropyboundconverse2} and \eqref{mutinfconverse2} that for every $\mbf g \in \mc G_{a}$ 
\begin{align}
    n(R-\delta) \leq \frac{n\log\det(\mbf I_{N_{R}}+\frac{1}{\sigma^2}\mbf g \mbf Q^{\star} \mbf g^{H}  )+1}{1-\theta}.
\nonumber \end{align}
Hence, 
\begin{align}
    n(R-\delta) \leq \frac{n \underset{\mbf g \in \mc G_{a}}{\min}\log\det(\mbf I_{N_{R}}+\frac{1}{\sigma^2}\mbf g \mbf Q^{\star} \mbf g^{H}  )+1}{1-\theta}.
\nonumber \end{align}
Since $\mbf Q^\star \in \mc Q_{P}$ (see Lemma \ref{traceQ}), it follows that
\begin{align}
    n(R-\delta) \leq \frac{n \underset{\mbf Q \in \mc Q_{P}}{\max} \underset{\mbf g \in \mc G_{a}}{\min}\log\det(\mbf I_{N_{R}}+\frac{1}{\sigma^2}\mbf g \mbf Q \mbf g^{H}  )+1}{1-\theta}.
\nonumber \end{align}
This implies that for sufficiently large $n$ and for every $\delta,\theta>0,$ we have
\begin{align}
    R \leq \frac{ \underset{\mbf Q \in \mc Q_{P}}{\max} \underset{\mbf g \in \mc G_{a}}{\min}\log\det(\mbf I_{N_{R}}+\frac{1}{\sigma^2}\mbf g \mbf Q \mbf g^{H}  )+\frac{1}{n}}{1-\theta}+\delta.
    \nonumber
\end{align}
It follows that
\begin{align}
    R&\leq \underset{\delta,\theta>0}{\inf} \ \underset{n\rightarrow\infty}{\lim} \left[\frac{ \underset{\mbf Q \in \mc Q_{P}}{\max} \underset{\mbf g \in \mc G_{a}}{\min}\log\det(\mbf I_{N_{R}}+\frac{1}{\sigma^2}\mbf g \mbf Q \mbf g^{H}  )+\frac{1}{n}}{1-\theta}+\delta\right] \nonumber \\
     &= \underset{\mbf Q \in \mc Q_{P}}{\max} \underset{\mbf g \in \mc G_{a}}{\min}\log\det(\mbf I_{N_{R}}+\frac{1}{\sigma^2}\mbf g \mbf Q \mbf g^{H}). \nonumber
\end{align}
This completes the converse proof of Theorem \ref{capacitycompoundchannels}.

\section{Conclusion and Discussion}
\label{conclusion}
In this paper, we considered the problem of message transmission and the problem of CR generation over point-to-point MIMO slow fading channels. The first goal of this paper was to derive a lower and an upper bound on  the outage transmission capacity of single-user MIMO slow fading channels with average input power constraint, AWGN and with arbitrary state distribution under the assumption of CSIR and to show that our bounds coincide except possibly at the points of  discontinuity of the outage transmission capacity, of which there are, at most, countably many. The second goal was to establish a lower and an upper bound on the outage
CR capacity of a two-source model with unidirectional communication over the MIMO slow fading channel with AWGN
and with arbitrary state distribution using our bounds on the 
 outage transmission capacity of the MIMO slow fading channel.
The obtained results are particularly relevant in the problem of correlation-assisted identification over MIMO slow fading channels, where Alice and Bob have now access to a correlated source. This is an extension to the work done in \cite{deterministicfading}, where the focus is on deterministic identification over fading channels.
One can therefore introduce the concept of outage in the correlation-assisted identification framework and  proceed analogously to \cite{identificationcode}  to construct identification codes for MIMO slow fading channels based on the concatenation of two transmission codes using CR as a resource. This allows to derive a lower bound on the outage correlation-assisted identification capacity of MIMO slow fading channels in the log-log scale.
As a future work, it would be interesting to study the
problem of CR generation in fast fading environments, where
the channel state varies over the time scale of transmission.


\section*{Acknowledgments}
H. Boche was supported in part by the German Federal Ministry of Education and Research (BMBF) within the national initiative on 6G Communication Systems through the research hub 6G-life under Grant 16KISK002, within the national initiative on Post Shannon Communication
(NewCom) under Grant 16KIS1003K. He was further supported by the German Research Foundation (DFG) within Germany’s Excellence Strategy EXC-2092–390781972.
M. Wiese was supported by the Deutsche Forschungsgemeinschaft (DFG, German
Research Foundation) within the Gottfried Wilhelm Leibniz Prize under Grant BO 1734/20-1, and within Germany’s Excellence Strategy EXC-2111-390814868. C.\ Deppe was supported in part by the German Federal Ministry of Education and Research (BMBF) under Grant 16KIS1005 and in part by the German Federal Ministry of Education and Research (BMBF) within the national initiative on 6G Communication Systems through the research hub 6G-life under Grant 16KISK002.
R. Ezzine was supported by the German Federal Ministry of Education and Research (BMBF) under Grant 16KIS1003K.

This work has been presented in part at the virtual IEEE International Symposium on Information Theory (ISIT 2021) and in part at the virtual IEEE Information Theory Workshop (ITW 2021).



%

\end{document}